\documentclass[runningheads]{llncs}

\usepackage[utf8]{inputenc}
\usepackage{amsmath,amssymb} %
\usepackage{hyperref} %
\usepackage[capitalize]{cleveref}
\Crefname{figure}{Figure}{Figures}
\usepackage{tabularx}
\usepackage{booktabs} %
\usepackage{multirow} 
\usepackage{makecell} 
\usepackage{paralist}
\usepackage{wrapfig}
\usepackage{placeins}
\usepackage{nicefrac}
\usepackage{dirtytalk}
\usepackage[prefix=solarized-]{xcolor-solarized} %
\usepackage{wrapfig}
\usepackage{diagbox}
\usepackage{fontawesome5}

\usepackage{xifthen}

\usepackage{tikz} 
\usetikzlibrary{arrows,automata,backgrounds,positioning,intersections,shapes,decorations.pathreplacing,positioning, patterns,calc,patterns.meta}
\usepackage[nointegrals]{wasysym}
\usepackage{tkz-euclide}
\usepackage{subcaption}
\usepackage{graphicx}
\usepackage{xspace}
\usepackage{pgfplots}
\pgfplotsset{compat=newest}
\usepgfplotslibrary{fillbetween}

\usepackage{algorithm}
\usepackage{listings}
\usepackage{algpseudocode}

\usepackage{mathtools}
\usepackage[binary-units]{siunitx}
\usepackage{nicefrac}

\usepackage{semantic}

\usepackage{xcolor}
\definecolor{darkgreen}{rgb}{0.0, 0.42, 0.24}
\definecolor{solBlue}{RGB}{38 139 210}
\definecolor{amethyst}{rgb}{0.6, 0.4, 0.8}
\definecolor{antiquefuchsia}{rgb}{0.57, 0.36, 0.51}
\definecolor{darkorchid}{rgb}{0.6, 0.2, 0.8}
\definecolor{darkviolet}{rgb}{0.58, 0.0, 0.83}
\definecolor{deepmagenta}{rgb}{0.8, 0.0, 0.8}
\definecolor{deeplilac}{rgb}{0.6, 0.33, 0.73}
\definecolor{electricpurple}{rgb}{0.75, 0.0, 1.0}
\definecolor{electricviolet}{rgb}{0.56, 0.0, 1.0}
\definecolor{mediumorchid}{rgb}{0.73, 0.33, 0.83}
\definecolor{patriarch}{rgb}{0.5, 0.0, 0.5}
\definecolor{phlox}{rgb}{0.87, 0.0, 1.0}
\definecolor{purple(munsell)}{rgb}{0.62, 0.0, 0.77}
\definecolor{purple(x11)}{rgb}{0.63, 0.36, 0.94}
\definecolor{alertcolor}{named}{purple(munsell)}

\pgfplotsset{
tick label style={font=\footnotesize},
label style={font=\small},
legend style={font=\footnotesize},
}

\definecolor{background}{named}{white}
\definecolor{solBase01}{RGB}{88 110 117}
\definecolor{solBase1}{RGB}{147 161 161}
\definecolor{solMagenta}{RGB}{211  54 130}
\definecolor{color2}{named}{solMagenta}
\definecolor{solGreen}{RGB}{133 153   0}
\definecolor{color3}{named}{solGreen}
\definecolor{solYellow}{RGB}{181 137   0}
\definecolor{color4}{named}{solYellow}
\definecolor{dartmouthgreen}{rgb}{0.05, 0.5, 0.06}
\definecolor{color5}{named}{dartmouthgreen}
\definecolor{solBlue}{RGB}{38 139 210}
\definecolor{color6}{named}{solBlue}
\definecolor{davysgrey}{rgb}{0.33, 0.33, 0.33}
\definecolor{fancyblue}{rgb}{0.01, 0.28, 1.0}
\definecolor{amethyst}{rgb}{0.6, 0.4, 0.8}
\definecolor{green3}{rgb}{0.01, 0.75, 0.24}
\definecolor{carrotorange}{rgb}{0.93, 0.57, 0.13}
\definecolor{darkraspberry}{rgb}{0.53, 0.15, 0.34}
\definecolor{cornflowerblue}{rgb}{0.39, 0.58, 0.93}
\definecolor{darkpastelpurple}{rgb}{0.59, 0.44, 0.84}
\definecolor{niceblue}{RGB}{38 139 210}
\definecolor{niceyellow}{RGB}{181 137   0}

\definecolor{mainColor}{named}{solBase01}
\definecolor{boxcolor}{named}{color4}
\definecolor{cinitloc}{named}{mainColor}
\definecolor{cpath1light}{named}{color6}
\definecolor{cpath1dark}{named}{fancyblue}
\definecolor{cpath2light}{named}{color3}
\definecolor{cpath2dark}{named}{color5}
\definecolor{cinterseg}{named}{boxcolor}
\definecolor{goaltrace1}{named}{cpath1dark}
\definecolor{goaltrace2}{named}{cpath2dark}
\definecolor{refinementcolor1}{named}{goaltrace1}
\definecolor{refinementcolor2}{named}{goaltrace2}
\definecolor{refinementcolor}{named}{refinementcolor1}
\definecolor{crefseg}{named}{refinementcolor}

\definecolor{nonstochConflictcolor}{named}{black}
\definecolor{stochConflictcolor}{named}{black}
\definecolor{samplevaluescolor}{named}{solarized-magenta}

\definecolor{cloc0}{named}{solarized-base01}
\definecolor{cloc1}{named}{solarized-magenta}
\definecolor{cloc2}{named}{goaltrace1}
\definecolor{cloc3}{named}{goaltrace2}
\definecolor{cloc4}{named}{solarized-blue}

\tikzstyle{initloc}=[cinitloc]
\tikzstyle{path1start}=[cpath1light,densely dotted]
\tikzstyle{path1end}=[cpath1dark]
\tikzstyle{path2start}=[cpath2light,dotted]
\tikzstyle{path2end}=[cpath2dark]

\tikzstyle{filling}=[very thick, fill, fill opacity = 0.2]
\tikzstyle{cone}=[draw = none, fill, fill opacity = 0.2]
\tikzstyle{border}=[very thick, line join = round]
\tikzstyle{guard}=[fill opacity=0.1]
\tikzstyle{guardticks}=[color3]
\tikzstyle{init}=[border, mainColor, line join = round]
\tikzstyle{refinement}=[pattern=north west lines, pattern color = goaltrace1] %
\tikzstyle{refinement2}=[pattern=north east lines, pattern color = goaltrace2] %

\definecolor{goaltrace1opaque}{HTML}{7b9df9}

\tikzstyle{refinedsegment1}=[opacity=0.6,very thick, line join = round, fill, fill opacity = 0.2, goaltrace1]
\tikzstyle{refinedsegment2}=[opacity=0.6,very thick, line join = round, fill, fill opacity = 0.2, goaltrace2]
\tikzstyle{refinedsegment}=[refinedsegment1]
\tikzstyle{intermediatesegment}=[line join = round, fill, cinterseg]

\tikzstyle{other}=[opacity = 0.15]
\tikzstyle{tick}=[fill=background,font= \scriptsize]

\tikzstyle{firstlayerdistance}=[node distance =1.4cm and 0.95cm]
\tikzstyle{secondlayerdistance}=[node distance =1.4cm and 0.75cm]
\tikzstyle{secondlayercenterdistance}=[node distance =1.4cm and 0.2cm]
\tikzstyle{thirdlayerdistance}=[node distance =1.4cm and 0.05cm]

\tikzstyle{fazit-node}=[anchor=north west, text width = 1.8cm, draw, fill=white, fill opacity = 0.1, text opacity = 1,rounded corners=5pt]

\tikzstyle{stochConflict}=[bend right, ultra thick, densely dashed,stochConflictcolor]
\tikzstyle{nonstochConflict}=[bend right, ultra thick, dash pattern= on 1pt off 1pt,nonstochConflictcolor]
\tikzstyle{mixedConflict}=[bend right, ultra thick,nonstochConflictcolor, dash pattern= on 1pt off 1pt on 1pt off 5pt,
postaction={draw,ultra thick, stochConflictcolor,dash pattern= on 3pt off 5pt,dash phase=4pt}]

\tikzset{smallarrow/.style={decoration={markings,mark=at position 1 with %
    {\arrow[scale=0.8]{>}}},postaction={decorate}}}

\tikzset{ wavyarrow/.style={%
    ->,
    decorate,
    decoration={%
      snake,
      segment length=4mm,
      amplitude=0.3mm,
    }
  }}

\makeatletter
\newcounter{groupcount}
\pgfplotsset{
    draw group line/.style n args={5}{
        after end axis/.append code={
            \setcounter{groupcount}{0}
            \pgfplotstableforeachcolumnelement{#1}\of\datatable\as\cell{%
                \def\temp{#2}
                \ifx\temp\cell
                    \ifnum\thegroupcount=0
                        \stepcounter{groupcount}
                        \pgfplotstablegetelem{\pgfplotstablerow}{[index]0}\of\datatable
                        \coordinate [yshift=#4] (startgroup) at (axis cs:\pgfplotsretval,0);
                    \else
                        \stepcounter{groupcount}
                        \pgfplotstablegetelem{\pgfplotstablerow}{[index]0}\of\datatable
                        \coordinate [yshift=#4] (endgroup) at (axis cs:\pgfplotsretval,0);
                    \fi
                \else
                    \ifnum\thegroupcount>1
                        \setcounter{groupcount}{0}
                        \draw [
                            shorten >=-#5,
                            shorten <=-#5
                        ] (startgroup) -- node [anchor=north] {#3} (endgroup);
                    \fi
                \fi
            }
            \ifnum\thegroupcount>1
                        \setcounter{groupcount}{0}
                        \draw [
                            shorten >=-#5,
                            shorten <=-#5
                        ] (startgroup) -- node [anchor=north] {#3} (endgroup);
            \fi
            \ifnum\thegroupcount=1
                        \setcounter{groupcount}{0}
                        \draw [
                            shorten >=-#5,
                            shorten <=-#5
                        ] ($(startgroup)-(1.5mm,0)$) -- node [anchor=north] {#3} ($(startgroup)+(1.5mm,0)$);
            \fi
        }
    }
}
\makeatother

\pgfplotsset{
tick label style={font=\footnotesize},
label style={font=\small},
legend style={font=\footnotesize},
}

\pgfmathdeclarefunction{gauss}{3}{%
  \pgfmathparse{1/(#3*sqrt(2*pi))*exp(-((#1-#2)^2)/(2*#3^2))}%
}

\newcommand{\flowsepspace}{\vspace{0.2ex}}

\newcommand{\invariantsepspace}{\vspace{1.5ex}}
\newcommand{\invariantsepspacesmall}{\vspace{1ex}}

\tikzstyle{smallnode}=[draw, text width = 0.9cm, %
align = center, font= \footnotesize, rounded corners, very thick, execute at begin node=
]

\tikzstyle{smallslimnode}=[draw, text width = 0.6cm, %
align = center, font= \footnotesize, rounded corners, very thick, execute at begin node=
]

\newcommand\scalemath[2]{\scalebox{#1}{\mbox{\ensuremath{\displaystyle #2}}}}

\newcommand{\Reals}{\ensuremath{\mathbb{R}}\xspace}
\newcommand{\Realspos}{\ensuremath{\Reals_{>0}}\xspace}

\newcommand{\Realsposzero}{\ensuremath{\Reals_{\geq 0}}\xspace}

\newcommand{\Rationals}{\mathbb{Q}}
\newcommand{\Naturals}{\ensuremath{\mathbb{N}}\xspace}
\newcommand{\I}{\ensuremath{\mathbb{I}}\xspace}

\newcommand{\identity}{\ensuremath{\mathit{id}}\xspace}

\newcommand{\support}{\ensuremath{\mathit{supp}}\xspace}

\newcommand{\CDFs}{\ensuremath{\mathbb{F}}\xspace}
\newcommand{\CDF}{\ensuremath{\mathit{distr}}\xspace}

\newcommand{\guard}{\mathit{guard}}
\newcommand{\reset}{\mathit{reset}}

\newcommand{\RA}{\ensuremath{\mathcal{H}}\xspace}
\newcommand{\RAC}{\ensuremath{\mathcal{R}}\xspace}

\newcommand{\last}{\textit{last}}

\newcommand{\semanticsarrowsac}[1]{\ensuremath{\xrightarrow{#1}_{\scalemath{0.5}{\RAC}}}}

\newcommand{\tinyR}{{\scalemath{0.6}{R}}}
\newcommand{\tinyN}{{\scalemath{0.7}{N}}}
\newcommand{\tinyS}{{\scalemath{0.7}{S}}}

\newcommand{\dimCont}{\ensuremath{d}\xspace}
\newcommand{\dimRandom}{\ensuremath{{d_{\tinyR}}}\xspace}

\newcommand{\proj}[1]{_{#1}}

\newcommand{\Loc}{\ensuremath{\mathit{Loc}}\xspace} 
\newcommand{\VarCont}{\ensuremath{\mathit{Var}}\xspace}
\newcommand{\LabRandom}{\ensuremath{\mathit{Lab}}\xspace} 
\newcommand{\Edge}{\ensuremath{\mathit{Jump}\xspace}}

\newcommand{\Flow}{\ensuremath{\mathit{Flow}}\xspace}
\newcommand{\FlowCont}{\ensuremath{\Flow}\xspace}

\newcommand{\Distr}{\ensuremath{\mathit{Distr}}\xspace}
\newcommand{\Inv}{\ensuremath{\mathit{Inv}}\xspace}
\newcommand{\Init}{\ensuremath{\mathit{Init}}\xspace} 
\newcommand{\Event}{\ensuremath{\mathit{Event}}\xspace}

\newcommand{\runs}[1]{\ensuremath{\mathit{Runs}_{#1}}\xspace}
\newcommand{\Initruns}[1]{\ensuremath{\mathit{Runs}_{#1}^{\Init}}\xspace}

\newcommand{\run}{\ensuremath{\pi}} %

\newcommand{\scheduler}{\ensuremath{\mathfrak{s}}\xspace}

\newcommand{\SchedulersProphetic}[1]{\ensuremath{\Schedulers}_{#1}\xspace} 
\newcommand{\Schedulers}{\ensuremath{\mathfrak{S}}\xspace}

\newcommand{\SchedulersProphMin}[1]{\ensuremath{\Schedulers_{#1}^{\textsl{min}}(\Goal, \tmax,\jumpmax)}\xspace}

\newcommand{\clockr}{\ensuremath{r}\xspace}
\newcommand{\sample}{\ensuremath{s}\xspace}

\newcommand{\norc}{\ensuremath{\perp}\xspace}

\newcommand{\States}[1]{\ensuremath{\mathcal{S}_{#1}}\xspace} %
\newcommand{\valCont}{\ensuremath{\nu}\xspace} %
\newcommand{\valSet}{\ensuremath{\mathcal{V}}\xspace} %

\newcommand{\rate}{\mathit{rate}\xspace}
\newcommand{\res}{\mathit{res}\xspace}

\newcommand{\tmax}{\ensuremath{t_\textsl{max}}\xspace} %
\newcommand{\jumpmax}{\ensuremath{\textsl{jmp}}\xspace} %
\newcommand{\prophecy}{\ensuremath{\kappa}\xspace} %
\newcommand{\Prophecies}[1]{\ensuremath{\mathcal{K}_{#1}}\xspace} %
\newcommand{\PropheciesSchedulerMin}[1]{\ensuremath{\mathcal{K}_{#1}^{\scheduler!}}\xspace} %
\newcommand{\PropheciesSchedulerMax}[1]{\ensuremath{\mathcal{K}_{#1}^{\scheduler}}\xspace} %

\newcommand{\estat}{\ensuremath{e_\textsl{stat}}\xspace}

\newcommand{\comptime}{\ensuremath{t_\textsl{comp}}}

\newcommand{\source}{\textit{source}}
\newcommand{\target}{\textit{target}}

\newcommand{\reachtree}{\ensuremath{\mathtt{R}}\xspace}
\newcommand{\treePath}{\ensuremath{\Pi}\xspace}
\newcommand{\modtree}{\ensuremath{\mathtt{T}}\xspace}
\newcommand{\fullymodtree}{\ensuremath{\mathtt{T}^{*}}\xspace}

\newcommand{\rtmap}[1]{\ensuremath{\reachtree(#1)}\xspace}
\newcommand{\maxid}[1]{\ensuremath{\##1}\xspace}

\newcommand{\Goal}{\ensuremath{\mathcal{G}}\xspace}

\newcommand{\hastate}{\ensuremath{\sigma}\xspace} %

\newcommand{\JumpChoices}{\ensuremath{\mathit{JumpChoices}}\xspace}
\newcommand{\TimeChoices}{\ensuremath{\mathit{TimeChoices}}\xspace}

\newcommand{\ftc}[1]{\ensuremath{\mathit{T_{#1}}}\xspace}
\newcommand{\fosr}[1]{\ensuremath{\mathit{D_{#1}}}\xspace}

\newcommand{\hypro}{\textsc{HyPro}\xspace} 

\newcommand{\realyst}{\textsc{RealySt}\xspace}

\newcommand{\prohver}{\textsc{ProHVer}\xspace}

\newcommand{\scientific}[1]{\num[scientific-notation = true, exponent-product = \cdot]{#1}}
\newcommand{\samples}[1]{\scientific{#1}}
\newcommand{\errortable}[1]{\scriptsize\num[scientific-notation = true, exponent-product = \cdot, round-mode = places, round-precision = 1]{#1}}

\newcommand{\probabmin}[1]{\num[round-mode = places, round-precision = 3]{#1}}
\newcommand{\probabminuni}[1]{\num[round-mode = places, round-precision = 3]{#1}}
\newcommand{\probabminunishort}[1]{\num[round-mode = places, round-precision = 1]{#1}}
\newcommand{\probab}[1]{\num[round-mode = places, round-precision = 6]{#1}}

\newcommand{\rt}[1]{\SI[round-mode = places, round-precision = 2]{#1}{\second}}

\newcommand{\tikzboxgreenshape}{%
    \begin{tikzpicture}[scale=0.2, baseline=0mm, semithick,sharp corners]
    \draw[cone,cloc3] (1,1) -- (0.5,1) -- (0,0.25) -- (0,0) -- (1,0) -- cycle;
    \draw[border, cloc3,thick]  (1,1) -- (0.5,1) -- (0,0.25) -- (0,0) -- (1,0);
    \end{tikzpicture}%
}

\newcommand{\tikzboxpinkshape}{%
    \begin{tikzpicture}[scale=0.2, baseline=0mm, semithick,sharp corners]
    \draw[cone,cloc1] (0,0) -- (1,1) -- (1,0.5) -- cycle;
    \draw[border, cloc1,thick] (0,0) -- (1,1) -- (1,0.5) -- cycle;
    \end{tikzpicture}%
}
\newcommand{\tikzboxblueshape}{%
    \begin{tikzpicture}[scale=0.2, baseline=0mm, semithick,sharp corners]
    \draw[cone, cloc2] (1.425,0.125) -- (0.125,0.125) -- (0.25,0.25) -- (1.425,0.25) -- cycle;
\draw[border, cloc2,ultra thick]  (1.425,0.125) -- (0.125,0.125) -- (0.25,0.25) -- (1.425,0.25);
    \end{tikzpicture}%
}

\newcommand{\depth}{\textit{depth}}

\DeclareMathSymbol{\shortminus}{\mathbin}{AMSa}{"39}

\newcommand{\children}{\textit{ch}}
\newcommand{\parent}{\textit{pa}}
\newcommand{\validChildren}{\ensuremath{\textit{ch}^{+}}}

\newcommand{\car}{\ensuremath{\textit{CAR}}\xspace}

\begin{document}

\title{Minimum Reachability Probabilities in Rectangular Automata with Random Clocks\thanks{Supported by the DFG project 471367371.}}
\titlerunning{Minimum Reachability Probabilities in RAC}
\author{Joanna Delicaris\inst{1}\orcidID{0000-0001-9455-4052} \and
Erika Ábrahám\inst{2}\orcidID{0000-0002-5647-6134} \and
Anne Remke\inst{1}\orcidID{0000-0002-5912-4767}}
\authorrunning{J. Delicaris et al.}
\institute{University of Münster, Münster, Germany \\
\email{joanna.delicaris,anne.remke@uni-muenster.de}\\
\and
RWTH Aachen University, Aachen, Germany\\
\email{abraham@cs.rwth-aachen.de}}

\maketitle

\begin{abstract}
Control applications for cyber-physical systems must make reliably safe control decisions in the presence of continuous dynamics as well as stochastic uncertainty. Providing safety guarantees for such systems requires formal modeling and analysis techniques that capture these aspects.
For modeling, in this paper we consider rectangular automata with random clocks under prophetic scheduling. For this model class, existing methods can compute only upper bounds on reachability probabilities, enabling optimistic, best-case safety reasoning.
We complement this view by introducing a novel method to compute lower bounds, thereby enabling worst-case analysis that is essential for safety-critical applications. Although both upper and lower bounds rely on reachability analysis, they are not dual: computing lower bounds requires an explicit separation of stochastic and nondeterministic choices along executions.
We implement our approach and demonstrate its practical feasibility on an electric vehicle charging scenario, showing that meaningful  worst-case guarantees can be obtained.
\end{abstract}

\section{Introduction}
\label{sec:introduction}
Rectangular automata with random clocks (RAC) are a powerful formalism for modeling stochastic hybrid systems that arise in safety-critical applications such as energy and mobility systems. They combine continuous dynamics, random timing, and nondeterministic behavior, enabling the joint modeling of aleatoric uncertainty (e.g., stochastic arrivals) and epistemic uncertainty (e.g., underspecified system behavior). For example, in electric vehicle charging systems~\cite{journal,tase,25StuebbeSYNASC}, randomness can capture arrival patterns, while nondeterminism reflects variability in charging strategies and system conditions.
The execution of a RAC starts from a nondeterministic initial state, which evolves continuously with rates chosen from rectangular sets, and allow discrete transitions at either random or nondeterministically chosen time points.
Existing analysis techniques for RAC~\cite{journal,tase} use flowpipe construction~\cite{alur1995AlgorithmicAnalysisHybrid} to compute all behaviors that \emph{may} reach a goal state. Encoding stochastic choices in the state space allows to integrate the probabilities of those goal-reaching paths, providing maximum reachability probabilities under prophetic scheduling, i.e., best-case guarantees for reaching the goal.
However, many applications require worst-case guarantees, captured by minimum reachability probabilities. For instance, an e-car should complete its trip without running out of battery—even with detours—with probability at least 0.95, regardless of how nondeterministic choices  are resolved.

While, in principle, minimum reachability equals 
1 minus the maximum probability of not reaching the goal, this observation is not practically useful, since the latter cannot be computed effectively. As a result, minimum reachability is not simply the dual of maximum reachability and requires a dedicated analysis approach. Hence, we propose a new method to compute minimum reachability probabilities for RAC, based on a reach-tree representation that explicitly separates stochastic and nondeterministic choices. In contrast to maximum reachability—which collects behaviors that may reach the goal via unions—our method identifies behaviors that guarantee reaching the goal by taking intersections over nondeterministic choices.
The main technical challenge is the interaction between stochastic and nondeterministic behavior. To address this, we introduce a classification of reach-tree branchings into stochastic, nondeterministic, and mixed. Stochastic branchings are handled via union, nondeterministic ones via intersection, and mixed cases are systematically decomposed to separate both effects. This separation is key to obtaining sound lower bounds.
As illustrated in Figure~\ref{fig:process-wenigerdetail}, our workflow follows the same outer structure as maximum reachability (flowpipe construction and integration), but differs in how reachable behaviors are combined: we compute the set of executions that enforce reaching the goal.

Our previous work either (i) focuses on maximum reachability, where all branchings can be handled uniformly via unions~\cite{journal,tase,25StuebbeSYNASC}, or (ii) restricts the model by excluding continuous nondeterminism, which simplifies minimum analysis~\cite{carina}. Our approach handles the full RAC model, including both continuous and discrete nondeterminism.
We implement our method in the tool \realyst and evaluate it on an electric vehicle charging case study, demonstrating that practical and meaningful worst-case guarantees can be computed.

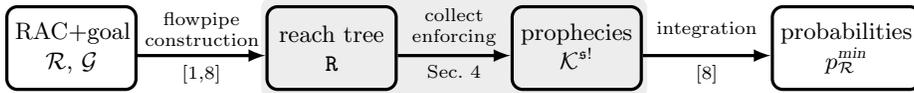
\begin{figure}[t]%
\centering
    \begin{tikzpicture}[
baseline,remember picture,
node distance = 0.3cm and 1cm,
box/.style={text width = 2cm, minimum height = 1.1cm, align = center, font= \footnotesize, very thick, anchor=north},
boxaut/.style={draw, text width = 2cm, minimum height = 2.5cm, align = center, font= \footnotesize, very thick, anchor=north, rounded corners},
boxlight/.style={draw, black!10!white, opacity=0.7, fill},
emptynode/.style={draw=none, inner sep=0},%
boxadd/.style={draw, text width = 2cm, minimum height = 1.4cm, align = center, font= \footnotesize, very thick, anchor=north},
section/.style={draw, fill=white, circle, inner sep=0.75mm, font= \footnotesize, thick},
e/.style={-latex, very thick},%
eno/.style={draw=none},%
l/.style={draw=none, minimum height=0cm, font = \scriptsize, align = center, text width = 1.5cm},%
textstyle/.style={font = \footnotesize},%
]

\pgfdeclarelayer{background}
\pgfsetlayers{background,main}

	\useasboundingbox (-0.85,-1.18) rectangle (11.3,0.08);

\node[box, draw, text width = 1.5cm, rounded corners] (sac) at (0,0) {RAC+goal \\$\RAC,\,\Goal$}; %
\node[box, draw, text width = 1.5cm, rounded corners, right= 1.7cm of sac] (rt) {reach tree \\ $\reachtree$};

\node[box, draw, text width = 1.5cm, rounded corners, right= 1.5cm of rt] (proph) {prophecies \\ $\PropheciesSchedulerMin{}$};

\node[box, draw, text width = 1.75cm, rounded corners, right= 1.7cm of proph] (prob) {probabilities \\ $p_{\RAC}^\textsl{min}$};

\begin{pgfonlayer}{background}
\draw[boxlight, rounded corners] ($(rt.south west)+(-0.05,-0.1)$) rectangle ($(proph.north east)+(0.05,0.1)$);
\end{pgfonlayer}

\draw[e] ([yshift=-0.1cm]sac.east) -- 
node[l,above] (flowpipe) {flowpipe\\ construction} 
node[l,below] {\cite{alur1995AlgorithmicAnalysisHybrid,journal}} 
([yshift=-0.1cm]rt.west);

\draw[e] ([yshift=-0.1cm]rt.east) --
node[l,above, minimum height=0.65cm] (approach) {collect\\enforcing}
node[l,below] {Sec.~\ref{sec:computation}}
([yshift=-0.1cm]proph.west);

\draw[e] ([yshift=-0.1cm]proph.east)  --
node[l,above, minimum height=0.65cm] (integrate) {integration}
node[l,below] {\cite{journal}} 
([yshift=-0.1cm]prob.west);

\end{tikzpicture} %
    \caption{
    Computation of minimum reachability probabilities (within the gray box).
    }
    \label{fig:process-wenigerdetail}
\end{figure}

\paragraph{Related Work}

To the best of our knowledge, besides our own previous work \cite{journal,tase,25StuebbeSYNASC}, the only approach to compute upper bounds on \emph{maximum} reachability probabilities for similar  models  is the tool \prohver~\cite{hahn2013CompositionalModellingAnalysis}, where \emph{discrete} nondeterminism  is resolved prophetically and  \emph{continuous} nondeterminism via a safe overapproximation. However, the computation of minimum reachability probabilities is not supported.
The only other approach that computes both maximum and minimum reachability probabilities is \cite{Pilch2021Optimizing,carina}, specifically tailored to urgent singular automata and therefore not supporting time or rate nondeterminism. 

Several approaches extend hybrid automata (HA) with probabilistic and stochastic elements, exhibiting different nondeterministic components.
For example, decidable subclasses of HA have been extended with discrete probability distributions on jumps \cite{sproston2000DecidableModelChecking,sproston2019Verification}, which preserves decidability of reachability.
Timed automata (TA) allow for various extensions as their continuous behavior is rather limited. Extensions of TA include    continuously distributed resets~\cite{kwiatkowska2000VerifyingQuantitativeProperties},   continuous probability distributions~\cite{bohnenkamp2006MODESTCompositionalModeling}, stochastic delays and jumps~\cite{bertrand2014StochasticTimedAutomata},  and (networks of) stochastic TA \cite{ballarini2013TransientAnalysisNetworks}, all resolving nondeterminism probabilistically.   
In contrast, maximum and minimum reachability probabilities are computed for TA with discrete distributions  in \cite{KWIATKOWSKA20071027}.
Approaches for more general models  apply stochastic approximation ~\cite{prandini2006StochasticApproximationMethod} or  a combination of  discretization and randomness~\cite{koutsoukos2008ComputationalMethodsVerification}. 
Abstractions for uncountable-state discrete-time stochastic processes either resolve all nondeterminism probabilistically~\cite{soudjani2015FAUST2FormalAbstractions}, or abstract to interval Markov decision processes~\cite{cauchi2019StocHyAutomatedVerification}.
The more complex setting of stochastic timed games \cite{Bouyer2009} makes quantitative reachability  undecidable for  1.5 players and $4$ clocks.

\paragraph{Outline}
We introduce the class of RAC and define  the reach tree, schedulers and nondeterminism in Section~\ref{sec:modelclass}, as in~\cite{journal}. 
Section~\ref{sec:classification} presents the classification of branchings, which allows the 
computation of minimum reachability probabilities as explained in Section~\ref{sec:computation}. 
We provide results for several exemplary models and an existing case study in Section~\ref{sec:casestudy}, and conclude the paper in Section~\ref{sec:conclusion}.

\section{Preliminaries}
\label{sec:modelclass}

We present \emph{rectangular automata with random clocks (RAC)} in Section~\ref{subsec:SARC}, as in \cite{journal}.
Section~\ref{subsec:nondeterminism}, defines schedulers for RAC to resolve nondeterminism, and Section~\ref{subsec:reach} introduces the reach tree as a result of reachability analysis.

Let $\Naturals$, $\Rationals$ and $\Reals$ be the set of all natural (including $0$), rational resp. real numbers; we use lower indices to restrict these sets, e.g. $\Realspos$ for the positive reals.
\emph{Rectangular sets} refer to cross products of real-valued intervals from $\mathbb{I}=\{[a,b], [a,\infty), (-\infty,b] \mid a,b \in \Rationals\}\cup \{(-\infty,\infty)\}$.
Given an ordered set of $d$ variables containing $x$ in the $i$th position, for any domain $D$ and $q = (q_1, \dots, q_d) \in D^d$, we use $q \proj{x}$ to refer to $q_i$. 
For $\CDF:\Realsposzero\rightarrow\Realsposzero$
let $\support(\CDF)=\{v\in\Realsposzero\,|\,\CDF(v)>0\}$.
We call $\CDF$ a \emph{(continuous) distribution} if it is absolutely continuous with $\int_{0}^{\infty}\CDF(v)\,dv=1$. Let $\CDFs$ denote the set of all distributions.

\subsection{Rectangular Automata with Random Clocks}\label{subsec:SARC}

Rectangular automata with random clocks (RAC) provide an expressive modeling language for systems that combine discrete  components and continuous dynamics with  random or nondeterministic timing. In such models, continuous variables evolve  with rates from bounded (rectangular) intervals, while discrete transitions ocur either after random delays or at nondeterministically chosen time points. When neglecting random aspects, the mixed discrete-continuous behavior is modeled by a rectangular automata (RA), as defined in \cite{alur1995AlgorithmicAnalysisHybrid,henzinger1998WhatDecidableHybrid,journal}.
\begin{definition}[RA Syntax]
A \emph{rectangular automaton} (RA) is a tuple $\RA = (\Loc,\allowbreak \VarCont,\allowbreak\Inv,\Init,\FlowCont,\Edge)$ with
\begin{itemize}
    \item a finite nonempty set $\Loc$ of \emph{locations},
    \item a finite nonempty ordered set $\VarCont$ of $|\VarCont| =d$ \emph{(continuous) variables}, %
     \item  functions $\Inv$ and $\Init$ of type $\Loc \rightarrow \I^{\dimCont}$, assigning an \emph{invariant} resp. \emph{initial states}  to each location, satisfying (i) $\Init(\ell)\subseteq\allowbreak\Inv(\ell)$ for all $\ell\in\Loc$ and (ii) $\Init(\ell)\not=\emptyset$ for some $\ell\in\Loc$, 
     \item  a function $\FlowCont:\Loc \rightarrow \mathbb{I}^{\dimCont}$, assigning \emph{(flow) rates} to the locations, with $\FlowCont(\ell) \not = \emptyset$ for all $\ell \in \Loc$,
     and
    \item
    a finite set  $\Edge \subseteq  \Loc \times \I^{\dimCont} \times (\mathbb{I}\cup \{\identity\})^{\dimCont} \times   \Loc$ of \emph{jumps} $e=(\ell,\allowbreak\guard,\allowbreak\reset,\allowbreak \ell')\in\Edge$ with 
    \emph{source location} $\source(e)=\ell$, \emph{target location} $\target(e)=\ell'$, 
    \emph{guard} $\guard(e)=\guard$, and 
    \emph{reset} $\reset(e)=\reset$, such that $\Inv(\ell)\cap\guard\not=\emptyset$ and for all $x\in\VarCont$, 
    (i) if $\reset\proj{x}=\identity$ then $\Inv(\ell)\proj{x}\cap\guard\proj{x}\subseteq\Inv(\ell')\proj{x}$, and otherwise, 
    (ii) if $\reset\proj{x}\not=\identity$, then $\reset\proj{x}\subseteq\Inv(\ell')\proj{x}$.
\end{itemize}
\end{definition}

\noindent We call $\RA$ \emph{nonblocking}~\cite{journal} if for each $\ell\in\Loc$ and $\nu\in\Inv(\ell)$, either there exists a jump $e\in\Edge$ with $\source(e)=\ell$ and $\nu\in\guard(e)$, 
or 
there exist $\delta\in\Realspos$, $\rate\in \Flow(\ell)$  with $\valCont+\delta\cdot \rate\in\Inv(\ell)$.
As in \cite{journal},  we extend RA with \emph{stochastic jumps} to model random events.

\begin{definition}[RAC Syntax]
    \label{def:RAC}
    A \emph{Rectangular Automaton with random Clocks (RAC)} is a tuple $\RAC = (\RA,\LabRandom,\allowbreak\Distr,\Event)$ with
    \begin{itemize}
        \item a RA $\RA=(\Loc, \VarCont,\Inv,\Init,\FlowCont,\Edge)$ with $\vert\VarCont\vert=d$,
        \item a finite ordered set $\LabRandom\subseteq \VarCont$ of $ | \LabRandom | = \dimRandom$ {\emph{random events} (also called \emph{random clocks}),}
         \item  a function $\Distr:\LabRandom\rightarrow\CDFs$,
         \item a function $\Event: \Edge \rightarrow (\LabRandom \cup \{\norc\})$, that declares a jump $e$ to be \emph{stochastic} if $\Event(e)\in\LabRandom$ and \emph{nonstochastic} if $\Event(e)=\,\norc$; we require stochastic jumps to be nonguarded, i.e, $\Event(e)\in\LabRandom$ implies $\guard(e) = \Reals^{d}$,
        \item for each random event $\clockr\in \LabRandom$  we require
        \begin{enumerate}[(i)]
            \item for each location $\ell \in \Loc$: 
        (a) $\Init(\ell)\proj{r}=[0,0]$,
        (b) $\Inv(\ell)\proj{r}=\Reals$, and
        (c) $\Flow(\ell)\proj{\clockr} = [1,1]$ if there is a stochastic jump $e \in \Edge$ with $\source(e)=\ell$ and $\Event(e)=\clockr$, 
        and $\Flow(\ell)\proj{\clockr} = [0,0]$ otherwise; and
        \item for each jump
        $e \in \Edge$:
        $\guard(e)\proj{r}=\Reals$  and $\reset(e)\proj{r}=id$,
        \end{enumerate}
        \item and removing all stochastic jumps from $\Edge$ of $\RA$ yields a nonblocking RA. 
    \end{itemize}
\end{definition}%

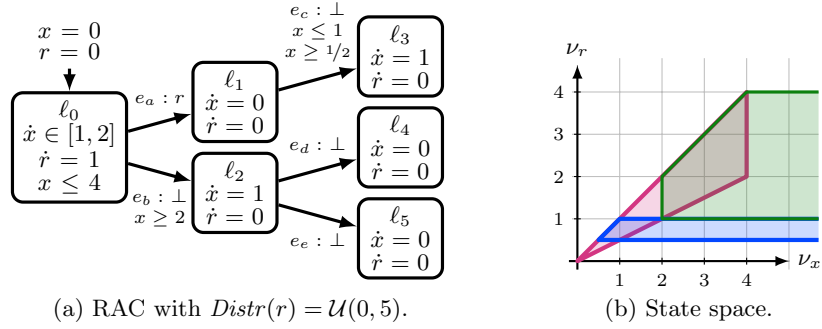
\begin{figure}[tb]
    \centering
    \begin{subfigure}[b]{.64\linewidth}
    \centering
        \begin{tikzpicture}[
scale=1.1,
baseline,remember picture,
node distance = -0.5cm and 1.4cm,
n/.style={draw, text width = 1.3cm, %
align = center, font= \footnotesize, rounded corners, very thick, execute at begin node=\setlength{\baselineskip}{8pt}%
},
en/.style={draw=none, minimum height=0cm, font = \scriptsize, align = center},%
c/.style={draw, fill, black, circle, inner sep=0, outer sep=0, minimum size=1mm},
l/.style={anchor=west, inner sep=0, font=\footnotesize},
init/.style={en, 
text width=1.2cm,
anchor=south west,  
inner xsep=0pt},
]

\useasboundingbox (-0.75,-1.6) rectangle (4.55,1.6);

\node[n](l0) at (0,0) {
    $\ell_0$\\\flowsepspace 
    $\dot{x}\in[1,2]$\\\flowsepspace  
    $\dot{r}=1$\\\flowsepspace  
    $x\leq 4$
};

\node[smallnode](l1) at (2,.55) {
    $\ell_1$\\\flowsepspace 
    $\dot{x}=0$\\\flowsepspace  
    $\dot{r}=0$
};

\node[smallnode](l2) at (2,-0.55) {
    $\ell_2$\\\flowsepspace 
    $\dot{x}=1$\\\flowsepspace  
    $\dot{r}=0$
};

\node[smallnode](l3) at (4,1.1) {
    $\ell_3$\\\flowsepspace 
    $\dot{x}=1$\\\flowsepspace  
    $\dot{r}=0$
};

\node[smallnode](l4) at (4,0) {
    $\ell_4$\\\flowsepspace 
    $\dot{x}=0$\\\flowsepspace  
    $\dot{r}=0$
};

\node[smallnode](l5) at (4,-1.1) {
    $\ell_5$\\\flowsepspace 
    $\dot{x}=0$\\\flowsepspace  
    $\dot{r}=0$
};

\node[n,above=0.3cm of l0, draw=none] (init) {$x=0$\\$r=0$};

\draw[-latex, very thick] (init) to node[en, above] {} (l0);
\draw[-latex, very thick] (l0) to node[en, above,yshift=1mm] {$e_a:r$} (l1);
\draw[-latex, very thick] (l0) to node[en, below, yshift=-1mm] {$e_b:\ \norc$\\$x \geq 2$} (l2);

\draw[-latex, very thick] (l1) to node[en, above,yshift=1mm] {$e_c:\ \norc$
\\$x \leq 1$\\$x \geq \nicefrac{1}{2}$} (l3);
\draw[-latex, very thick] (l2) to node[en, above,yshift=1mm] {$e_d:\ \norc$} (l4);
\draw[-latex, very thick] (l2) to node[en, below,yshift=-1mm] {$e_e:\ \norc$} (l5);

\end{tikzpicture}
        \caption{RAC with $\Distr(r)=\mathcal{U}(0,5)$.}\label{fig:rac-example-automaton}
    \end{subfigure}
    \hfil
    \begin{subfigure}[b]{.32\linewidth}
    \centering
        \begin{tikzpicture}[
baseline, remember picture,
font=\footnotesize, scale=2.25, 
tick/.style={fill=white,font= \scriptsize}]

\useasboundingbox (-0.17,-0.11) rectangle (1.44,1.25);

\node (ursprung) at (0,0) {};
\draw[help lines, lightgray, xstep=0.25, ystep=0.25] ($(ursprung)-(0.05,0.05)$) grid (1.425, 1.2);
\node[fill=white] (y) at (0,1.25) {$\valCont_r$};
\node[fill=white] (x) at (1.375,0) {$\valCont_x$};
\draw[-latex,thick] ($(ursprung)-(0,0.05)$) to (y);
\draw[-latex,thick] ($(ursprung)-(0.05,0)$) to (x);

\node[tick] (1) at (0.25,-0.11) {$1$};
\node[tick] (2) at (0.5,-0.11) {$2$};
\node[tick] (3) at (0.75,-0.11) {$3$};
\node[tick] (4) at (1,-0.11) {$4$};
\draw ($(0.25,0)+(0,0.02)$) -- ($(0.25,0)-(0,0.02)$);
\draw ($(0.5,0)+(0,0.02)$) -- ($(0.5,0)-(0,0.02)$);
\draw ($(0.75,0)+(0,0.02)$) -- ($(0.75,0)-(0,0.02)$);
\draw ($(1,0)+(0,0.02)$) -- ($(1,0)-(0,0.02)$);

\node[tick] (1) at (-0.11,0.25) {$1$};
\node[tick] (2) at (-0.11,0.5) {$2$};
\node[tick] (3) at (-0.11,0.75) {$3$};
\node[tick] (4) at (-0.11,1) {$4$};
\draw ($(0,0.25)+(0.02,0)$) -- ($(0,0.25)-(0.02,0)$);
\draw ($(0,0.5)+(0.02,0)$) -- ($(0,0.5)-(0.02,0)$);
\draw ($(0,0.75)+(0.02,0)$) -- ($(0,0.75)-(0.02,0)$);
\draw ($(0,1)+(0.02,0)$) -- ($(0,1)-(0.02,0)$);

\draw[cone, cloc1] (0,0) -- (1,1) -- (1,0.5) -- cycle;
\draw[border, cloc1,ultra thick] (0,0) -- (1,1) -- (1,0.5) -- cycle;

\draw[cone, cloc2] (1.425,0.125) -- (0.125,0.125) -- (0.25,0.25) -- (1.425,0.25) -- cycle;
\draw[border, cloc2,ultra thick]  (1.425,0.125) -- (0.125,0.125) -- (0.25,0.25) -- (1.425,0.25);

\draw[cone, cloc3] (1.425,1) -- (1,1) -- (0.5,0.5) --  (0.5,0.25) -- (1.425,0.25) -- cycle;
\draw[border,cloc3] (1.425,1) -- (1,1) -- (0.5,0.5) --  (0.5,0.25) -- (1.425,0.25);

\end{tikzpicture}
        \caption{State space.}\label{fig:rac-example-flowpipe}
    \end{subfigure}
    \caption{Running Example: RAC model with reachable state space in 
    $\ell_0$ and $\ell_1$ (pink), $\ell_2$, $\ell_4$ and $\ell_5$ (green), and $\ell_3$ (blue).}
    \label{fig:rac-example}
\end{figure}

\noindent Figure~\ref{fig:rac-example-automaton} illustrates a small RAC with six locations and one random event, which serves as running example. Figure~\ref{fig:rac-example-flowpipe} shows the corresponding reachable state space. 
For the rest of this paper, assume a RAC $\RAC = (\RA,\LabRandom,\Distr,\allowbreak\Event)$ with $\RA=(\Loc, \allowbreak\VarCont, \allowbreak\Inv,  \allowbreak\Init,\allowbreak\FlowCont,\Edge)$, $ | \VarCont | = \dimCont$ and  $ | \LabRandom | = \dimRandom$.
Furthermore, to measure the global time for time-bounded reachability computations, we assume a distinguished variable $t\in\VarCont\setminus\LabRandom$ with $\Init(\ell)\proj{t}=[0,0]$, $\Flow(\ell)\proj{t}=[1,1]$, $\Inv(\ell)\proj{t}=\Reals$, $\guard(e)\proj{t}=\Reals$ and $\reset(e)\proj{t}=\identity$ for all $\ell\in\Loc$ and $e\in\Edge$.
To simplify RAC illustrations, we may %
omit (i) the global time $t$, (ii) initial state and flow of random clocks $r\in\LabRandom$, (iii) all guards and invariants that pose no restrictions (e.g., if $\Inv(\ell)\proj{x}=\Reals$ and $\Inv(\ell)\proj{y}=(-\infty,2]$, then we only write $y\leq2$) and (iv) resets to identity, %
and write  $c$ for point intervals $[c,c]$. %

\paragraph{States} The labels of stochastic jumps specify their stochastic behavior. 
A random event $\clockr$ \emph{occurs} when a stochastic jump labeled with $\clockr$ is taken.
Initially, for each random event $\clockr$ an \emph{expiration time}  $\sample\proj{\clockr}$ is sampled from its assigned distribution $\Distr(\clockr)$. 
The event $r$ can occur if jumps with label $r$ have been enabled for a total of $\sample\proj{\clockr}$ time units.
To measure these \emph{durations of enabledness}, we define random events $\clockr$ to be included in the set of continuous variables, and model them as stopwatches, initialized to $0$ and running when a jump with label $\clockr$ is enabled;
therefore, we call random events $\clockr$ also \emph{random clocks}.
A \emph{state} $\hastate=(\ell,\valCont,\sample)\in\States{\RAC}=\Loc\times\Reals^{\dimCont}\times\Realsposzero^\dimRandom$ %
of $\RAC$ specifies the current location $\ell$, the valuation $\valCont$ for the continuous variables (which includes   the durations of enabledness $\valCont\proj{r}$ for each $r\in\LabRandom$), and the expiration times $\smash{\sample=(\sample_1,\dots,\sample_{\dimRandom})}$ of random events.

\paragraph{Semantics} The  operational semantics for a RAC $\RAC$ defined in Figure~\ref{fig:operationalsemanticsRAC} specifies the evolution of the state of $\RAC$, by taking a nonstochastic jump (Rule \texttt{Jump}$_{\tinyN}$) or a stochastic jump (Rule \texttt{Jump}$_{\tinyS}$), or by letting time elapse (Rule \texttt{Flow}). %
Recall that stochastic jumps are nonguarded by Definition~\ref{def:RAC}.
Note that the reset $res \in \reset$ for a jump or rate $\rate \in \FlowCont(\ell)$ for a time step can be uniquely determined from $\valCont$ and $\valCont'$; where useful, we annotate the step with this additional information, yielding $\hastate\semanticsarrowsac{e,res}\hastate'$ resp. $\hastate\semanticsarrowsac{\delta,\rate}\hastate'$.

\begin{figure}[tb]
    \centering
\input{img/rac}
\vspace*{-4ex}
    \caption{Operational semantics for a RAC $\RAC$.}
    \label{fig:operationalsemanticsRAC}
\end{figure}

\begin{definition}
A (finite) \emph{run} of $\RAC$ is a finite sequence $\run=\hastate_0\xrightarrow{a_1}
\ldots\xrightarrow{a_n}\hastate_n$, such that 
(i) $\hastate_i=(\ell_i,\valCont_i,\sample_i)\in\States{\RAC}$ for all $i\in\{0,\ldots,n\}$,
(ii) $\valCont_0\in\Inv(\ell_0)$, 
(iii) $a_{i}\in \Reals_{\geq0}\cup\Edge$ and $\hastate_{i\shortminus 1}\semanticsarrowsac{a_{i}}\hastate_{i}$ for all $i\in\{1,\ldots,n\}$
and (iv) 
if $a_i=\delta_i\in \mathbb{R}_{\geq0}$ is a time step and $i<n$, then we require $a_{i+1}\in\Edge$ to be a jump.
\end{definition}
We  call $ | \run | = n$ the length of $\run$, define $\last(\run)=\hastate_n$ and let $\runs{\RAC}$ be the set of all runs of $\RAC$. %
We call $\run$ \emph{initial} if $\valCont_0\in\Init(\ell_0)$,
$(\valCont_0)\proj{\clockr} = 0$, and $(\sample_0)\proj{\clockr} \in \support(\Distr(\clockr))$ for all $\clockr \in \LabRandom$.
Let $\Initruns{\RAC}$ be the set of all initial runs of $\RAC$.
A state $\hastate=(\ell,\valCont,\sample)$ is \emph{reachable in $\RAC$} if there exists an initial run 
$\run$ with $\last(\run)=\hastate$.

 The above semantics is dedicated for RAC where each random event occurs at most once per run, implying no loops with stochastic jumps.
 Otherwise, the operational semantics must be extended as in~\cite{journal}.
Since we compute time- and jump-bounded reachability, this is not restrictive: any RAC can be \emph{unrolled} to meet the jump bound, as formalized in \cite{journal}.

\subsection{Nondeterminism and Prophetic 
Schedulers}\label{subsec:nondeterminism}

A RAC $\RAC$ is nonde\-ter\-ministic in 
(i) the choice of initial location and valuation from $\mathit{InitialChoices}_{\RAC}=\{(\ell,\valCont) \in \Loc\times\Reals^{\dimCont} \mid \valCont\in\Init(\ell)
    \}$,
(ii) the choice of duration and flow rate from
     $\TimeChoices_{\RAC}(\run)\allowbreak = 
     \{(\delta,\rate)\in \Realsposzero\times\Reals^{\dimCont}  \mid\exists \hastate'\in\States{\RAC}.\ 
     \last(\run)\semanticsarrowsac{\delta,\rate}\hastate'
     \text{ and }
     \run\xrightarrow{\delta} \hastate'\in\runs{\RAC}
     \}$
     of a time step extending a run $\run$ of $\RAC$, 
 and (iii)  the choice of jump and reset from
     $\JumpChoices_{\RAC}(\hastate) = \{(e,\res)\in\Edge\times (\Reals \cup \{id\})^\dimCont \mid \exists \hastate'\in\States{\RAC}.\ \hastate \semanticsarrowsac{e,\res} \hastate'\}$  from state $\hastate$.
    The definition of $\TimeChoices_{\RAC}$
requires that $\delta$ and $\rate$ together define a valid time step according to the semantics and the run definition.

All nondeterminism can be resolved by a scheduler, which turns a RAC into a purely stochastic model with well-defined reachability probabilities.
We use  \emph{prophetic schedulers} \cite{Pilch2021Optimizing}, which have full information on the past history (i.e., the run leading to the current state) including \emph{all}
samples of expiration times, making future
random events predictable. This latter information is contained in a \emph{prophecy} 
$\prophecy: \LabRandom\rightarrow\Reals_{\geq 0}$, specifying $\prophecy(r)$ to be the expiration time of random event $r\in\LabRandom$. 
Let $\Prophecies{\RAC}$ be the set of all prophecies for $\RAC$.
While prophetic schedulers can predict the timing of future random events—an unrealistic assumption in most settings—they are well-suited to perform a \emph{worst-case} analysis as their foresight lets them resolve choices in the most adverse way, especially when uncontrollable uncertainties are modeled nondeterministically.

\begin{definition}
 \label{def:Scheduler}
 \sloppy A \emph{(prophetic history-dependent) scheduler $\scheduler$ for $\RAC$} defines for each prophecy $\prophecy\in\Prophecies{\RAC}$ a function  
 \[\scheduler_{\prophecy} : \Initruns{\RAC}\cup\{\epsilon\} \rightarrow  (\Loc\times \Reals^{\dimCont}) \cup  (\Realsposzero\times \Reals^{\dimCont}) \cup  (\Edge\times (\Reals \cup \{id\})^\dimCont),\] 
 s.t. $\smash{\scheduler_{\prophecy}(\epsilon)\in\mathit{InitialChoices}_{\RAC}}$ and 
 $\scheduler_{\prophecy}(\run) \in \JumpChoices_{\RAC}(\last(\run)) \cup \TimeChoices_{\RAC}(\run)$ for 
  $\smash{\run\in \Initruns{\RAC}}$.
Let $\SchedulersProphetic{\RAC}$ denote all schedulers for $\RAC$.
 \end{definition} 
\noindent 
The scheduler and its function under the current prophecy together uniquely determine the execution.

\begin{definition}
For a RAC $\RAC$, a scheduler $\scheduler$ for $\RAC$, a prophecy $\prophecy$, and a fixed $n\in\Naturals$, \emph{the run of length $n$ induced by $\scheduler_{\prophecy}$ in $\RAC$} is the unique initial run $\run(\RAC,\scheduler_{\prophecy},n)\in\Initruns{\RAC}$ of length $n$ such that:

\begin{enumerate}[(i)]

\item if $n=0$, then $\run(\RAC,\scheduler_{\prophecy},0)=(\ell,\valCont,\sample)$ with $(\ell,\valCont)=\scheduler_{\prophecy}(\epsilon)$, $\valCont\proj{r}=0$ and $\sample\proj{r}=\prophecy(r)$ for all $r\in\LabRandom$, and

\item if $n\!>\!0$, then 
$\run(\RAC,\scheduler_{\prophecy},n)\!=\!\run_{n\shortminus 1}\!\xrightarrow{a}\!\hastate'$ with $\smash{\run_{n \shortminus 1}\!=\!\allowbreak\run(\RAC,\scheduler_{\prophecy}, n{-}1)}$,
 $\last(\run_{n\shortminus 1})\allowbreak=\hastate$, $\scheduler_{\prophecy}(\hastate)=(a,v)$, and $\hastate'$ unique state of $\RAC$ with $\last(\run_{n\shortminus 1})\semanticsarrowsac{a,v}\hastate'$.
\end{enumerate}%
\end{definition}
\noindent For simplicity, we only consider goal \emph{locations}, %
though our approach can be 
extended to other goal definitions.
Let  $\Goal$ be a set of goal locations, $\tmax\in\Reals_{\geq 0}$ a time bound, and $\jumpmax\in\Naturals$ a jump bound.
We say that \emph{$\scheduler_{\prophecy}$ reaches $\Goal$ in $\RAC$ within bounds $(\tmax,\jumpmax)$} if 
there is $n\in\Naturals$, $\run=\run(\RAC,\scheduler_{\prophecy},n)=\hastate_0\xrightarrow{a_1}\ldots\xrightarrow{a_n}\hastate_n$, $\hastate_i=(\ell_i,\valCont_i,\sample)$ such that
(i) $|\{i\in\{1,\ldots,n\}\mid a_i\in\Edge\}|\leq \jumpmax$, 
(ii) $\ell_n\in\Goal$,
and
(iii) $(\valCont_n)\proj{t}\leq\tmax$.
A run $\run\in\Initruns{\RAC}$ \emph{fulfills} a prophecy $\prophecy$, if there exists a scheduler $\scheduler_{\prophecy}$ for $\RAC$ that induces $\run$.

The minimum probability $p_{\RAC}^{\textsl{min}}(\Goal, \tmax,\jumpmax)$ of reaching $\Goal$ in $\RAC$ within bounds $(\tmax,\jumpmax)$ is defined by the integral over all those prophecies for which the goal cannot be avoided within the given bounds.

\begin{definition}%
    \label{lemma:probability}
For each scheduler $\scheduler\in\SchedulersProphetic{\RAC}$, let
$\PropheciesSchedulerMin{}
$ be the set of all
prophecies $\prophecy$ for which $\scheduler_{\prophecy}$ cannot avoid $\Goal$ in $\RAC$ within bounds $(\tmax,\jumpmax)$.
Then, the \emph{prophetic minimum reachability probability to reach $\Goal$ in $\RAC$ within bounds $(\tmax,\jumpmax)$} is: 
\begin{equation}\label{eq:scheduler}\textstyle
p_{\RAC}^\textsl{min}(\Goal, \tmax,\jumpmax) 
= \min_{\scheduler \in \SchedulersProphetic{\RAC}} \Big( \int_{\PropheciesSchedulerMin{}} G(\prophecy) \ d\prophecy \Big),
\end{equation}
where  $G(\prophecy){=} \prod_{\clockr\in\LabRandom} \Distr(\clockr)(\prophecy(\clockr))$.
Let $\SchedulersProphMin{\RAC} \subseteq\SchedulersProphetic{\RAC}$ be the set of all schedulers yielding this minimum probability.%
\footnote{Note that, $G$ is a continuous function, $\PropheciesSchedulerMin{} \subseteq \Realsposzero^{\dimRandom}$, i.e., the integral over $\PropheciesSchedulerMin{}$ is defined.}
\end{definition}
\noindent
A \emph{minimum prophetic scheduler} makes optimal choices for all prophecies, i.e., always avoids the goal if possible.  In our setting, the minimum probability always exists, since the formalism restricts to weak inequalities and admits only finitely many scheduler decisions due to the jump bound and the alternation between time steps and discrete jumps.

\paragraph{Remark}
Other scheduler classes \cite{dargenio2018HierarchySchedulerClasses,MathisMemocode}, e.g., non-prophetic or memoryless, have less information and cannot achieve lower minimum  reachability probabilities. For stochastic automata, \cite{dargenio2018HierarchySchedulerClasses} shows that history-dependent prophetic schedulers are most powerful. For hybrid Petri nets, \cite{MathisDiss} establishes that history-dependent prophetic schedulers are more powerful than their memoryless and non-prophetic counterparts. As they can be transformed into singular automata with random clocks (a subclass of RAC)~\cite{transformationNFM}, we expect these results can be transferred.%

\subsection{Reach Tree}\label{subsec:reach}
First, we analyze the reachability of the goal set via flowpipe construction while disregarding the stochastic distributions. %
This step is represented in Figure~\ref{fig:process-wenigerdetail} by the first arrow and does not differ from the computation of maximum reachability probabilities~\cite{journal}.
More precisely, we first consider arbitrary expiration times (even if they are not in the support of the respective distributions). Technically, in the reachability analysis we omit (i) the expiration time samples, (ii) condition $\valCont\proj{\clockr} = \sample\proj{\clockr} $ in Rule \texttt{Jump}$_{\tinyS}$, and (iii) condition $\forall \clockr \in \LabRandom. \     \valCont'\proj{\clockr} \leq  \sample\proj{\clockr}$ in Rule \texttt{Flow}.
Thus, we consider states $(\ell,\valCont)$ instead of $(\ell,\valCont,\sample)$.
This relaxed semantics allows us to first disregard the stochastic behavior.
We use $(\ell,\valSet)$ to denote $\{(\ell,\valCont) \mid \valCont \in \valSet\}$. %

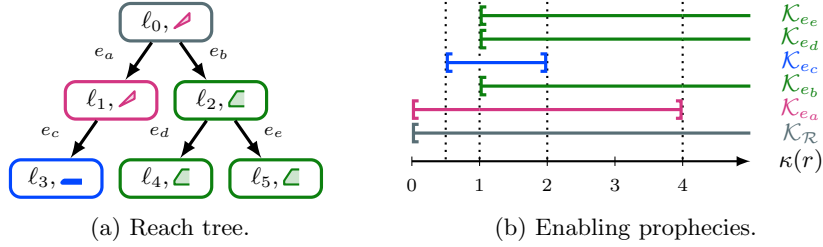
\begin{figure}[t]
    \centering
    \begin{subfigure}{0.43\textwidth}
    \centering
        \begin{tikzpicture}[
	scale=1,
	node distance = 1cm and 0.8cm,
	baseline,
	remember picture,
	n/.style={draw, text width = 1cm, text = black, %
		align = center, font= \footnotesize, rounded corners, very thick, execute at begin node=\setlength{\baselineskip}{8pt}%
	},
	tn/.style={draw, text width = 0.2cm, text = black, %
		align = center, font= \footnotesize, circle, very thick, execute at begin node=\setlength{\baselineskip}{8pt}%
	},
	nd/.style={draw=none, text width = 2.5cm, anchor=left, text = black, %
		align = left, font= \footnotesize, rounded corners, very thick, execute at begin node=\setlength{\baselineskip}{8pt}%
	},
	en/.style={draw=none, minimum height=0cm, font = \scriptsize, align = center,outer sep=1mm},%
	ref/.style={draw, circle, minimum width=1.5mm, inner sep=0, fill,samplevaluescolor},%
	k/.style={draw=none, circle, minimum width=0, inner sep=0,outer sep=0, fill,samplevaluescolor},%
	c/.style={draw, fill, black, circle, inner sep=0, outer sep=0, minimum size=1mm},
	l/.style={anchor=west, inner sep=0, font=\footnotesize},
        secondlayerdistancesingle/.style={node distance =0.5cm and 0.1cm},
        secondlayerdistance/.style={node distance =0.4cm and 0.1cm},
        firstlayerdistance/.style={node distance =0.7cm and  1.75cm}
	]

    \begin{scope}
        
	\node[n, draw=cloc0](l0) at (0,0)  { $\ell_0,\tikzboxpinkshape$ };
	\node[n, draw=cloc1, secondlayerdistancesingle,below left = of l0.south]  (l1) {$\ell_1,\tikzboxpinkshape$}; 
	\node[n, draw=cloc3, secondlayerdistancesingle, below right= of l0.south]  (l2) {$\ell_2,\tikzboxgreenshape$};

	\node[n, draw=cloc2, secondlayerdistancesingle,below left = of l1.south]  (l3) {$\ell_3,\tikzboxblueshape$}; 
    
	\node[n, draw=cloc3, secondlayerdistancesingle,below left = of l2.south]  (l4) {$\ell_4,\tikzboxgreenshape$}; 
	\node[n, draw=cloc3, secondlayerdistancesingle, below right= of l2.south]  (l5) {$\ell_5,\tikzboxgreenshape$};

\draw[-latex, very thick] (l0) to
     node[en, left,stochConflict,yshift=1mm] 
     { $e_a$ } 
	(l1);
    
\draw[-latex, very thick] (l0) to
    node[en, right,stochConflict,yshift=1mm] 
     { $e_b$ } 
	(l2);

    \draw[-latex, very thick] (l1) to
     node[en, left,stochConflict,yshift=1mm] 
     { $e_c$ } 
	(l3);

    \draw[-latex, very thick] (l2) to
     node[en, left,stochConflict,yshift=1mm] 
     { $e_d$ } 
	(l4);
    
\draw[-latex, very thick] (l2) to
    node[en, right,stochConflict,yshift=1mm] 
     { $e_e$ } 
	(l5);
    
    \end{scope}

\end{tikzpicture}
        \caption{Reach tree.}\label{fig:runningexample-reachtree}
    \end{subfigure}\hfil
    \begin{subfigure}{0.55\textwidth}
    \centering
        \begin{tikzpicture}[
xscale=0.9,
yscale=1.25,
lineshiftup/.style={yshift=0.6pt},
lineshiftdown/.style={yshift=-0.6pt},
lineshiftdowntwice/.style={yshift=-1.8pt},
decoratenodebrace/.style={decorate,decoration={brace,amplitude=4pt,mirror},yshift=0pt, font=\scriptsize},
font= \footnotesize]

	\useasboundingbox (-0.2,-0.7) rectangle (6.4,1.5);

\begin{scope}[yshift=-0.3cm]
    \draw [thick,black,-latex] (0,0) -- (5,0);
    \node[fill=white] (x) at (5.75,0) {$\prophecy(r)$};
    \foreach \x in {0,1,2,...,4}  {
        \node[tick] (\x) at (\x,-0.25) {$\x$}; 
        \draw (\x,-0.05)--(\x,+0.05);
    }
    \foreach \x in {0.5,1,2,4}  {
        \draw[dotted,thick] (\x,0.02)--(\x,+1.7);
    }
\end{scope}

\begin{scope}[yshift=0cm]
   \node[draw=none, font=\footnotesize,cloc0] at (5.75,0) {$\Prophecies{\RAC}$};
    \draw [Bracket-,very thick,cloc0] (0,0) -- (5,0);
\end{scope}
\begin{scope}[yshift=0.25cm]
   \node[draw=none, font=\footnotesize,cloc1] at (5.75,0) {$\Prophecies{e_a}$};
    \draw [Bracket-,very thick,cloc1] (0,0) -- (1,0);
    \draw [-Bracket,very thick,cloc1,] (1,0) -- (4,0);
\end{scope}
\begin{scope}[yshift=0.5cm]
    \node[draw=none, font=\footnotesize,cloc3] at (5.75,0) {$\Prophecies{e_b}$};
    \draw [Bracket-,very thick,cloc3,yshift=0mm, ,] (1,0) -- (4,0);
    \draw [very thick,cloc3,yshift=0mm] (4,0) -- (5,0);
\end{scope}
\begin{scope}[yshift=0.75cm]
   \node[draw=none, font=\footnotesize,cloc2] at (5.75,0) {$\Prophecies{e_c}$};
    \draw [Bracket-Bracket,very thick,cloc2] (0.5,0) -- (2,0);
\end{scope}
\begin{scope}[yshift=1cm]
    \node[draw=none, font=\footnotesize,cloc3] at (5.75,0) {$\Prophecies{e_d}$};
    \draw [Bracket-,very thick,cloc3,yshift=0mm, ,] (1,0) -- (4,0);
    \draw [very thick,cloc3,yshift=0mm] (4,0) -- (5,0);
\end{scope}
\begin{scope}[yshift=1.25cm]
    \node[draw=none, font=\footnotesize,cloc3] at (5.75,0) {$\Prophecies{e_e}$};
    \draw [Bracket-,very thick,cloc3,yshift=0mm, ,] (1,0) -- (4,0);
    \draw [very thick,cloc3,yshift=0mm] (4,0) -- (5,0);
\end{scope}

\end{tikzpicture}
        \caption{Enabling prophecies.}\label{fig:runningexample-prophecies}
    \end{subfigure}
        \caption{Running Example: Reach tree and enabling prophecies.}\label{fig:runningexample-tree-and-prophecies}
\end{figure}

We assess the bounded reachability of $\Goal$ within bounds $\tmax$ and $\jumpmax$ via iteratively computing bounded time successors 
$\ftc{\ell}(\valSet) = \{\valCont'\,{\in}\, \Reals^d\,|\, 
\exists\valCont\in\valSet,\delta\in\Realsposzero,\allowbreak \text{ s.t. } \allowbreak(\ell,\valCont) \semanticsarrowsac{\delta} (\ell,\valCont')\allowbreak 
\text{ with }\allowbreak \smash{\valCont'\proj{t}\leq\tmax}  \}$
and discrete jump successors
$\fosr{e}(\valSet)=\{\valCont'\in\Reals^d\,|\,
\exists\valCont\in\valSet,\allowbreak \text{ s.t. } \allowbreak
(\source(e),\valCont) \semanticsarrowsac{e} (\target(e),\valCont') 
\}$  
for certain
$\ell\in\Loc$, $\valSet\subseteq\Reals^d$ and $e\in\Edge$ (without considering the stochastic distributions).

The result of reachability analysis is a \emph{reach tree}, fully defined in Appendix~\ref{app:reachtree}.
Intuitively, for a location $\ell\in\Loc$  with $\Init(\ell)\not=\emptyset$, the \emph{$(\tmax,\jumpmax)$-bounded reach tree for $\RAC$ from $\ell$ 
to $\Goal$} is an annotated tree $\reachtree = (N,E)$ with a nonempty finite set of \emph{nodes} $N\subseteq \Naturals\times\Loc\times2^{\Reals^{\dimCont}}$ and a set of \emph{edges} $ E = N \times  \Edge \times N$.
Each node $i=(id^i,\ell^i,\valSet^i)\in N$ has a unique  identifier $id^i$, which increases strictly monotonically  from parents to children, and stores a set of states $(\ell^i,\valSet^i)$. 
The node $root$ stores the state set $(\ell,\ftc{\ell}(\Init(\ell)))$.
For each $i=(id^i,\ell^i,\valSet^i)\in N$, %
if either $\ell^i\in\Goal$, 
or if the path length from $root$ to $i$ equals $\jumpmax$,  node $i$ has no children.
Otherwise,  node $i$ has exactly one child $(j,\ell^j,\valSet^j)$ for each edge $e\in\Edge$ with $\source(e)=\ell^i$, $\target(e)=\ell^j$, and $\valSet^j=\ftc{\target(e)}(\fosr{e}(\valSet^i))\not=\emptyset$.
\noindent Figure~\ref{fig:runningexample-reachtree} shows the reach tree of the running example, indicating
  state sets $\valSet$ by icons in reference to Figure~\ref{fig:rac-example-flowpipe}.

Let in the following $\reachtree=(N,E)$ be the $(\tmax,\jumpmax)$-bounded reach tree for $\RAC$ from some $\ell\in\Loc$ to $\Goal$ with $\Init(\ell)\not=\emptyset$.
We use $\parent(i)$ and $\children(i)$ to denote the parent resp. children of node $i$ in $\reachtree$.
A \emph{path in $\reachtree$} is a sequence $\treePath=i_0,\dots,i_m$ with $m\in\Naturals$, %
and for $\smash{0 \leq j < m}$, $\smash{i_j\in N}$ and $(i_j,e,i_{j{+}1})\in E$ for some $e\in\Edge$. %
For each node $i\in N$, let $\treePath_i$ denote the path in $\reachtree$ from the root to $i$.
Paths in $\reachtree$ represent the runs of $\RAC$:
\begin{definition} \label{def:represents}

 A path $\treePath\!=\!i_0,\dots,i_m$ in $\reachtree$ \emph{represents} a run
$\smash{\run\!=\!\hastate_0\xrightarrow{a_1}\ldots\xrightarrow{a_n}\hastate_n}$$\in\Initruns{\RAC}$ with $\hastate_i\!=\!(\ell_i,\valCont_i)$ 
    if $(\valCont_n)\proj{t}\leq \tmax$ and either $\smash{n=m=0}$, $\ell_0=\ell$, or %
    \begin{enumerate}[(i)]
    \item $n\!>\!0$, $(i_{m\shortminus 1},a_n,i_m)\!\in\! E$, and 
    $\treePath_{i_{m-1}}$
    represents $\smash{\hastate_0\xrightarrow{a_1}\!{\ldots}\!\xrightarrow{a_{n{\shortminus}1}}\hastate_{n{\shortminus}1}}$, or
    \item $n\!>\!0$, $a_n=\delta\in\Realsposzero$, and  $\treePath$ represents $\hastate_0\xrightarrow{a_1}\ldots\xrightarrow{a_{n\shortminus 1}}\hastate_{n\shortminus 1}$.     
    \end{enumerate}
\end{definition}
Note that even though the above relaxed semantics disregards the actual stochastic behavior, the random clocks still measure the durations of enabledness. 
Thus a run $\run$
under the relaxed semantics is a run under the original semantics (c.f. Figure~\ref{fig:operationalsemanticsRAC}) iff each stochastic jump with label $\clockr$ is taken in a state $(\ell,\valCont,\sample)$ with $\valCont\proj{\clockr}=\sample\proj{\clockr}\in\support(\Distr(\clockr))$
(i.e., if there exists a prophecy $\prophecy$ that is fulfilled by $\run$). 
Recall  that on each path in the underlying graph of $\RAC$, each random event labels at most one jump. Thus after the occurrence of $\clockr$, its duration of enabledness will not change. 
Hence, the leafs of the reach tree encode the expiration times (prophecies) for their paths and hence \emph{represented} runs, which allows to compute the probability mass of each path in the reach tree by integration.

\section{Classification of Branchings in a Reach Tree}
\label{sec:classification}
To identify a scheduler that minimizes the probability to reach $\Goal$, we classify the branchings in the reach tree into stochastic, nondeterministic and mixed branchings. 
The classification is  based on the set of enabling prophecies. 

\subsection{Enabling Prophecies}\label{subsec:enabling-proph}
Recall that a prophecy defines the expiration times of all random events and a 
scheduler defines a function for each prophecy $\prophecy\in \Prophecies{\RAC}$. Hence, given a run $\run$, a scheduler chooses an available time step or jump from $\TimeChoices_{\RAC}(\run) \cup \JumpChoices_{\RAC}(\last(\run))$.
Prophecies might constrain these choices.

\begin{definition}%
A set of prophecies $\Prophecies{} \subseteq\Prophecies{\RAC}$ is called \emph{enabling} for
    $(i,e,j)\in E$ 
    with child node $j\in \children(i)$  and jump $e$,
    if 
    \begin{enumerate}[(i)]
        \item for all $\prophecy \in \Prophecies{}$ there exists a run $\run=\run'\xrightarrow{e}\hastate\in \Initruns{\RAC}$ such that $\run$ fulfills $\prophecy$ and $\treePath_i$ represents $\run'$, and
        \item for all runs $\run=\run'\xrightarrow{e}\hastate\in \Initruns{\RAC}$, if $\treePath_i$ represents $\run'$, then  there  exists a prophecy $\prophecy\in \Prophecies{}$ that is fulfilled by $\run$.
    \end{enumerate}
\end{definition}
\noindent
Let $\Prophecies{e}(\treePath_i)$ be the set of all enabling prophecies for jump $e$ with $(i,e,j)\in E$, given path $\treePath_i$, which we  refer to as $\Prophecies{j}$, and 
$\Prophecies{root}=\Prophecies{\RAC}$ for the root node.
Note that when $i$ is not the root, i.e., there exists $(i',e',i)\in E$, it always holds that $\Prophecies{e}(\treePath_i)\subseteq\Prophecies{e'}(\treePath_{i'})$.
We also remark that the set of enabling prophecies $\Prophecies{e}(\treePath_i)$ for child node $j$  with $(i,e,j)\in E$ %
can also be 
a \emph{null set}, i.e., a Lebesgue measurable set  with measure zero,
e.g., the empty set $\emptyset$ or a point interval. 
If the enabling prophecies form a null set,  the probability that any $\prophecy\in\Prophecies{e}(\treePath_i)$ realizes is zero, i.e.,
the child node $j$ can  be reached with probability zero. %
Let $\validChildren(i)$ denote the set of all child nodes with positive probability, i.e.
all $j\in \children(i)$ with $(i,e,j)\in E$, $e \in \Edge$ and  $\Prophecies{e}(\treePath_i)$ not a null set.
Figure~\ref{fig:runningexample-prophecies} illustrates the enabling prophecies for all nodes in the running example.

\subsection{Different Branching Types}\label{subsec:branchings}
We classify branchings at node $i$ in a reach tree as \emph{stochastic}, \emph{nondeterministic} or \emph{mixed}, based on the sets of enabling prophecies for $i$'s 
children $\validChildren(i)$.  
Here, we differentiate whether the enabling prophecies of all children are (i) all disjoint, (ii) all equal, or (iii) partly overlapping. Intuitively, at any inner node in a reach tree, the branching is stochastic if for each state in the node there exists a state with the same prophecy in exactly one child node, i.e., if the prophecy uniquely determines the next jump. In contrast, a branching is nondeterministic if the same prophecies are represented in all children, meaning that the choice of the next jump is independent of the prophecy. In the remaining mixed branches, the next jump is determined in some cases stochastically and in other nondeterministically.
Examples showcase the different branching types in Figure~\ref{fig:branchings}.
Assume for those fragments of RAC that all $x\in\VarCont$ have $\valCont\proj{x}=0$ initially.
Figure~\ref{fig:branchings-enabling-prophecies} illustrates  
 the enabling prophecies
for those examples.

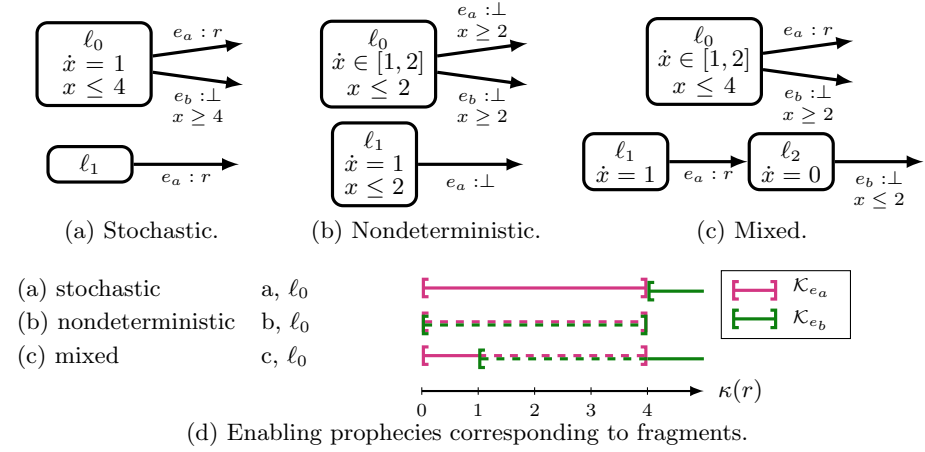
\begin{figure}[t]
    \centering
    \begin{subfigure}[b]{.29\linewidth}
    \centering
        \begin{tikzpicture}[
scale=1.1,
baseline,remember picture,
node distance = -0.5cm and 1.4cm,
n/.style={draw, text width = 1.3cm, %
align = center, font= \footnotesize, rounded corners, very thick, execute at begin node=\setlength{\baselineskip}{8pt}%
},
en/.style={draw=none, minimum height=0cm, font = \scriptsize, align = center},%
c/.style={draw, fill, black, circle, inner sep=0, outer sep=0, minimum size=1mm},
l/.style={anchor=west, inner sep=0, font=\footnotesize},
init/.style={en, 
text width=1.2cm,
anchor=south west,  
inner xsep=0pt},
]

\useasboundingbox (-0.65,-0.68) rectangle (1.8,0.5);

\node[n](l0) at (0,0) {
$\ell_0$\\\flowsepspace 
    $\dot{x}=1$\\\flowsepspace 
    $x\leq 4$
};

\coordinate (a) at (1.8,0.25);
\coordinate (b) at (1.8,-0.25);

\draw[-latex, very thick] (l0) to node[en, above] {$e_a:r$} (a);
\draw[-latex, very thick] (l0) to node[en, below] {$e_b:\norc$\\$x\geq 4$} (b);

\end{tikzpicture}
        \begin{tikzpicture}[
scale=1.1,
baseline,remember picture,
node distance = -0.5cm and 1.4cm,
n/.style={draw, text width = 0.9cm, %
align = center, font= \footnotesize, rounded corners, very thick, execute at begin node=\setlength{\baselineskip}{8pt}%
},
en/.style={draw=none, minimum height=0cm, font = \scriptsize, align = center},%
c/.style={draw, fill, black, circle, inner sep=0, outer sep=0, minimum size=1mm},
l/.style={anchor=west, inner sep=0, font=\footnotesize},
init/.style={en, 
text width=1.2cm,
anchor=south west,  
inner xsep=0pt},
]

\useasboundingbox (-0.65,-0.45) rectangle (1.8,0.5);

\node[n](l0) at (0,0) {
        $\ell_1$\\\invariantsepspace 
};

\coordinate (a) at (1.8,0);

\draw[-latex, very thick] (l0) to node[en, below] {$e_a:r$} (a);

\end{tikzpicture}
        \caption{Stochastic.}\label{fig:branchings-stochastic}
    \end{subfigure}\hfill
    \begin{subfigure}[b]{.3\linewidth}
    \centering
        \begin{tikzpicture}[
scale=1.1,
baseline,remember picture,
node distance = -0.5cm and 1.4cm,
n/.style={draw, text width = 1.3cm, %
align = center, font= \footnotesize, rounded corners, very thick, execute at begin node=\setlength{\baselineskip}{8pt}%
},
en/.style={draw=none, minimum height=0cm, font = \scriptsize, align = center},%
c/.style={draw, fill, black, circle, inner sep=0, outer sep=0, minimum size=1mm},
l/.style={anchor=west, inner sep=0, font=\footnotesize},
init/.style={en, 
text width=1.2cm,
anchor=south west,  
inner xsep=0pt},
]

\useasboundingbox (-0.65,-0.68) rectangle (1.8,0.5);

\node[n](l0) at (0,0) {
    $\ell_0$\\\flowsepspace 
    $\dot{x}\in[1,2]$\\\flowsepspace  
    $x\leq 2$
};
\coordinate (a) at (1.8,0.25);
\coordinate (b) at (1.8,-0.25);
\draw[-latex, very thick] (l0) to node[en, above] {$e_a:\norc$\\$x\geq 2$} (a);
\draw[-latex, very thick] (l0) to node[en, below] {$e_b:\norc$\\$x\geq 2$} (b);

\end{tikzpicture}%
        \begin{tikzpicture}[
scale=1.1,
baseline,remember picture,
node distance = -0.5cm and 1.4cm,
n/.style={draw, text width = 0.9cm, %
align = center, font= \footnotesize, rounded corners, very thick, execute at begin node=\setlength{\baselineskip}{8pt}%
},
en/.style={draw=none, minimum height=0cm, font = \scriptsize, align = center},%
c/.style={draw, fill, black, circle, inner sep=0, outer sep=0, minimum size=1mm},
l/.style={anchor=west, inner sep=0, font=\footnotesize},
init/.style={en, 
text width=1.2cm,
anchor=south west,  
inner xsep=0pt},
]

\useasboundingbox (-0.65,-0.45) rectangle (1.8,0.5);

\node[n](l0) at (0,0) {
    $\ell_1$\\\flowsepspace  
    $\dot{x}=1$\\\flowsepspace 
    $x\leq 2$
};
\coordinate (a) at (1.8,0);
\draw[-latex, very thick] (l0) to node[en, below] {$e_a:\norc$} (a);

\end{tikzpicture}
        \caption{Nondeterministic.}\label{fig:branchings-nondet}
    \end{subfigure}\hfill
    \begin{subfigure}[b]{.38\linewidth}
    \centering
        \begin{tikzpicture}[
scale=1.1,
baseline,remember picture,
node distance = -0.5cm and 1.4cm,
n/.style={draw, text width = 1.3cm, %
align = center, font= \footnotesize, rounded corners, very thick, execute at begin node=\setlength{\baselineskip}{8pt}%
},
en/.style={draw=none, minimum height=0cm, font = \scriptsize, align = center},%
c/.style={draw, fill, black, circle, inner sep=0, outer sep=0, minimum size=1mm},
l/.style={anchor=west, inner sep=0, font=\footnotesize},
init/.style={en, 
text width=1.2cm,
anchor=south west,  
inner xsep=0pt},
]

\useasboundingbox (-0.65,-0.68) rectangle (1.8,0.5);

\node[n](l0) at (0,0) {
    $\ell_0$\\\flowsepspace 
    $\dot{x}\in[1,2]$\\\flowsepspace  
    $x\leq 4$
};

\coordinate (a) at (1.8,0.25);
\coordinate (b) at (1.8,-0.25);

\draw[-latex, very thick] (l0) to node[en, above] {$e_a:r$} (a);
\draw[-latex, very thick] (l0) to node[en, below] {$e_b:\norc$\\$x\geq 2$} (b);

\end{tikzpicture}
        \begin{tikzpicture}[
scale=1.1,
baseline,remember picture,
node distance = -0.5cm and 1.4cm,
n/.style={draw, text width = 1.2cm, %
align = center, font= \footnotesize, rounded corners, very thick, execute at begin node=\setlength{\baselineskip}{8pt}%
},
en/.style={draw=none, minimum height=0cm, font = \scriptsize, align = center},%
c/.style={draw, fill, black, circle, inner sep=0, outer sep=0, minimum size=1mm},
l/.style={anchor=west, inner sep=0, font=\footnotesize},
init/.style={en, 
text width=1.2cm,
anchor=south west,  
inner xsep=0pt},
]

\useasboundingbox (-0.55,-0.45) rectangle (3.6,0.5);

\node[n,text width=0.9cm](l0) at (0,0) {
    $\ell_1$\\\flowsepspace  
    $\dot{x}=1$
};
\node[n,text width=0.9cm](l1) at (2,0) {
    $\ell_2$\\\flowsepspace  
    $\dot{x}=0$
};
\coordinate (a) at (2,0);
\coordinate (b) at (3.6,0);
\draw[-latex, very thick,] (l0) to node[en, below] {$e_a:r$} (l1);
\draw[-latex, very thick] (l1) to node[en, below] {$e_b:\norc$\\$x\leq 2$} (b);

\end{tikzpicture}
        \caption{Mixed.}\label{fig:branchings-mixed}
    \end{subfigure}\vspace{2ex}
    \begin{subfigure}[b]{1\linewidth}
        \begin{tikzpicture}[
xscale=0.5,
lineshiftup/.style={yshift=0.6pt},
lineshiftdown/.style={yshift=-0.6pt},
font= \footnotesize]

\begin{scope}[yshift=0.9cm]
\node[draw=none, anchor=west, font=\footnotesize] at (-11,0) {(a) stochastic};
\node[draw=none, anchor=west, font=\footnotesize] at (-4.5,0) {
       \subref{fig:branchings-stochastic}, $\ell_0$
    };
\begin{scope}[xscale=1.5]
    \draw [Bracket-Bracket,very thick,cloc1,lineshiftup] (0,0) -- (4,0);
    \draw [Bracket-,very thick,cloc3,yshift=-0mm,lineshiftdown] (4,0) -- (5,0);
\end{scope}
\end{scope}

\begin{scope}[yshift=0.45cm]
    \node[draw=none, anchor=west, font=\footnotesize] at (-11,0) {(b) 
    nondeterministic}; 
    \node[draw=none, anchor=west, font=\footnotesize] at (-4.5,0) {
        \subref{fig:branchings-nondet}, $\ell_0$
    };
\begin{scope}[xscale=1.5]
    \draw [Bracket-Bracket,very thick,cloc1,dashed,lineshiftup] (0,0) -- (4,0);
    \draw [Bracket-Bracket,very thick,cloc3, dashed,lineshiftdown] (0,0) -- (4,0);
\end{scope}
\end{scope}

\node[draw=none, anchor=west, font=\footnotesize] at (-11,0) {(c) mixed};
    \node[draw=none, anchor=west, font=\footnotesize] at (-4.5,0) {
       \subref{fig:branchings-mixed}, $\ell_0$
    };
\begin{scope}[xscale=1.5]
\draw [Bracket-,very thick,cloc1,lineshiftup] (0,0) -- (1,0);
\draw [-Bracket,very thick,cloc1,lineshiftup,dashed] (1,0) -- (4,0);
\draw [Bracket-,very thick,cloc3,yshift=0mm, dashed,lineshiftdown] (1,0) -- (4,0);
\draw [very thick,cloc3,yshift=0mm,lineshiftdown] (4,0) -- (5,0);
\end{scope}

\begin{scope}[xscale=1.5]
\draw [thick,black,-latex] (0,-0.45) -- (5,-0.45);
\node[fill=white,anchor=west] (x) at (5.1,-0.45) {$\prophecy({\clockr})$};

\begin{scope}[yshift=-0.45cm]
\foreach \x in {0,1,2,...,4}  {
    \node[tick] (\x) at (\x,-0.25) {$\x$}; 
    \draw (\x,-0.05)--(\x,+0.05);
}
\end{scope}

\begin{scope}[yshift=-0.4cm]
\node[draw, fill=white, text width = 1.5cm, rectangle, anchor=west, font=\scriptsize] (legend) at (5.3,1.1) {
    $\qquad\quad\Prophecies{e_a}$\\[0.5em]
    $\qquad\quad\Prophecies{e_b}$}; 
\draw[very thick,Bracket-Bracket,cloc1] ($(legend.west)+(0.2,0.19)$) -- ($(legend.west)+(1,0.19)$);
\draw[very thick,Bracket-Bracket,cloc3] ($(legend.west)+(0.2,-0.19)$) -- ($(legend.west)+(1,-0.19)$);
\end{scope}
\end{scope}

\end{tikzpicture}
\vspace*{-2ex}
\caption{Enabling prophecies corresponding to fragments.
        }\label{fig:branchings-enabling-prophecies}
    \end{subfigure}
    \caption{Additional fragments of RAC that  cause different branchings. %
    }
    \label{fig:branchings}
\end{figure}

\paragraph{Stochastic Branching}\label{subsec:stoch-branch}
If a scheduler has no choice in the presence of a random event, 
we call a branching \emph{stochastic}, as the stochastic behavior determines the duration of the time step and chooses the jump. %
The intersection of the enabling prophecies for each two children is then a null set and we call  them \emph{disjoint}. %
\begin{definition}\label{def:stochasticbranching}
    A branching at node $i$ in $\reachtree$ is called \emph{stochastic}, if
    \begin{enumerate}[(i)]
    \item %
     there exist at least two child nodes in $\validChildren(i)$, and
     the enabling prophecies $\Prophecies{j},\Prophecies{k}$ of each two child nodes $j,k\in \validChildren(i)$, $j\not=k$, are disjoint, or
     \item %
     there exists only one child node $j\in \validChildren(i)$, 
     and $e\in \Edge$ with
       $(i,e,j)\in E$ is stochastic with $\Event(e)\in\LabRandom$.
     \end{enumerate}
\end{definition}

\noindent 
 Figure~\ref{fig:branchings-stochastic} provides stochastic branching examples: $\ell_0$  illustrates (i) and $\ell_1$ illustrates (ii). 
 In both locations, the stochastic jump $e_a$ is taken when random event $r$ occurs. 
In $\ell_0$, nonstochastic jump $e_b$ can be taken only when $\valCont\proj{x}=4$ and if the stochastic event did not occur before.
Figure~\ref{fig:branchings-enabling-prophecies} illustrates the enabling prophecies for $e_a$ and $e_b$ from $\ell_0$ as disjoint: $\prophecy(r)\leq 4$ for $e_a$ and $\prophecy(r)\geq4$ for $e_b$.

\paragraph{Nondeterministic Branching}\label{subsec:nondet-branch}

We call a branching \emph{nondeterministic}, if prophecies do not restrict the time and jump choices. %
While {discrete} nondeterministic choices  are explicitly shown in the structure of the reach tree,  choices  ruled by time or rate nondeterminism are not visible in the  reach tree.
In a nondeterministic branching, each two children have the same enabling prophecies.

\begin{definition}\label{def:nondetbranching}
A branching at node $i$ in $\reachtree$ is \emph{nondeterministic}, if  for all child nodes $j \in \validChildren(i)$ with $(i,e,j)\in E$, 
 $e\in \Edge$ is nonstochastic with $\Event(e)=\:\norc$
and the enabling prophecies $\Prophecies{j}$ are equal to $i$'s enabling prophecies  $\Prophecies{i}$.
\end{definition}

\noindent
Note that if the scheduler has no choice at all and the next step is \emph{independent from the given prophecy},
we still classify the branching as \emph{nondeterministic}.

Figure~\ref{fig:branchings-nondet} depicts nondeterministic branchings:
In $\ell_0$, the scheduler may choose a time step $\delta\leq2$ with rate s.t. $\valCont\proj{x}=2$ after $\delta$, and the following jump ($e_a$ or $e_b$). 
In $\ell_1$, 
the only choice is the duration of the time step before leaving. %
 The enabling prophecies for $e_a$ and $e_b$ from $\ell_0$ are shown in Figure~\ref{fig:branchings-enabling-prophecies}:  both jumps have the same enabling prophecies ($\prophecy(r)\leq 4$) and fully cover the set of enabling prophecies of parent $\ell_0$, if assumed to be $\prophecy(r)\leq 4$. 

 Further, the running example contains a nondeterministic branching at node $(\ell_2,\tikzboxgreenshape)$, which can be reached with any prophecy $\prophecy(r)\geq 1$ (c.f. Figure~\ref{fig:runningexample-prophecies}).
 The enabling prophecies for both children $(\ell_4,\tikzboxgreenshape)$ and $(\ell_5,\tikzboxgreenshape)$ are equal and fully cover the enabling prophecies of $(\ell_2,\tikzboxgreenshape)$: the decision of which jump to take is completely decided by the scheduler, hence the branching is nondeterministic.

\paragraph{Mixed Branching}\label{subsec:mixed-branch}

In a mixed branching, the possible choices of a scheduler depend on  the prophecy, without being completely determined by it. 
Thus, the enabling prophecies for at least two children are not disjoint.

\pagebreak
\begin{definition}%
\label{def:mixedbranching}
     A branching at node $i$ is called \emph{mixed}, if 
      \begin{enumerate}[(i)]
      \item 
     there exist at least two
     child nodes $j,k \in \validChildren(i)$, $j\not=k$, with overlapping sets of enabling prophecies $\Prophecies{j},\Prophecies{k}$, and either
        \begin{enumerate}[(a)]
            \item there also exist some child nodes $l,m \in \validChildren(i)$
        with non-equal sets of enabling prophecies $\Prophecies{l},\Prophecies{m}$, or
            \item  the enabling prophecies $\Prophecies{l}$ of all $l \in \validChildren(i)$ are equal to $i$'s enabling prophecies $\Prophecies{i}$,
            and there exists some child   $m \in \validChildren(i)$  with $(i,e,m)\in E$ and $e$ is a stochastic jump with $\Event(e)=r\in\LabRandom$, or
            \item  the enabling prophecies of all child nodes $l \in \validChildren(i)$ are real subsets of $i$'s enabling prophecies $\Prophecies{i}$, or
        \end{enumerate}
        \item  there exists only one child $j\in \validChildren(i)$ with 
        $(i,e,j)\in E$ such that $e\in \Edge$ nonstochastic with $\Event(e)=\:\norc$ and a set of enabling prophecies that is a real subset 
        of $i$'s enabling prophecies $\Prophecies{i}$.
     \end{enumerate}
\end{definition}

\noindent  
Figure~\ref{fig:branchings-mixed} illustrates mixed branchings:
(i) In $\ell_0$, for prophecies $\prophecy(\clockr)\leq1$, stochastic jump $e_a$ has to be taken before $e_b$ is enabled.
 For $\prophecy(\clockr)\geq4$, the scheduler must take jump $e_b$ before the invariant in $\ell_0$ is violated.
For  $\prophecy(\clockr)\in [1,4]$, the scheduler chooses  a suitable time step and  jump $e_b$, or waits the prophesied time until random event $r$ and jump $e_a$ (i.e., choose a time step $\delta=\prophecy(\clockr)$ with rate s.t. $\valCont\proj{x}\leq4$ after $\delta$, and then  takes $e_a$).
(ii) In $\ell_2$, %
the scheduler can either take $e_b$, or choose to stay in $\ell_2$.
Only for $\prophecy(\clockr)\leq2$, the scheduler can choose to take $e_b$ at any time. 
The scheduler can also choose to stay in $\ell_2$ forever for any prophecy. 
In particular, if $\prophecy(\clockr)\geq2$, $e_b$ is never enabled, so the scheduler has to stay in $\ell_2$.
The enabling prophecies for $e_a$ and $e_b$ from $\ell_0$, as  illustrated in Figure~\ref{fig:branchings-enabling-prophecies}, overlap for $\prophecy(r)\in[1,4]$ and uniquely determine the jump for $\prophecy(r)\leq 1$ and $\prophecy(r)\geq 4$.

 Further, the running example contains mixed branchings at nodes $(\ell_0,\tikzboxpinkshape)$ and $(\ell_1,\tikzboxpinkshape)$ (c.f. Figure~\ref{fig:runningexample-reachtree} for the reach tree).
 The enabling prophecies of both children of $(\ell_0,\tikzboxpinkshape)$ ($(\ell_1,\tikzboxpinkshape)$ and $(\ell_2,\tikzboxgreenshape)$) overlap partly: for $\prophecy(r)\in[1,4]$, the scheduler may choose between both children (c.f. Figure~\ref{fig:runningexample-prophecies}, and Definition~\ref{def:mixedbranching}, (i)(a))). 
While $(\ell_1,\tikzboxpinkshape)$ only has one child node $(\ell_3,\tikzboxblueshape)$, its enabling prophecies do not fully cover the parents enabling prophecies. Hence, while the scheduler can choose jump $e_c$ if $\prophecy\in [0.5,2]$, this choice is restricted by stochastic behavior as illustrated by the enabling prophecies (c.f. Figure~\ref{fig:runningexample-prophecies}, and Definition~\ref{def:mixedbranching}, (ii)).

\section{Computation of Minimum Reachability Probabilities}
\label{sec:computation}
The challenge in computing minimum reachability probabilities for RAC  arises from the combination of continuous nondeterminism and stochasticity, since stochastic and nondeterministic branchings need to be handled differently. While stochastic branchings require the \emph{union} over all enabling prophecies, nondeterministic branchings require  the \emph{intersection} over all enabling prophecies.
Recall that the minimum probability to reach a goal state is not dual to the maximum probability to reach a goal state and hence requires a dedicated approach. 

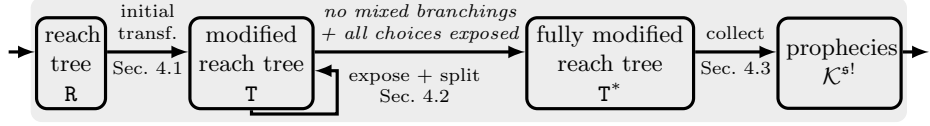
\begin{figure}[t]%
\centering
    \begin{tikzpicture}[
baseline,remember picture,
node distance = 0.3cm and 1cm,
box/.style={text width = 2cm, minimum height = 1.2cm, align = center, font= \footnotesize, very thick, anchor=north},
boxaut/.style={draw, text width = 2cm, minimum height = 2.5cm, align = center, font= \footnotesize, very thick, anchor=north, rounded corners},
boxlight/.style={draw, black!10!white, opacity=0.7, fill},
emptynode/.style={draw=none, inner sep=0},%
boxadd/.style={draw, text width = 2cm, minimum height = 1.4cm, align = center, font= \footnotesize, very thick, anchor=north},
section/.style={draw, fill=white, circle, inner sep=0.75mm, font= \footnotesize, thick},
e/.style={-latex, very thick},%
eno/.style={draw=none},%
l/.style={draw=none, minimum height=0cm, font = \scriptsize, align = center, text width = 1.5cm},%
textstyle/.style={font = \footnotesize},%
]

\pgfdeclarelayer{background}
\pgfsetlayers{background,main}

	\useasboundingbox (-0.83,-1.35) rectangle (11.37,0.2);

\node[box, draw, text width = 0.7cm, rounded corners,] (rt) at (0,0){reach tree \\ $\reachtree$};

\node[box, draw, text width = 1.42cm, rounded corners, right= 1.08cm of rt] (mrt) {modified reach tree \\ $\modtree$};

\node[box, draw, text width = 2cm, rounded corners, right= 2.8cm of mrt] (fmrt) {fully modified reach tree \\ $\fullymodtree$};

\node[box, draw, text width = 1.42cm, rounded corners, right= 1.08cm of fmrt] (proph) {prophecies \\ $\PropheciesSchedulerMin{}$};

\begin{pgfonlayer}{background}
\draw[boxlight, rounded corners] ($(rt.south west)+(-0.05,-0.15)$) rectangle ($(proph.north east)+(0.05,0.25)$);
\end{pgfonlayer}

\draw[e] ([yshift=+0.15cm]rt.east) --
node[l,above, minimum height=0.65cm] (modify) {initial\\ \phantom{\textit{p}}transf.\phantom{\textit{p}}}
node[l,below] (modify) {Sec.~\ref{subsec:modtree}}
([yshift=+0.15cm]mrt.west);

\draw[e] (mrt) -- ([yshift=-0.05cm]mrt.south) -| 
node[l,right, near end, minimum height=0.65cm,yshift=0.1cm,xshift=0cm,text width = 1.8cm] (modify) {expose + split\\Sec.~\ref{subsec:handling}} %
([xshift=0.3cm,yshift=-0.1cm]mrt.east) --
([yshift=-0.1cm]mrt.east);

\draw[e] ([yshift=+0.15cm]mrt.east) -- 
node[l,above,text width = 2.8cm] (compute) {\textit{no mixed branchings\\+ all choices  exposed}} %
([yshift=+.15cm]fmrt.west);

\draw[e] ([yshift=+0.15cm]fmrt.east) -- 
node[l,above] (compute) {\phantom{\textit{p}}collect\phantom{\textit{p}}} %
node[l,below] {Sec.~\ref{subsec:fully}} 
([yshift=+0.15cm]proph.west);

\draw[e] ([xshift=-0.35cm,yshift=+0.15cm]rt.west) -- ([yshift=+0.15cm]rt.west);
\draw[e] ([yshift=+0.15cm]proph.east) -- ([xshift=+0.35cm,yshift=+0.15cm]proph.east);

\end{tikzpicture}
    \caption{Steps of the minimum approach, detailing the gray box in Figure~\ref{fig:process-wenigerdetail}.}
    \label{fig:process-detailed}
\end{figure}

The computation of minimum reachability probabilities requires a \emph{modified reach tree}. %
We then modify this tree %
such that we can compute \emph{minimum reachability probabilities}. %
Figure~\ref{fig:process-detailed} illustrates the steps of this section, detailing Figure~\ref{fig:process-wenigerdetail}.

\subsection{Modified Reach Tree}\label{subsec:modtree}

We extend our
reach tree $\reachtree = (N_{\reachtree},E_{\reachtree})$
to a \emph{modified reach tree $\modtree=(N,E)$ for $\reachtree$}
with a nonempty finite set of nodes $N\subseteq \Naturals^2\times\Loc\times2^{\Reals^{\dimCont}}$ and a set of edges $ E = N \times  (\Edge\cup \{\tau\}) \times2^{\Prophecies{\RAC}} \times N$ such that
\begin{enumerate}[(i)]
    \item $N$ consists of one or more ``copies'' $i=((id_{\reachtree},id_{\modtree}), \ell, \valSet)$ of each original node $\rtmap{i}= (id_{\reachtree},\ell,\valSet)\in N_{\reachtree}$, augmented with a second identifier $id_{\modtree}\geq 0$ that re-establishes uniqueness, and 
    \item $E$ consists of (a) edges $(i,e,\Prophecies{},j)$ for some pairs of nodes whose origins are connected by an edge $(\rtmap{i},e,\rtmap{j})\in E_{\reachtree}$, annotated with  enabling prophecies $\Prophecies{} \subseteq \Prophecies{e}(\treePath_{\rtmap{i}})$, %
      and (b) \emph{stutter-edges} $(i,\tau,\Prophecies{},j)$ for some pairs of nodes with the same origin $\rtmap{i}=\rtmap{j}$, annotated with $\tau$ instead of a jump and some set of prophecies $\Prophecies{} \subseteq \Prophecies{e'}(\treePath_{\rtmap{\parent(i)}})$ for $(\parent(i),e',i)\in E$. %
\end{enumerate}
\noindent %
We use $\children(i)$,  $\parent(i)$, $\treePath_i$, and $\validChildren(i)$ analogously as for reach trees.
We define the \emph{initial transformation} of a reach tree $\reachtree= (N_{\reachtree},E_{\reachtree})$ into a modified reach tree $\modtree=(N,E)$ for $\reachtree$ as follows: (i) $N$ consists of one node $i=((id_{\reachtree},0), \ell, \valSet)$ for each original node $\rtmap{i}=(id_{\reachtree}, \ell, \valSet)\in N_{\reachtree}$, and (ii) $E$ consists of one edge $(i,e,\Prophecies{},j)\in E$ for each original edge $(\rtmap{i},e,\rtmap{j})\in E_{\reachtree}$ with $\Prophecies{}=\Prophecies{e}(\treePath_{\rtmap{i}})$.
Note that there is a one-to-one correspondence between the paths in $\reachtree$ and $\modtree$
after the initial transformation.
We apply further modifications, which are all \emph{path-maintaining}, i.e.,  each path in the resulting modified reach tree $\modtree$ 
has exactly one matching path in $\reachtree$, and for each path in $\reachtree$ there will be at least one matching path in $\modtree$.
 A formal definition is provided in
Appendix~\ref{app:modtree}.
We extend Definition~\ref{def:represents} by stating that a run $\run$ \emph{represented} by a path $\treePath$ in $\reachtree$ is \emph{represented} by all matching paths $\treePath'$ in $\modtree$, potentially containing additional stutter-edges  (see Appendix~\ref{app:modtree}).

After the initial transformation, the set of prophecies $\Prophecies{}$ annotating an edge $(i,e,\Prophecies{},j)\in E$ with $e\in \Edge$ is equal to the set of enabling prophecies $\Prophecies{e}(\treePath_i)$. 
We can hence use the annotated prophecies for both stutter-edges and edges with $e\in \Edge$
to determine the type of branching in a modified reach tree. 

We propose modifications on the modified reach tree which  may result in a restricted set of prophecies annotating the edges, while the excluded prophecies are 
maintained in another part of the tree, e.g., in a copy of some subtree.
We create a \emph{valid copy} of a 
\emph{subtree with root node $i$}
while \emph{restricting to a given subset of prophecies} $\Prophecies{}\subseteq \Prophecies{\RAC}$:
for %
$i=((id_{\reachtree},id_{\modtree}), \ell, \valSet)$ the copied subtree starts from the new node $i'=((id_{\reachtree},{id}'_{\modtree}), \ell, \valSet)$ with 
$id'_{\modtree}=\maxid{id_{\reachtree}}:= |\{i \in N \mid {id}_{\reachtree}^{i}={id}_{\reachtree}\}|$.
We then add copies of
(i) 
all nodes of the subtree  to $N$,
while iteratively increasing $id_{\modtree}$ %
s.t.
 $(id_{\reachtree},id_{\modtree})$ remains unique 
in $\modtree$, and 
(ii) 
all edges of the subtree to $E$, connecting the new nodes, while replacing annotated prophecies $\Prophecies{}'$ by  $\Prophecies{}''=\Prophecies{}' \cap \Prophecies{} $.

\subsection{Modifications of the Reach Tree}\label{subsec:exposesplit}

\paragraph{Exposing the Choice to Stay Forever}
\label{subsec:exposing}

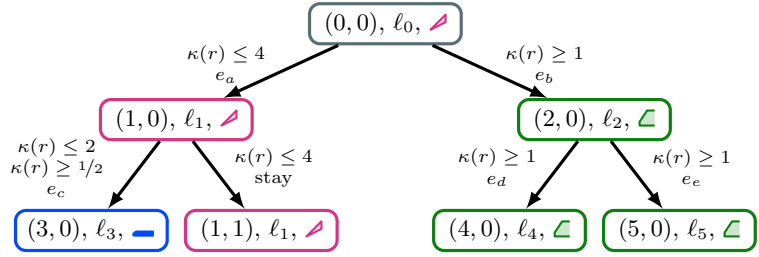
\begin{figure}[t]
    \centering
    \centering
        \begin{tikzpicture}[
	scale=1,
	node distance = 1cm and 0.8cm,
	baseline,
	remember picture,
	n/.style={draw, text width = 1.8cm, text = black, %
		align = center, font= \footnotesize, rounded corners, very thick, execute at begin node=\setlength{\baselineskip}{8pt}%
	},
	tn/.style={draw, text width = 0.2cm, text = black, %
		align = center, font= \footnotesize, circle, very thick, execute at begin node=\setlength{\baselineskip}{8pt}%
	},
	nd/.style={draw=none, text width = 2.5cm, anchor=left, text = black, %
		align = left, font= \footnotesize, rounded corners, very thick, execute at begin node=\setlength{\baselineskip}{8pt}%
	},
	en/.style={draw=none, minimum height=0cm, font = \scriptsize, align = center,outer sep=1mm},%
	ref/.style={draw, circle, minimum width=1.5mm, inner sep=0, fill,samplevaluescolor},%
	k/.style={draw=none, circle, minimum width=0, inner sep=0,outer sep=0, fill,samplevaluescolor},%
	c/.style={draw, fill, black, circle, inner sep=0, outer sep=0, minimum size=1mm},
	l/.style={anchor=west, inner sep=0, font=\footnotesize},
        secondlayerdistancesingle/.style={node distance =0.9cm and 0.1cm},
        secondlayerdistance/.style={node distance =0.4cm and 0.1cm},
        firstlayerdistance/.style={node distance =0.7cm and  1.75cm}
	]

    \begin{scope}
        
	\node[n, draw=cloc0](l0) at (0,0)  { $(0,0),\,\ell_0,\,\tikzboxpinkshape$};
	\node[n, draw=cloc1, firstlayerdistance,below left = of l0.south]  (l1) {$(1,0),\,\ell_1,\,\tikzboxpinkshape$}; 
	\node[n, draw=cloc3, firstlayerdistance, below right= of l0.south]  (l2) {$(2,0),\,\ell_2,\,\tikzboxgreenshape$};

	\node[n, draw=cloc2, secondlayerdistancesingle,below left = of l1.south]  (l3) {$(3,0),\,\ell_3,\,\tikzboxblueshape$}; 
	\node[n, draw=cloc1, secondlayerdistancesingle,below right = of l1.south]  (l1b) {$(1,1),\,\ell_1,\,\tikzboxpinkshape$}; 
    
	\node[n, draw=cloc3, secondlayerdistancesingle,below left = of l2.south]  (l4) {$(4,0),\,\ell_4,\,\tikzboxgreenshape$}; 
	\node[n, draw=cloc3, secondlayerdistancesingle, below right= of l2.south]  (l5) {$(5,0),\,\ell_5,\,\tikzboxgreenshape$};

\draw[-latex, very thick] (l0) to
     node[en, left,stochConflict,yshift=1mm] 
     {$\prophecy(r)\leq 4$ \\  $e_a$ } 
	(l1);
    
\draw[-latex, very thick] (l0) to
    node[en, right,stochConflict,yshift=1mm] 
     { $\prophecy(r)\geq 1$ \\ $e_b$ } 
	(l2);

    \draw[-latex, very thick] (l1) to
     node[en, left,stochConflict,yshift=1mm,xshift=-2mm] 
     {$\prophecy(r)\leq 2$ \\$\prophecy(r)\geq \nicefrac{1}{2}$ \\  $e_c$ } 
	(l3);

    \draw[-latex, very thick] (l1) to
     node[en, right,stochConflict,yshift=1mm] 
     {$\prophecy(r) \leq 4$  \\  stay } 
	(l1b);

    \draw[-latex, very thick] (l2) to
     node[en, left,stochConflict,yshift=1mm] 
     { $\prophecy(r)\geq 1$ \\ $e_d$ } 
	(l4);
    
\draw[-latex, very thick] (l2) to
    node[en, right,stochConflict,yshift=1mm] 
     {$\prophecy(r)\geq 1$ \\  $e_e$ } 
	(l5);
    
    \end{scope}

\end{tikzpicture}
        \caption{Running Example: Exposing the choice to stay in $\ell_1$.}\label{fig:runningexample-exposing}
\end{figure}

In any location $\ell$, w.r.t. the invariant,
a scheduler may choose to stay in $\ell$ \emph{forever} (or until reaching the time bound).
This choice is possible in nondeterministic and mixed branchings, if the invariant of $\ell$ does not restrict the scheduler's ability to reach $\tmax$ in this location 
(e.g., in $\ell_1$ of the running example with $\Inv(\ell_1)=\Reals^2$).
However, the choice of staying forever is not explicit in the reach tree.
The following definition makes such choices explicit in the reach tree by adding an absorbing node. We restrict the set of enabling prophecies for that node %
to all prophecies that enable the scheduler to stay until the time bound is reached. %

\begin{definition}[Exposing]
   Assume %
   node $i=((id_{\reachtree},id_{\modtree}), \ell, \valSet)\in N$ 
   with $\{\valCont \in \valSet \mid \valCont\proj{t} = \tmax \}\not = \emptyset$,
   and the branching at $i$ is either mixed, or nondeterministic. 
   We \emph{expose} the choice to stay in $i$ by modification of the sets $N$ and $E$:
    \begin{enumerate}[(i)]
        \item add a node $i_s =((id_{\reachtree},id'_{\modtree}), \ell, \valSet)$ with $id_{\modtree}'=\maxid{id_{\reachtree}}$, and
        \item add a stutter-edge  $(i,\tau,\Prophecies{},i_s)$ that connects the new node with $i$, with annotated set of prophecies $\Prophecies{}$ such that %
            (a)
            if $i$ is root, then $\Prophecies{}=\Prophecies{\RAC}$, and
            (b)
            if $(i',e,\Prophecies{}^i,i)\in E$ for some node $i'\in N$, $e\in \Edge\cup\{\tau\}$, then $\Prophecies{} \subseteq \Prophecies{}^i$%
            , s.t., for all $\prophecy\in \Prophecies{}^i$ it holds that  $\prophecy\in \Prophecies{}$ iff 
            there exists a scheduler $\scheduler_{\prophecy}$ and a run $\run$ represented by $\treePath_i$, s.t.
            $\last(\run)\semanticsarrowsac{\scheduler_{\prophecy}(\run)}(\ell, \valCont)$ and
            $\valCont\proj{t}=\tmax$.
    \end{enumerate}
\end{definition}
\noindent
Figure~\ref{fig:runningexample-exposing} provides an illustration of the exposing step at the running example's mixed branching at $((1,0),\ell_1,\tikzboxpinkshape)$. (Note that this illustration contains the node ids.) As it is possible to stay indefinitely in this node, we add a new absorbing node as child. As staying is not restricted, the enabling prophecies for this new node $((1,1),\ell_1,\tikzboxpinkshape)$ equal the enabling prophecies of $((1,0),\ell_1,\tikzboxpinkshape)$. %

\begin{figure}[t]%
    \centering
    \begin{subfigure}[b]{.44\linewidth}
    \centering
        \begin{tikzpicture}[
	scale=1,
	node distance = 1cm and 0.8cm,
	baseline,
	remember picture,
	n/.style={draw, text width = 1.8cm, text = black, %
		align = center, font= \footnotesize, rounded corners, very thick, execute at begin node=\setlength{\baselineskip}{8pt}%
	},
	tn/.style={draw, text width = 0.2cm, text = black, %
		align = center, font= \footnotesize, circle, very thick, execute at begin node=\setlength{\baselineskip}{8pt}%
	},
	nd/.style={draw=none, text width = 2.5cm, anchor=left, text = black, %
		align = left, font= \footnotesize, rounded corners, very thick, execute at begin node=\setlength{\baselineskip}{8pt}%
	},
	en/.style={draw=none, minimum height=0cm, font = \scriptsize, align = center,outer sep=1mm},%
	ref/.style={draw, circle, minimum width=1.5mm, inner sep=0, fill,samplevaluescolor},%
	k/.style={draw=none, circle, minimum width=0, inner sep=0,outer sep=0, fill,samplevaluescolor},%
	c/.style={draw, fill, black, circle, inner sep=0, outer sep=0, minimum size=1mm},
	l/.style={anchor=west, inner sep=0, font=\footnotesize},
        secondlayerdistancesingle/.style={node distance =0.9cm and 0.1cm},
        secondlayerdistance/.style={node distance =0.4cm and 0.1cm},
        firstlayerdistance/.style={node distance =0.7cm and  1.75cm}
	]

    \begin{scope}
        
	\node[n, draw=cloc0](l0) at (0,0)  { $(0,0),\,\ell_0,\,\tikzboxpinkshape$};
	\node[n, draw=cloc1, secondlayerdistancesingle,below left = of l0.south]  (l1) {$(1,0),\,\ell_1,\,\tikzboxpinkshape$}; 
	\node[n, draw=cloc3, secondlayerdistancesingle, below right= of l0.south]  (l2) {$(2,0),\,\ell_2,\,\tikzboxgreenshape$};

\draw[-latex, very thick] (l0) to
     node[en, left,stochConflict,yshift=1mm] 
     {$\prophecy(r)\leq 4$ \\  $e_a$ } 
	(l1);
    
\draw[-latex, very thick] (l0) to
    node[en, right,stochConflict,yshift=1mm] 
     { $\prophecy(r)\geq 1$ \\ $e_b$ } 
	(l2);

\node[draw=none,below = 0.0cm of l1.south] (dots) {$\vdots$};
\node[draw=none,below = 0.0cm of  l2.south] (dots) {$\vdots$};

    \end{scope}

\end{tikzpicture}
        \caption{Mixed branching at $\ell_0$.}\label{fig:splitting-a}
    \end{subfigure}\hfill
    \begin{subfigure}[b]{.54\linewidth}
    \centering
        \begin{tikzpicture}[
xscale=0.9,
yscale=1.25,
lineshiftup/.style={yshift=0.6pt},
lineshiftdown/.style={yshift=-0.6pt},
lineshiftdowntwice/.style={yshift=-1.8pt},
decoratenodebrace/.style={decorate,decoration={brace,amplitude=4pt,mirror},yshift=0pt, font=\scriptsize},
font= \footnotesize]

	\useasboundingbox (-0.2,-1.25) rectangle (6.4,0.7);

\begin{scope}[yshift=-0.3cm]
    \draw [thick,black,-latex] (0,0) -- (5,0);
    \node[fill=white] (x) at (5.75,0) {$\prophecy(r)$};
    \foreach \x in {0,1,2,...,4}  {
        \node[tick] (\x) at (\x,-0.25) {$\x$}; 
        \draw (\x,-0.05)--(\x,+0.05);
    }
    \foreach \x in {1,4}  {
        \draw[dotted,thick] (\x,0.02)--(\x,+0.8);
    }
\end{scope}

\begin{scope}[yshift=0cm]
   \node[draw=none, font=\footnotesize,cloc1] at (5.75,0) {$\Prophecies{}^1$};
    \draw [Bracket-,very thick,cloc1] (0,0) -- (1,0);
    \draw [-Bracket,very thick,cloc1,] (1,0) -- (4,0);
\end{scope}
\begin{scope}[yshift=0.3cm]
    \node[draw=none, font=\footnotesize,cloc3] at (5.75,0) {$\Prophecies{}^2$};
    \draw [Bracket-,very thick,cloc3,yshift=0mm, ,] (1,0) -- (4,0);
    \draw [very thick,cloc3,yshift=0mm] (4,0) -- (5,0);
\end{scope}

\draw [decoratenodebrace] (0.05,-0.7) -- node [below, yshift=-0.1cm, pos=.5,font=\footnotesize] {$e_a$} (0.95,-0.7);

\draw [decoratenodebrace] (1.05,-0.7) -- node [below, yshift=-0.1cm, pos=.5,font=\footnotesize] {$e_a$ or $e_b$} (3.95,-0.7);

\begin{scope}
\clip (4,-2) rectangle(5,2);
\draw [decoratenodebrace] (4.05,-0.7) to  node [below, yshift=-0.1cm, pos=.5,font=\footnotesize] {$e_b$} (5.2,-0.7);
\end{scope}

\end{tikzpicture}
        \caption{Prophecy sets $\Prophecies{}^1$, $\Prophecies{}^2$.}\label{fig:splitting-b}
    \end{subfigure}
    
    \begin{subfigure}[b]{1\linewidth}
    \centering
        \begin{tikzpicture}[
	scale=1,
	node distance = 1cm and 0.8cm,
	baseline,
	remember picture,
	n/.style={draw, text width = 1.8cm, text = black, %
		align = center, font= \footnotesize, rounded corners, very thick, execute at begin node=\setlength{\baselineskip}{8pt}%
	},
	tn/.style={draw, text width = 0.2cm, text = black, %
		align = center, font= \footnotesize, circle, very thick, execute at begin node=\setlength{\baselineskip}{8pt}%
	},
	nd/.style={draw=none, text width = 2.5cm, anchor=left, text = black, %
		align = left, font= \footnotesize, rounded corners, very thick, execute at begin node=\setlength{\baselineskip}{8pt}%
	},
	en/.style={draw=none, minimum height=0cm, font = \scriptsize, align = center,outer sep=1mm},%
	ref/.style={draw, circle, minimum width=1.5mm, inner sep=0, fill,samplevaluescolor},%
	k/.style={draw=none, circle, minimum width=0, inner sep=0,outer sep=0, fill,samplevaluescolor},%
	c/.style={draw, fill, black, circle, inner sep=0, outer sep=0, minimum size=1mm},
	l/.style={anchor=west, inner sep=0, font=\footnotesize},
        secondlayerdistancesingle/.style={node distance =0.9cm and 0.1cm},
        secondlayerdistance/.style={node distance =0.4cm and 0.1cm},
        firstlayerdistance/.style={node distance =0.7cm and  1.75cm}
	]

	\node[n, draw=cloc0](l0) at (0,0) {$(0,0),\,\ell_0,\,\tikzboxpinkshape$}; 
    
	\node[n, draw=cloc0,firstlayerdistance,below left = of l0.south](l01)  {$(0,1),\,\ell_0,\,\tikzboxpinkshape$}; 
	\node[n, draw=cloc0,firstlayerdistance,below  = of l0.south](l02)  {$(0,2),\,\ell_0,\,\tikzboxpinkshape$}; 
	\node[n, draw=cloc0,firstlayerdistance,below right = of l0.south](l03)  {$(0,3),\,\ell_0,\,\tikzboxpinkshape$}; 
	\node[n, draw=cloc1, secondlayerdistance,below left = of l01.south]  (l11) {$(1,1),\,\ell_1,\,\tikzboxpinkshape$}; 
	\node[n, draw=cloc1, secondlayerdistance,below left = of l02.south]  (l12) {$(1,2),\,\ell_1,\,\tikzboxpinkshape$}; 
	\node[n, draw=cloc3, secondlayerdistance, below right= of l02.south] (l21){$(2,1),\,\ell_2,\,\tikzboxgreenshape$};
	\node[n, draw=cloc3, secondlayerdistance, below right= of l03.south] (l22) {$(2,2),\,\ell_2,\,\tikzboxgreenshape$};

\node[draw=none,below = 0.0cm of  l11.south] (dots) {$\vdots$};
\node[draw=none,below = 0.0cm of  l12.south] (dots) {$\vdots$};
\node[draw=none,below = 0.0cm of  l21.south] (dots) {$\vdots$};
\node[draw=none,below = 0.0cm of  l22.south] (dots) {$\vdots$};

    \draw[-latex, very thick, ] (l0) to
    node[near start,en, left,stochConflict,yshift=+1mm] 
     {$\prophecy(r)\leq 1$ } 
	(l01.70);

    \draw[-latex, very thick, ] (l0) to
    node[,en, ,stochConflict,yshift=-0.5mm, xshift=-1mm, right, align=left]
     {$\prophecy(r)$\\$\in [1,4]$ } 
	(l02);
    
    \draw[-latex, very thick, ] (l0) to
    node[near start,en, right,stochConflict,yshift=+1mm] 
     {$\prophecy(r)\geq 4$ } 
	(l03.110);

    \draw[-latex, very thick] (l01) to   
    node[en, left,stochConflict,yshift=1mm] {$e_a$ } 
    (l11);
    \draw[-latex, very thick] (l02) to 
    node[en, left,stochConflict,yshift=1mm] {$e_a$ }  (l12);
    \draw[-latex, very thick] (l02) to
    node[en, right,stochConflict,yshift=1mm] {$e_b$ }   (l21);
    \draw[-latex, very thick] (l03) to
    node[en, right,stochConflict,yshift=1mm] {$e_b$ }  (l22);

\end{tikzpicture} 
        \caption{Modified reach tree after splitting.}\label{fig:splitting-c}
    \end{subfigure}
    \caption{Running Example: Splitting of a mixed branching at $\ell_0$.
    }
    \label{fig:illustrate-splitting}
\end{figure}
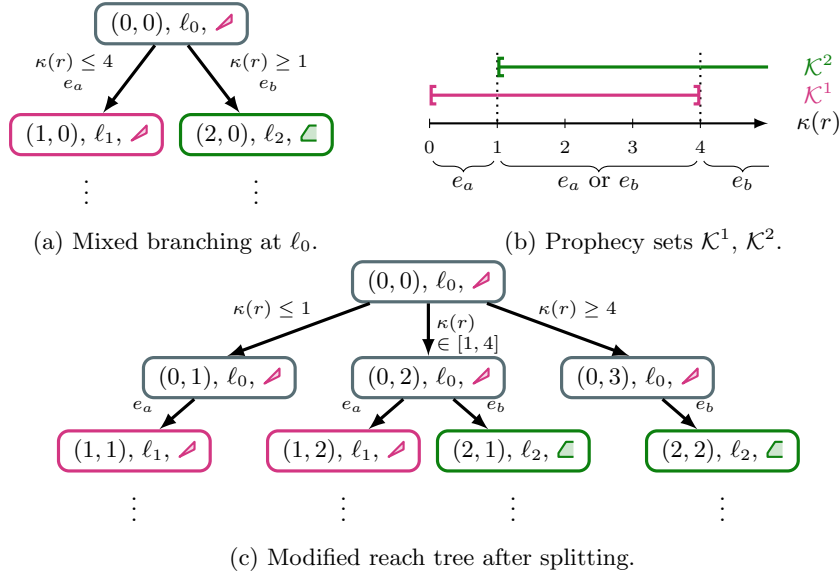

\paragraph{Splitting a Mixed Branching}\label{subsec:handling}

After exposing all choices to stay, mixed branchings require special handling, which first inserts an additional layer of nodes and then splits the sets of enabling prophecies leading to the nodes in the extra layer into disjoint sets.
Figure~\ref{fig:illustrate-splitting} illustrates  the concept of  one splitting step in the running example, which realizes the splitting of  enabling prophecies at node $((0,0),\,\ell_0,\,\tikzboxpinkshape)$, which has a \emph{mixed} branching.
In this example, sets of prophecies $\Prophecies{}^1=\{\prophecy\in \Prophecies{\RAC} \mid \prophecy(r)\in [0,4] \}$ and $\Prophecies{}^2=\{\prophecy\in \Prophecies{\RAC} \mid \prophecy(r)\in [1,\infty) \}$ are  split into three sets of prophecies, such that the resulting sets enable the same set of jumps, i.e., $\{e_a\},\{e_a,e_b\},\{e_a\}$. 
This is done by identifying the intersection of the sets of prophecies $\Prophecies{}^1$ and $\Prophecies{}^2$, as well as each set of prophecies without that intersection, as illustrated in Figure~\ref{fig:splitting-b} by the dotted lines.
For each of the resulting sets of prophecies, an additional child node  is added to $((0,0),\,\ell_0,\,\tikzboxpinkshape)$ in the modified tree, and the connecting edge is annotated with the restricted set of prophecies. 
Every new child node is then attached with 
a copy of those original child nodes of $((0,0),\,\ell_0,\,\tikzboxpinkshape)$ (and their subtrees), for which the corresponding jump is enabled by the restricted set of prophecies.
The resulting  modified reach tree  is shown in Figure~\ref{fig:splitting-c}.
Prophecies are split into sets of $\Prophecies{}^I$ such that the same subset of children $I \subseteq \validChildren(i)$ is reachable for all $\prophecy\in \Prophecies{}^I$:
\begin{definition}[Splitting]
For %
a mixed branching at node $i\allowbreak= \allowbreak((id_{\reachtree}, \allowbreak id_{\modtree}), \ell, \valSet)\in N$,
we \emph{split} the set of prophecies $\bigcup_{(i,e,\Prophecies{},j)\in E} \Prophecies{}$ into disjoint sets of prophecies $\Prophecies{}^I \in \mathbb{K} $ with $I\subseteq \validChildren(i)$, $I\not = \emptyset$, such that 
\begin{enumerate}[(i)]
    \item for all $\Prophecies{}^I\in \mathbb{K}$,  all $ \prophecy\in \Prophecies{}^I$ and all $ j \in I$ with $(i,e,\Prophecies{}^j,j)\in E$: $\prophecy\in \Prophecies{}^j$,
    \item for all $\smash{j \in \validChildren(i)}$ and all $ \prophecy\in \Prophecies{}^j$, it exists $ \smash{\Prophecies{}^I \in \mathbb{K}}$ s.t. $ \smash{\prophecy \in\Prophecies{}^I} $ and $\smash{ j \in I}$, 
    \item for all $\Prophecies{}^I\in \mathbb{K}$, $\Prophecies{}^I$ is not a null set,  and
    \item for all $\Prophecies{}^I,\Prophecies{}^{I'}\in \mathbb{K}$, $I\not = I'$,  $\Prophecies{}^I\cap\Prophecies{}^{I'}$ is a null set.
\end{enumerate}
We  realize \emph{splitting} in the tree by modifying sets $N$ and $E$. For each $\Prophecies{}^I\in \mathbb{K}$: 
    \begin{enumerate}[(i)]
        \item add one node $i_I=((id_{\reachtree},id'_{\modtree}), \ell, \valSet)$ with $id_{\modtree}'=\maxid{id_{\reachtree}}$,
        \item add one $\tau$-edge $(i,\tau,\Prophecies{}^I,i_I)$ connecting the new node $i_I$  annotated  with $\Prophecies{}^I$,
        \item for each $j\in I$ 
        add a valid copy of 
        the subtree with root $j$,
        starting from $j_I$, while restricting to prophecies $\Prophecies{}^I$,
        \item  add an edge $(i_I,e^j,\Prophecies{}',j_I) $ for all nodes $j_I$  with $(i,e^j,\Prophecies{}^j,j)\in E$ and 
            $\Prophecies{}'= \Prophecies{}^j \cap \Prophecies{}^I$ 
            that connects the new subtree with $i_I$,
    \end{enumerate}
    and for all $j\in \validChildren(i)$, $j\not = i_I$ for any $\Prophecies{}^I\in \mathbb{K}$, remove $(i,e,\Prophecies{},j)$ and $j$.
    
\end{definition}
Recall that $\children(i)$ may also contain child nodes that cannot be reached, as their enabling prophecies are null sets. Hence,  considering $\validChildren(i)$ for the sets $I$ is sufficient. %
Splitting a mixed branching might result in non-convex sets of enabling prophecies,
as the set of prophecies for a single jump is restricted by removing all
prophecies that also enable another jump (i.e., via set difference). %
Non-convex sets can be  represented in different ways; when using convex polytopes as state-space representation, one possibility is to use a union of
convex polytopes.
A series of splitting-operations  results in an exponential growth of the number of convex polytopes needed to represent a non-convex polytope.
A splitting operation does not introduce new mixed behavior; copies inherit it only if it already exists in the subtree, preserving the finiteness of the modified reach tree.

We remark that in mixed branchings with multiple edges, the overlap of the childrens' enabling prophecies can lead to many sets $I$. 
A \emph{separation} step can simplify the mixed case by specifying a winning random event  in an additional layer of nodes. 
This optional step is  formalized and illustrated in  Appendix~\ref{app:separating}.

\subsection{Minimum Reachability Probabilities}\label{subsec:fully}

We call a modified reach tree $\modtree$ 
\emph{fully modified}, %
if (i) initially,  it was created from $\reachtree$
via the  path-maintaining initial transformation,
(ii) it has only been changed by  {path-maintaining} modifications, e.g.,  \emph{splitting} or \emph{exposing}, 
(iii) all branchings in $\modtree$ are classified as \emph{stochastic} or \emph{nondeterministic} branchings and (iv)  all choices  to stay have been exposed.
Assume in the following  a fully modified reach tree $\fullymodtree=(N^{*},E^{*})$.
We can now define the set of prophecies enforcing to reach $\Goal$ by using intersection of the enabling prophecies in nondeterministic branchings, and by taking the union of enabling prophecies in stochastic branchings:
\begin{definition}[Enforcing set of prophecies]\label{def:enforcing} %
We define  the set $\Prophecies{}^{\Goal!}(i)$ of prophecies \emph{enforcing} a scheduler to reach $\Goal$ from node $i=((id_{\reachtree},id_{\modtree}),\ell, \valSet)\in N^{*}$.
\begin{align*}
\Prophecies{}^{\Goal!}(i) =
    \begin{cases}
        \bigcup_{j\in \validChildren(i)} \Prophecies{}^{\Goal!}(j), 
        &\text{if } \validChildren(i)\not = \emptyset \text{ and stochastic branching at } i\\
        \bigcap\nolimits_{j\in \validChildren(i)} \Prophecies{}^{\Goal!}(j),  
        &\text{if } \validChildren(i)\not = \emptyset \text{ and nondet. branching at } i,\\
        \Prophecies{}'
        & \text{if }\validChildren(i)= \emptyset,\ \ell\in\Goal, \text{ and } \exists (i',e,\Prophecies{}',i)\in E^{*}, \\
        \Prophecies{\RAC} 
        & \text{if }\validChildren(i)= \emptyset,\ \ell\in\Goal, \text{ and } i \text{ root of } \fullymodtree,\\
        \emptyset 
        & \text{otherwise.}
    \end{cases}
\end{align*}

\noindent We use $ \Prophecies{}^{\Goal!}$ for $\Prophecies{}^{\Goal!}(i)$ if $i$ is the root node of $\fullymodtree$.
\end{definition}

\begin{lemma}\label{lemma:kskg}
For each scheduler $\scheduler\in\SchedulersProphetic{\RAC}$, let $\PropheciesSchedulerMin{}$ be the set of all prophecies $\prophecy$ for which $\scheduler_{\prophecy}$ cannot avoid $\Goal$ in $\RAC$ within bounds $(\tmax,\jumpmax)$. Then, for all minimum schedulers $\scheduler\in\SchedulersProphMin{\RAC}$ it holds that
${\PropheciesSchedulerMin{}}
    = \Prophecies{}^{\Goal!}.$%
\end{lemma}

\begin{proof}
Assume a RAC \RAC, a time horizon $\tmax\in\Realspos$, and a jump bound $\jumpmax\in\Naturals$. Flowpipe construction~\cite{alur1995AlgorithmicAnalysisHybrid} computes a reach tree, which contains in its nodes the exact sets of all states of \RAC that are reachable within the given bounds. Note that the random clock values are explicitly stored in the states. By assumption, each random event occurs at most once and, once taken, the valuation of its clock remains constant thereafter. Thus, each leaf state encodes exactly the prophecies enabling its path under some scheduler.

Enabling sets of prophecies are  computed via projection onto the stochastic domain (i.e., extracting the prophecies from the states)  and  annotated  on the edges of the reach tree, yielding  a modified reach tree.
This initial transformation, as well as the operations exposing and splitting are path-maintaining; hence a given fully modified reach tree correctly contains all enabling sets of prophecies.

Consider a fully modified reach tree $\fullymodtree$ (without mixed  branchings) and a prophecy $\prophecy\in \Prophecies{\RAC}$. Removing nodes not leading to a $\prophecy$ leaf yields a tree with only nondeterministic branching, as stochastic branches already partition prophecies.

``$\subseteq$'': If  $\prophecy\in \PropheciesSchedulerMin{}$, a minimum scheduler $\scheduler_{\prophecy}$ is enforced to reach a goal state. %
That means that
all paths in $\fullymodtree$, collected over the nondeterministic branchings for $\prophecy$, end in a goal location.
For all those branchings, $\prophecy$ has to be in all enabling sets of prophecies, by definition of $\PropheciesSchedulerMin{}$.
Hence, by Definition~\ref{def:enforcing}, $\prophecy\in  \Prophecies{}^{\Goal!}$.

``$\supseteq$'': If $\prophecy\in  \Prophecies{}^{\Goal!}$, for all leaf nodes $i$ where $\prophecy$ is in the enabling set on the edge to $i$, $\ell^i\in \Goal$, and, for all nondeterministic branchings on all paths to those nodes, $\prophecy$ is in all enabling sets of prophecies and all forked paths end in goal locations. Hence, $\scheduler_{\prophecy}$ cannot avoid the goal and thus $\prophecy\in \PropheciesSchedulerMin{}$. \hfill \qed
\end{proof}

We compute the minimum reachability probability via the set of enforcing prophecies: 
$p_{\RAC}^\textsl{min}(\Goal, \tmax,\jumpmax) 
=  \int_{\Prophecies{}^{\Goal!}} G(\prophecy) \ d\prophecy $, given  RAC $\RAC$, set of goal locations $\Goal$, bounds $\tmax$ and $\jumpmax$ and the joint probability distribution for all random events $G(\prophecy)=\allowbreak \prod_{\clockr\in\LabRandom} \Distr(\clockr)(\prophecy(\clockr))$. %
The complexity of computing $\Prophecies{}^{\Goal!}$ is exponential in the worst case, as the reachability analysis is exponential in the jump depth, and splitting mixed branchings is exponential in the number of children of a branching, see Appendix~\ref{app:complexity} for details.
Given $\Prophecies{}^{\Goal!}$ and $G(\prophecy)$, as in \cite{journal}, multidimensional integration is used to yield the reachability probability.

We provide reachability probabilities for the running example in Table~\ref{tab:runningex}, for different goal locations. 
We compute the set of prophecies $\PropheciesSchedulerMin{}$ enforcing the scheduler to reach $\Goal$ by hand and integrate over the probability distribution.

\begin{table*}[tb]
    \centering
    \scriptsize
    \caption{
      Running Example: $p_{\RAC}^\textsl{min}$ and $p_{\RAC}^\textsl{max}$ for goal $\Goal$, $\tmax=5$, $\jumpmax=2$, $r\sim\mathcal{U}(0,5)$; with   prophecies $\PropheciesSchedulerMin{}$ (\emph{enforcing}, for $p_{\RAC}^\textsl{min}$) and $\PropheciesSchedulerMax{}$ (\emph{allowing}, for $p_{\RAC}^\textsl{max}$).
      }
    \label{tab:runningex}
    \newcolumntype{Y}{>{\centering\arraybackslash}X}
    \renewcommand{\arraystretch}{1}
    \begin{tabularx}{\linewidth}{cYYYYY}
        \toprule
          $\Goal$
          & $\ell_1$ 
          & $\ell_2$ 
          & $\ell_3$ 
          & $\ell_4$  
          & $\ell_5$   \\
          \midrule
          $\PropheciesSchedulerMin{} $
            & $\leq 0.5$
            & $\geq 4$ 
            & $\emptyset$ 
            & $\emptyset$
            & $\emptyset$
            \\ 
            $p_{\RAC}^\textsl{min}$
            & \probabminunishort{0.1} 
            & \probabminunishort{0.2} 
            & \probabminunishort{0.0} 
            & \probabminunishort{0.0} 
            & \probabminunishort{0.0}
            \\
        \midrule
          $\PropheciesSchedulerMax{} $
            & $\leq 4$
            & $\geq 1$
            & $\in [0.5,2]$   
            & $\geq 1$
            & $\geq 1$
            \\ 
            $p_{\RAC}^\textsl{max}$
            & \probabminunishort{0.8} 
            & \probabminunishort{0.8} 
            & \probabminunishort{0.3} 
            & \probabminunishort{0.8} 
            & \probabminunishort{0.8}
            \\
          \bottomrule
    \end{tabularx}
\end{table*}

For $\Goal=\{\ell_1\}$, the minimum scheduler is only forced to reach the goal if $\prophecy(r)\leq 0.5$, resulting in a minimum probability of $0.1$.
For  $\Goal=\{\ell_3\}$, enforcing prophecies are empty and $ p_{\RAC}^\textsl{min}(\{\ell_3\}, 5,2)=0$, since it is always possible to stay in location $\ell_1$ as exposed in Figure~\ref{fig:runningexample-exposing}. 
For $\Goal=\{\ell_2\}$, a scheduler can avoid that location entirely if $\prophecy(r)\leq 4$ as illustrated in the splitting in Figure~\ref{fig:illustrate-splitting}. Hence, the minimum probability to reach $\ell_2$ is $0.2$.
If the goal location is $\ell_4$, it is always possible to decide to go to location $\ell_5$. Following Definition~\ref{def:enforcing}, the set of enforcing prophecies to reach $\Goal=\{\ell_4\}$ is empty due to intersection with the empty set at the nondeterministic branching at $\ell_2$. Hence, the minimum reachability probability is $0$. This also holds analogously for $\Goal=\{\ell_5\}$.

\section{Experimental Results}
\label{sec:casestudy}
Our approach is implemented in \realyst\footnote{\realyst and model files 
are available through \url{go.uni-muenster.de/realyst}.}~\cite{realyst},  
which relies on \hypro for flowpipe construction~\cite{Schupp2017}. %
Numerical integration with Monte Carlo Vegas~\cite{lepage1978integration} yields reachability probabilities and provides a statistical error\footnote{The statistical error is reliable only if the number of samples is large enough for the distribution (whose mean equals the true integral) to be Gaussian~\cite{lepage1978integration}.}.
Evaluations are performed using an AMD Ryzen $9$ $5900$X with 12$\times$\SI{3.70}{\giga\hertz} and \SI{32}{\giga\byte} RAM.
We further provide an evaluation of a set of small examples 
in Appendix~\ref{app:smallexamples}.

We evaluate our approach on the existing CAR case study from~\cite{journal,tase} to showcase feasibility and limitations of the presented approach. %
The case study models an electrical car that is initially charging until a random event  ends \emph{charging} and starts \emph{driving}. 
During charging, the model can progress to different states (\texttt{A}$\rightarrow$\texttt{B}$\rightarrow$\texttt{C}$\rightarrow$\texttt{full}) with lower charging rates for an increased state of charge. %
Each charging state has a nondeterministic (rectangular) charging rate, and changing to the next charging state also includes time nondeterminism; together, this approximates the complex battery charging behavior.
While driving, either (i) the drive is ended (determined  by a random event), (ii) the battery is empty which stops the drive, or (iii) another random event determines a \emph{detour}, ultimately restarting the charging-and-driving process. 
Appendix~\ref{app:car} illustrates the RAC.

We compute the minimum probability to drain the battery, i.e., reaching goal location $\ell_{\texttt{empty}}$, which can happen while driving. %
Modeling zero or $1$ detour(s) explicitly, we consider $\car^{0}$ and $\car^{1}$ with up to $7$ resp. $15$ locations (c.f. $\lvert \Loc \rvert$ in Table~\ref{tab:car-esults}). %
We choose the same parameters as in~\cite{journal}:
Bounds are $\tmax=10$, $\jumpmax=7$ for $\car^{0}$,  
and $\tmax=20$, $\jumpmax=14$ for $\car^{1}$.
Random events occur more often with more detours and are distributed as follows: 
(i) the charging duration follows a folded normal distribution $\mathcal{N}(1.5,2)$,
(ii) the duration of the drive  follows  $\mathcal{N}(2,0.75)$, and
(iii) the random event that triggers a detour is distributed with $Exp(1)$.
The model considers three different charging variants \texttt{A}, \texttt{AB} and \texttt{ABC} (indicating the available charging states), %
and in addition to the \emph{rectangular} flow version also a variant where all flow rates are fixed to a constant (\emph{singular} variant). %
Integration uses \samples{1e5} resp. \samples{2e6} samples for $\car^{0}$, $\car^{1}$.

Results and computation times $\comptime$ for $p_{\RAC}^\textsl{min}$  for all variants are given in Table~\ref{tab:car-esults}.
As expected, the minimum reachability probability is always smaller than the maximum one.
Minimum reachability increases with more charging states, %
and decreases when allowing more detours, as the scheduler ultimately has  more options to avoid the goal.

\begin{table}[t]
    \centering
    \scriptsize
    \caption[Results CAR maximum]{
      $p_{\RAC}^\textsl{min}$, $p_{\RAC}^\textsl{max}$ for a (i) rectangular variant and (ii) singular variant of the \emph{CAR} model  for  $\Goal=\{\ell_\texttt{empty}\}$; %
    computation times $\comptime$ and error $\estat$ for $p_{\RAC}^\textsl{min}$.
      }
    \label{tab:car-esults}
    \newcolumntype{Y}{>{\centering\arraybackslash}X}
    \renewcommand{\arraystretch}{1}
    \begin{tabularx}{\linewidth}{p{0.5cm}p{0.5cm}YYYYYY}
        \toprule
         \multicolumn{2}{c}{\ }     &    \multicolumn{3}{c}{$\car^{0}$} &  \multicolumn{3}{c}{$\car^{1}$}  \\
      \cmidrule(lr){3-5}\cmidrule(lr){6-8}%
               \multicolumn{2}{p{1.05cm}}{variant}    & \texttt{A}& \texttt{AB}& \texttt{ABC}& \texttt{A} &\texttt{AB}& \texttt{ABC}\\%&  \texttt{A}   \\
        \multicolumn{2}{p{1.45cm}}{\parbox{1.45cm}{\scriptsize$\lvert\Loc\rvert, \lvert\reachtree\rvert, \lvert\fullymodtree\rvert$}}
        & \scriptsize$5,8,38$ 
        & \scriptsize $6,12,96$ 
        & \scriptsize $7,16,212$
        & \scriptsize$11,36,1632$ 
        & \scriptsize$13,88,$\tiny$\,{>}20k$  
        &\scriptsize$15,167,$\tiny$\,{>}250k$    \\
          \midrule
          \multirow{4}{*}{(i)} %
            & $p_{\RAC}^\textsl{min}$ 
            & \probab{0.1031432998097423} & \probab{0.132673005547679} & \probab{0.1331008708845563} & 
            \probab{0.1262443258551222} & \probab{0.1329856005144871} & \probab{0.1330104133334692} %
            \\
            & {$\comptime$}  
            & \rt{0.143914} & \rt{0.342713} & \rt{0.43265} 
            & \rt{59.0613} & \rt{68.659} & \rt{89.4327} 
            \\
            & $\estat$ 
            & \errortable{6.155165374595945e-05} & \errortable{0.0003681473895223651} & \errortable{7.566521100871698e-05} 
            & \errortable{0.0001228957873996281} & \errortable{0.0006802934059132314} & \errortable{0.0001800865579697933} %
            \\\cmidrule{2-8} 
            & $p_{\RAC}^\textsl{max}$  
            & \probab{0.3751797499988345} & \probab{0.4506640726652035} & \probab{0.4636189079148219} & 
            \probab{0.4904922873525423} & \probab{0.58259785817392} & \probab{0.6061182478991292} %
            \\
          \midrule
          \multirow{4}{*}{(ii)} %
            & $p_{\RAC}^\textsl{min}$  
            & \probab{0.1826180452547154} & \probab{0.2332422087540737} & \probab{0.2339228127660594} & 
            \probab{0.2480703272146985} & \probab{0.2564476194868545} & \probab{0.2345494681186517} %
            \\
            & {$\comptime$}  
            & \rt{0.142142} & \rt{0.382174} & \rt{0.465531} 
            & \rt{28.5423} & \rt{41.1311} & \rt{96.8459} 
            \\ 
            & $\estat$ 
            & \errortable{0.0001025686957639721} & \errortable{0.0005609190087067634} & \errortable{0.000139865812804895} & 
            \errortable{0.0002646173268515717} & \errortable{0.0002669765833851491} & \errortable{0.001003910508143656} %
            \\\cmidrule{2-8} 
            & $p_{\RAC}^\textsl{max}$  
            & \probab{0.2386263041785614} & \probab{0.3099455316487339} & \probab{0.3114187220156135} & 
            \probab{0.3006291343838571} & \probab{0.3962775381627917} & \probab{0.4018305818201353} %
            \\
          \bottomrule
    \end{tabularx}
\end{table}

In the minimum setting,
the impact of additional mixed branchings is significant for both the collection of prophecies and the integration (c.f. Figure~\ref{fig:new_plots}). 
$\vert \fullymodtree \vert$ in Table~\ref{tab:car-esults} serves as indicator for the  computational effort. 
All three charging locations induce \emph{mixed branchings}, as
the choice of the nonstochastic jump to switch charging location is subject to rate and time nondeterminism. 
Resolving mixed branchings requires splitting, so the number of nodes in $\fullymodtree$ increases even more for additional charging locations, which reflects in the computation time. 
Appendix~\ref{app:car}  further details the computation times.  %
Recall that splitting operations can result in non-convex polytopes, represented as union of convex polytopes.
For many mixed branchings, this  increases the computation time for integration as the containment check of Monte Carlo Vegas has to consider a larger number of polytopes. 
In summary, this case study highlights the challenges caused by mixed branchings, and the approach performs better on models with fewer mixed branchings, e.g. caused by concurrent random and nondeterministic time delays. Note, that the classification of branchings serves to identify  mixed branchings. It is unclear whether this can be done on a purely syntactical level.

Considering a $2$-detours-variant $\car^{2}$, reachability analysis already takes $\sim$\SI{190}{\second} for charging $\texttt{A}$ in the singular variant.
\realyst is able to compute maximum reachability probabilities in \SI{398}{\second}, while the minimum approach spends \SI{583}{\second} for branching classifications, and aborts after additional \SI{40}{\minute} due to a memory overflow during prophecy collection, caused by the amount of 
polytopes.
The rectangular variant for that setting already aborts due to a memory overflow during reachability analysis for both maximum and minimum after \SI{35}{\minute}.

\definecolor{color-time-forward}{named}{cpath1light}
\definecolor{color-time-processing}{named}{cpath2dark}
\definecolor{color-time-refinement}{named}{color2}
\definecolor{color-time-integration}{named}{amethyst}
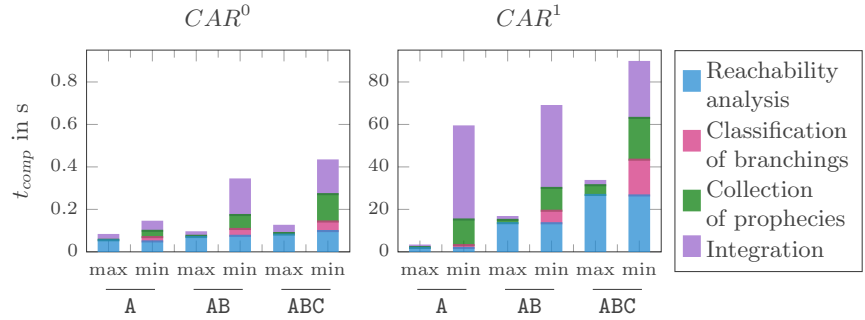
\begin{figure}[tb]
    \centering
    \pgfplotstableread[col sep=comma]{plots/data/new/minimum_results_maxmin_0.csv}\datatable%
\pgfplotsset{every tick label/.append style={font=\scriptsize}}
\begin{tikzpicture}[font=\sffamily, fill opacity = 0.75, opacity = 0.75]
    \begin{axis}[colormap/viridis,
    scale only axis=true,
    title=$CAR^{0}$,
        ybar stacked,
        ylabel={$\comptime$ in s}, 
        xticklabels={max,min,max,min,max,min},
         typeset ticklabels with strut,
        xtick=data,
        minor x tick num=1,
        enlarge x limits=0.1,
        legend pos=north west,
        grid style=dashed,
        bar width=0.26cm,
        ymin=0,        
        ymax=0.95,%
        width=3.5cm,
         height=0.22\textwidth,
     draw group line={chargingtype}{A}{\texttt{A}}{-0.5cm}{1pt},
     draw group line={chargingtype}{AB}{\texttt{AB}}{-0.5cm}{1pt},
     draw group line={chargingtype}{ABC}{\texttt{ABC}}{-0.5cm}{1pt}
    ]

    \addplot[color=color-time-forward, fill, thick, name path=A] table[x=table_index, y=t_flowpipe, col sep=comma] \datatable;
    \addplot[color=color-time-refinement,fill,  thick, name path=A] table[x=table_index, y=t_branchings, col sep=comma] \datatable;
    \addplot[color=color-time-processing,fill,  thick, name path=A] table[x=table_index, y=t_collection, col sep=comma] \datatable;
    \addplot[color=color-time-integration, fill, thick, name path=A] table[x=table_index, y=t_integration, col sep=comma] \datatable;
    \end{axis}

\end{tikzpicture}%
\hspace{-0.4em}
    \pgfplotstableread[col sep=comma]{plots/data/new/minimum_results_maxmin_1.csv}\datatable%
\pgfplotsset{every tick label/.append style={font=\scriptsize}}
\begin{tikzpicture}[font=\sffamily, fill opacity = 0.75, opacity = 0.75]
    \begin{axis}[colormap/viridis,
    scale only axis=true,
    title=$CAR^{1}$,
        ybar stacked,
        xticklabels={max,min,max,min,max,min},
         typeset ticklabels with strut,
        xtick=data,
        minor x tick num=1,
        enlarge x limits=0.1,
        legend pos=north west,
        grid style=dashed,
         bar width=0.26cm,
        ymin=0,        
        ymax=95,%
        width=3.5cm,
        height=0.22\textwidth,
        legend style={
      font=\footnotesize,
      cells={anchor=west, align=left},
      legend columns=1, 
      at={(1.05,1)}, %
      anchor=north west,
    },
     draw group line={chargingtype}{A}{\texttt{A}}{-0.5cm}{1pt},
     draw group line={chargingtype}{AB}{\texttt{AB}}{-0.5cm}{1pt},
     draw group line={chargingtype}{ABC}{\texttt{ABC}}{-0.5cm}{1pt}
    ]

    \addplot[color=color-time-forward, fill, thick, name path=A] table[x=table_index, y=t_flowpipe, col sep=comma] \datatable;
    \addplot[color=color-time-refinement,fill,  thick, name path=A] table[x=table_index, y=t_branchings, col sep=comma] \datatable;
    \addplot[color=color-time-processing,fill,  thick, name path=A] table[x=table_index, y=t_collection, col sep=comma] \datatable;
    \addplot[color=color-time-integration, fill, thick, name path=A] table[x=table_index, y=t_integration, col sep=comma] \datatable;
    \legend{{Reachability\\analysis}, {Classification\\ of branchings},{Collection\\ of prophecies}, Integration}
    \end{axis}

\end{tikzpicture}
\vspace*{-2ex}
    \caption[Computation times]{
    Computation times for the rectangular variant %
    of $\car^{0}$ and $\car^{1}$. }%
    \label{fig:new_plots}
\end{figure}

\section{Conclusion}
\label{sec:conclusion}
The proposed approach to compute \emph{minimum reachability probabilities} complements the computation of maximum reachability probabilities~\cite{journal,tase,25StuebbeSYNASC} for RAC.
Together, we estimate the impact of epistemic uncertainty modeled via continuous and discrete nondeterminism in stochastic hybrid systems with random events.
Our approach relies on flowpipe construction and a classification of branchings in the reach tree, which is needed in the minimum setting.
In future work, we plan to consider minimum scheduler synthesis, which introduces even more challenges and requires an adaption of the \emph{backward refinement} from~\cite{journal,tase}.

\subsubsection{Acknowledgement.}
The need to classify branchings for the computation of minimum reachability probabilities has first been mentioned by Annabell Petri \cite{annabell}. 

\subsubsection{Data Availability.} An artifact containing our implementation in \realyst, model files for the car case study and the small examples, and scripts that reproduce our experiments and the results presented in Section~\ref{sec:casestudy} and Appendix~\ref{app:smallexamples} is archived and available at \url{https://doi.org/10.5281/zenodo.18990720}~\cite{artifact}.

\bibliographystyle{splncs04}
\bibliography{references}

\begin{thebibliography}{10}
\providecommand{\url}[1]{\texttt{#1}}
\providecommand{\urlprefix}{URL }
\providecommand{\doi}[1]{https://doi.org/#1}

\bibitem{alur1995AlgorithmicAnalysisHybrid}
Alur, R., Courcoubetis, C.A., Halbwachs, N., Henzinger, T.A., Ho, P., Nicollin,
  X., Olivero, A., Sifakis, J., Yovine, S.: The algorithmic analysis of hybrid
  systems. Theoretical Computer Science  \textbf{138}(1),  3--34 (1995).
  \doi{10.1016/0304-3975(94)00202-T}

\bibitem{ballarini2013TransientAnalysisNetworks}
Ballarini, P., Bertrand, N., Horv{\'a}th, A., Paolieri, M., Vicario, E.:
  Transient analysis of networks of stochastic timed automata using stochastic
  state classes. In: Proc. of the 10th Int. Conf. on Quantitative Evaluation of
  Systems. LNCS, vol.~8054, pp. 355--371. Springer (2013).
  \doi{10.1007/978-3-642-40196-1\_30}

\bibitem{bertrand2014StochasticTimedAutomata}
Bertrand, N., Bouyer, P., Brihaye, T., Menet, Q., Baier, C., Gr{\"o}{\ss}er,
  M., Jurdzinski, M.: Stochastic timed automata. Logical Methods in Computer
  Science  \textbf{10}(4),  1--73 (2014). \doi{10.2168/LMCS-10(4:6)2014}

\bibitem{bohnenkamp2006MODESTCompositionalModeling}
Bohnenkamp, H., D'Argenio, P.R., Hermanns, H., Katoen, J.P.: {MODEST}: {A}
  compositional modeling formalism for hard and softly timed systems. IEEE
  Transactions on Software Engineering  \textbf{32}(10),  812--830 (2006).
  \doi{10.1109/TSE.2006.104}

\bibitem{Bouyer2009}
Bouyer, P., Forejt, V.: Reachability in stochastic timed games. In: Albers, S.,
  Marchetti{-}Spaccamela, A., Matias, Y., Nikoletseas, S.E., Thomas, W. (eds.)
  Automata, Languages and Programming, 36th Int. Colloquium. LNCS, vol.~5556,
  pp. 103--114. Springer (2009). \doi{10.1007/978-3-642-02930-1\_9}

\bibitem{cauchi2019StocHyAutomatedVerification}
Cauchi, N., Abate, A.: {StocHy}: {A}utomated verification and synthesis of
  stochastic processes. In: Proc. of the 25th Int. Conf. on Tools and
  Algorithms for the Construction and Analysis of Systems (TACAS'19). pp.
  247--264. Springer (2019). \doi{10.1145/3302504.3313349}

\bibitem{dargenio2018HierarchySchedulerClasses}
D'Argenio, P.R., Gerhold, M., Hartmanns, A., Sedwards, S.: A {Hierarchy} of
  {Scheduler} {Classes} for {Stochastic} {Automata}. In: Proc. of 21st {Int.}
  {Conference} on {Foundations} of {Software} {Science} and {Computation}
  {Structures} (FOSSACS 2018). {LNCS}, vol. 10803, pp. 384--402. Springer
  (2018). \doi{10.1007/978-3-319-89366-2\_21}

\bibitem{journal}
Delicaris, J., Remke, A., {\'{A}}brah{\'{a}}m, E., Schupp, S., St{\"{u}}bbe,
  J.: Maximizing reachability probabilities in rectangular automata with random
  events. Sci. Comput. Program.  \textbf{240},  103213 (2025).
  \doi{10.1016/J.SCICO.2024.103213}

\bibitem{tase}
Delicaris, J., Schupp, S., {\'{A}}brah{\'{a}}m, E., Remke, A.: Maximizing
  reachability probabilities in rectangular automata with random clocks. In:
  Proc. of the 17th Int. Symp. on Theoretical Aspects of Software Engineering
  (TASE'23). LNCS, vol. 13931, pp. 164--182. Springer (2023).
  \doi{10.1007/978-3-031-35257-7\_10}

\bibitem{realyst}
Delicaris, J., St{\"{u}}bbe, J., Schupp, S., Remke, A.: {RealySt}: {A} {C++}
  tool for optimizing reachability probabilities in stochastic hybrid systems.
  In: Proc. of the 16th EAI Int. Conf. on Performance Evaluation Methodologies
  and Tools. LNCS, vol.~539, pp. 170--182. Springer (2023).
  \doi{10.1007/978-3-031-48885-6\_11}

\bibitem{artifact}
Delicaris, J., Ábrahám, E., Remke, A.: Artifact for: Minimum reachability
  probabilities in rectangular automata with random clocks (Mar 2026).
  \doi{10.5281/zenodo.18990721}

\bibitem{hahn2013CompositionalModellingAnalysis}
Hahn, E.M., Hartmanns, A., Hermanns, H., Katoen, J.P.: A compositional
  modelling and analysis framework for stochastic hybrid systems. Formal
  Methods in System Design  \textbf{43}(2),  191--232 (2013).
  \doi{10.1007/s10703-012-0167-z}

\bibitem{henzinger1998WhatDecidableHybrid}
Henzinger, T.A., Kopke, P.W., Puri, A., Varaiya, P.: What's decidable about
  hybrid automata? Journal of Computer and System Sciences  \textbf{57}(1),
  94--124 (1998). \doi{10.1006/jcss.1998.1581}

\bibitem{koutsoukos2008ComputationalMethodsVerification}
Koutsoukos, X.D., Riley, D.: Computational methods for verification of
  stochastic hybrid systems. IEEE Transactions on Systems, Man, and Cybernetics
  - Part A: Systems and Humans  \textbf{38}(2),  385--396 (2008).
  \doi{10.1109/TSMCA.2007.914777}

\bibitem{kwiatkowska2000VerifyingQuantitativeProperties}
Kwiatkowska, M.Z., Norman, G., Segala, R., Sproston, J.: Verifying quantitative
  properties of continuous probabilistic timed automata. In: Proc. of the 11th
  Int. Conf. on Concurrency Theory. LNCS, vol.~1877, pp. 123--137. Springer
  (2000). \doi{10.1007/3-540-44618-4\_11}

\bibitem{KWIATKOWSKA20071027}
Kwiatkowska, M., Norman, G., Sproston, J., Wang, F.: Symbolic model checking
  for probabilistic timed automata. Information and Computation
  \textbf{205}(7),  1027--1077 (2007).
  \doi{https://doi.org/10.1016/j.ic.2007.01.004}

\bibitem{lepage1978integration}
Lepage, G.P.: A new algorithm for adaptive multidimensional integration.
  Journal of Computational Physics  \textbf{27}(2),  192--203 (1978).
  \doi{10.1016/0021-9991(78)90004-9}

\bibitem{MathisDiss}
Niehage, M.: SMC-Based Learning meets Simulation on a Symbolic State-Space in
  Hybrid Petri Nets. Ph.D. thesis, University of Münster, Münster, Germany
  (2025)

\bibitem{MathisMemocode}
Niehage, M., Hartmanns, A., Remke, A.: Learning optimal decisions for
  stochastic hybrid systems. In: 19th {ACM-IEEE} International Conference on
  Formal Methods and Models for System Design. pp. 44--55. {ACM} (2021).
  \doi{10.1145/3487212.3487339}

\bibitem{annabell}
Petri, A.: Computing Minimal Reachability Probabilities for Rectangular Hybrid
  Automata with Random Clocks through Backwards Partition Refinement. Master's
  thesis, University of Münster, Münster, Germany

\bibitem{Pilch2021Optimizing}
Pilch, C., Schupp, S., Remke, A.: Optimizing reachability probabilities for a
  restricted class of stochastic hybrid automata via flowpipe-construction. In:
  Proc. of the 18th Int. Conference on Quantitative Evaluation of Systems
  (QEST'21). pp. 435--456 (2021). \doi{10.1007/978-3-030-85172-9\_23}

\bibitem{transformationNFM}
Pilch, C., Krause, M., Remke, A., {\'{A}}brah{\'{a}}m, E.: A transformation of
  hybrid petri nets with stochastic firings into a subclass of stochastic
  hybrid automata. In: 12th International Symposium on {NASA} Formal Methods.
  pp. 381--400. LNCS, Springer (2020). \doi{10.1007/978-3-030-55754-6\_23}

\bibitem{prandini2006StochasticApproximationMethod}
Prandini, M., Hu, J.: A stochastic approximation method for reachability
  computations. In: Stochastic Hybrid Systems: Theory and Safety Critical
  Applications, LNCS, vol.~337, pp. 107--139. Springer (2006).
  \doi{10.1007/11587392\_4}

\bibitem{Schupp2017}
Schupp, S., {\'A}brah{\'a}m, E., Makhlouf, I.B., Kowalewski, S.: {HyPro}: {A}
  {C}++ library of state set representations for hybrid systems reachability
  analysis. In: Proc. of NFM'17, vol. 10227, pp. 288--294. Springer (2017).
  \doi{10.1007/978-3-319-57288-8\_20}

\bibitem{carina}
da~Silva, C., Schupp, S., Remke, A.: Optimizing reachability probabilities for
  a restricted class of stochastic hybrid automata via flowpipe construction.
  ACM Trans. Model. Comput. Simul.  \textbf{33}(4),  18:1--18:27 (2023).
  \doi{10.1145/3607197}

\bibitem{soudjani2015FAUST2FormalAbstractions}
Soudjani, S.E.Z., Gevaerts, C., Abate, A.: {FAUST2}: {F}ormal {A}bstractions of
  {U}ncountable-{ST}ate {ST}ochastic processes. In: Proc. of the 21st Int.
  Conf. on Tools and Algorithms for the Construction and Analysis of Systems
  (TACAS'15). LNCS, vol.~9035, pp. 272--286. Springer (2015).
  \doi{10.1007/978-3-662-46681-0\_23}

\bibitem{sproston2000DecidableModelChecking}
Sproston, J.: Decidable model checking of probabilistic hybrid automata. In:
  Proc. of the 6th Int. Symp. on Formal Techniques in Real-Time and
  Fault-Tolerant Systems (FTRTFT'00). LNCS, vol.~1926, pp. 31--45. Springer
  (2000). \doi{10.1007/3-540-45352-0\_5}

\bibitem{sproston2019Verification}
Sproston, J.: Verification and control for probabilistic hybrid automata with
  finite bisimulations. J. Log. Algebraic Methods Program.  \textbf{103},
  46--61 (2019). \doi{10.1016/j.jlamp.2018.11.001}

\bibitem{25StuebbeSYNASC}
Stübbe, J., Remke, A., Ábrahám, E.: Scaling up reachability analysis for
  rectangular automata with random clocks. In: 27th Int. Symposium on Symbolic
  and Numeric Algorithms for Scientific Computing (SYNASC 2025). pp. 21--29
  (2025). \doi{10.1109/SYNASC69064.2025.00011}

\end{thebibliography}

\begin{appendix}
\renewcommand{\thesection}{\Alph{section}}
\FloatBarrier
\section{Appendix}\label{sec:appendix}

\subsection{Reach Tree}
\label{app:reachtree}
\begin{definition}[Reach tree]
For a location $\ell\in\Loc$  with $\Init(\ell)\not=\emptyset$,
the \emph{$(\tmax,\jumpmax)$-bounded reach tree for $\RAC$ from $\ell$ to $\Goal$} is an annotated tree $\reachtree = (N,E)$ with a nonempty finite set of \emph{nodes} $N\subseteq \Naturals\times\Loc\times2^{\Reals^{\dimCont}}$ and a set of \emph{edges} $ E = N \times  \Edge \times N$, 
    satisfying the following:
    \begin{itemize}

\item for each $i_1=(\textit{id}_1,\ell_1,\valSet_1)\in N$ and $i_2=(\textit{id}_2,\ell_2,\valSet_2)\in N\setminus\{i_1\}$ we have $\textit{id}_1\not=\textit{id}_2$, and $(i_1,e,i_2)\in E$ implies $\textit{id}_1<\textit{id}_2$ for all $e\in\Edge$;

\item there is a \emph{root} node $\textit{root}=(\textit{id}_0,\ell_0,\valSet_0)\in N$ with $\ell_0=\ell$, $\valSet_0=\ftc{\ell}(\Init(\ell))$, and $\textit{root}\not=i_2$ for all $(i_1,e,i_2)\in E$; 

\item each $i_2\in N\setminus\{\textit{root}\}$ has a unique $\textit{parent}(i_2)=i_1\in N$ with $(i_1,e,i_2)\in E$ for one $e\in\Edge$; let $\depth(\textit{root})=0$ and $\depth(i)=\depth(\textit{parent}(i))+1$ for all $i\in N\setminus\{\textit{root}\}$, and let $\textit{children}(i_1)=\{i_2\in N\,|\,\exists e\in\Edge.\ (i_1,e,i_2)\in E\}$;
\item for each $i_1=(\textit{id}_1,\ell_1,\valSet_1)\in N$:
(i) if 
either 
$\textit{depth}(i_1)=\jumpmax$ 
or $\ell_1 \in \Goal$
then $\textit{children}(i_1)=\emptyset$,
(ii) otherwise for each $e\in\Edge$:
(a) if $\source(e)=\ell_1$ and $\valSet_2 = \ftc{\target(e)}(\fosr{e}(\valSet_1))\not=\emptyset$ then there exists $(i_1,e,i_2)\in E$ with $i_2=(\textit{id}_2,\allowbreak\target(e),\valSet_2)$ and $e\not=e'$ for all $(i_1,e',i_2')\in E\setminus\{(i_1,e,i_2)\}$, and
(b) otherwise $e\not=e'$ for all $(i_1,e',i_2)\in E$.
\end{itemize} 
\end{definition}

\FloatBarrier

\subsection{Path-maintaining operations}\label{app:modtree}

\begin{definition}[Path-maintaining]
Assume $\reachtree$ is a reach tree and $\modtree_{\reachtree}=(N,E)$ a modified reachtree.
We call an operation, altering the sets $N$ and $E$ of $\modtree_{\reachtree}$, to be \emph{path-maintaining}, if before and after alteration it holds that %
    \begin{enumerate}[(i)]
    \item%
    for each
path $\treePath=i_0,\dots,i_n$ in the modified reach tree $\modtree_{\reachtree}$, 
there exists exactly one matching path $\treePath'=i'_0, \dots, i'_{m}$, $m\leq n$ without stutter-edges in $\reachtree$ s.t.
for each two nodes $i_{k}, i_{k+1}\in N$ in $\treePath$ with $(i_k,e,\Prophecies{},i_{k+1})\in E$, either 
\begin{enumerate}[(a)]
    \item 
    there exist nodes $i'_{k'}, i'_{k'+1}\in N_{\reachtree}$ in $\treePath'$ with $\rtmap{i_k}=i'_{k'}$, $\rtmap{i_{k+1}}=i'_{k'+1}$ and  $(i'_{k'},e,i'_{k'+1})\in E_{\reachtree}$, or 
    \item
    the edge between $i_{k}$ and $ i_{k+1}$ is a \emph{stutter-edge}, i.e., there exists one node $i'_k\in N_{\reachtree}$ and $\treePath'$ with $\rtmap{i_k}=\rtmap{i_{k+1}}=i'_k$ and $e=\tau$, %
\end{enumerate}
    \item%
and for each path $\treePath=i_0,\dots,i_n$ in $\reachtree$, there is at least one matching path $\treePath'=i'_0, \dots, i'_{m}$, $n\leq m$ in $\modtree_{\reachtree}$, potentially with additional stutter-edges, such that
for each two nodes $i_{k}, i_{k+1}\in N_{\reachtree}$ in $\treePath$ with $(i_k,e,i_{k+1})\in E_{\reachtree}$, there exist nodes $i'_{k'}, i'_{k''}\in N$, $k' < k''$ in $\treePath'$ 
with $\rtmap{i'_{k'}}=i_k$, $\rtmap{i'_{k''}}=i_{k+1}$
and either 
\begin{enumerate}[(a)]
    \item 
    $(i'_{k'},e,\Prophecies{},i'_{k'+1})\in E$  for some $\Prophecies{}\subseteq \Prophecies{\RAC}$ and $k''=k'+1$, or
   \item 
   there exists nodes $j_0,\dots,j_l\in N$ with $\rtmap{j_0}=\dots=\rtmap{j_{l'}}=\rtmap{i'_{k'}}=i_k$ and $\rtmap{j_{l'+1}}=\dots=\rtmap{j_l}=\rtmap{i'_{k''}}=i_{k+1}$, $0\leq l'< l$ and $(j_{l'},e,\Prophecies{},j_{l'+1})\in E$ for some $\Prophecies{}\subseteq \Prophecies{\RAC}$. 
\end{enumerate}
\end{enumerate}
\end{definition}

\subsection{Separating Mixed Branchings}\label{app:separating}

Competing random events can be \emph{separated} by an additional layer of nodes, where each node specifies a winning random event. Then, copies of the original child nodes  (and their subtrees) are attached to every node of the separating layer and the prophecies are restricted according to the winning event.

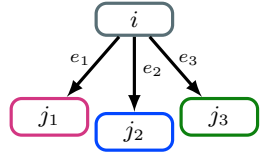
\begin{figure}[tbh]
    \centering
    \begin{tikzpicture}[
	scale=1,
	node distance = 1cm and 0.8cm,
	baseline,
	remember picture,
	n/.style={draw, text width = 0.8cm, text = black, %
		align = center, font= \footnotesize, rounded corners, very thick, execute at begin node=\setlength{\baselineskip}{8pt}%
	},
	tn/.style={draw, text width = 0.2cm, text = black, %
		align = center, font= \footnotesize, circle, very thick, execute at begin node=\setlength{\baselineskip}{8pt}%
	},
	nd/.style={draw=none, text width = 2.5cm, anchor=left, text = black, %
		align = left, font= \footnotesize, rounded corners, very thick, execute at begin node=\setlength{\baselineskip}{8pt}%
	},
	en/.style={draw=none, minimum height=0cm, font = \scriptsize, align = center,outer sep=1mm},%
	ref/.style={draw, circle, minimum width=1.5mm, inner sep=0, fill,samplevaluescolor},%
	k/.style={draw=none, circle, minimum width=0, inner sep=0,outer sep=0, fill,samplevaluescolor},%
	c/.style={draw, fill, black, circle, inner sep=0, outer sep=0, minimum size=1mm},
	l/.style={anchor=west, inner sep=0, font=\footnotesize},
        secondlayerdistance/.style={node distance =0.8cm and 0.6cm},
        firstlayerdistance/.style={node distance =1.4cm and  1.75cm}
	]

    \begin{scope}
        
	\node[n, draw=cloc0](l0) at (0,0)  { $i$ };
	\node[n, draw=cloc1, secondlayerdistance,below left = of l0.south]  (l1) {$j_1$}; %
	\node[n, draw=cloc2, secondlayerdistance, below = of l0.south,yshift=-2mm] (l2) {$j_2$}; %
	\node[n, draw=cloc3, secondlayerdistance, below right= of l0.south] (l3) {$j_3$}; %

\draw[-latex, very thick] (l0) to
    node[near start,en, left,nonstochConflict,yshift=-1mm] {$e_1$}  
	(l1);
    
\draw[-latex, very thick] (l0) to
    node[near start,en, right,nonstochConflict,yshift=-2mm,xshift=-1mm] {$e_2$}  
	(l2);
    
\draw[-latex, very thick] (l0) to
    node[near start,en, right,nonstochConflict,yshift=-1mm] {$e_3$}  
	(l3);
    
    \end{scope}

\end{tikzpicture}
    \caption{
    Mixed branching at node $i$ with three children $j_1,j_2,j_3$, where $j_1$ and $j_3$ are reached by competing random events $r$ and $q$.
    } 
    \label{fig:illustrate-separating}
\end{figure}

\begin{figure*}[t]
\centering
    \begin{subfigure}{0.47\textwidth}
    \centering
        \begin{tikzpicture}[
scale=0.95,
lineshiftup/.style={yshift=0.6pt},
lineshiftdown/.style={yshift=-0.6pt},
lineshiftdowntwice/.style={yshift=-1.8pt},
decoratenodebrace/.style={decorate,decoration={brace,amplitude=4pt,mirror},yshift=0pt, font=\scriptsize},
font= \footnotesize,
prophecyarea/.style={draw=none, fill, fill opacity = 0.3},]

	\useasboundingbox (-0.3,-0.3) rectangle (4.5,4.5);

    \node[fill=white,font=\scriptsize] (x) at (4.3,-0.1) {$\prophecy(r)$};
    \node[fill=white,font=\scriptsize,yshift=-0.8mm] (y) at (-0.1,4.3) {$\prophecy(q)$};
    \draw [thick,black,-latex] (0,0) -- (4,0);
    \draw [thick,black,-latex] (0,0) -- (0,4);
    \foreach \x in {1,2,...,3}  {
        \node[tick] (\x) at (\x,-0.25) {$\x$}; 
        \draw (\x,-0.05)--(\x,+0.05);
    }
    \foreach \y in {1,2,...,3}  {
        \node[tick] (\y) at (-0.25,\y) {$\y$}; 
        \draw (-0.05,\y)--(0.05,\y);
    }

\draw[prophecyarea, cloc3] (0,0) -- (3,3) -- (3.75,3) -- (3.75,0) -- cycle;
\draw[prophecyarea, cloc2] (1,1) -- (3.75,1) -- (3.75,3.75) -- (1,3.75) -- cycle;
\draw[prophecyarea, cloc1] (0,0) -- (3,3) -- (3,3.75) -- (0,3.75) -- cycle;

\draw[dashed, thick] (0,0)--(4,4);
\draw[dotted, thick] (1,1)--(1,3.8);
\draw[dotted, thick] (1,1)--(3.8,1);
\draw[dotted, thick] (3,3)--(3,3.8);
\draw[dotted, thick] (3,3)--(3.8,3);

\node[draw=none, font=\scriptsize,cloc3,xshift=1mm] at (4.2,0.75) {$\Prophecies{}^3$};
\node[draw=none, font=\scriptsize,cloc2,xshift=0mm] at (4.2,3.4) {$\Prophecies{}^2$};
\node[draw=none, font=\scriptsize,cloc1] at (0.75,4.2) {$\Prophecies{}^1$};

\begin{scope}[yshift=4cm]
    \draw [very thick,cloc1] (0.05,0) -- (3,0);
\end{scope}
\begin{scope}[yshift=3.9cm]
    \draw [very thick,cloc2] (1,0) -- (3.75,0);
\end{scope}

\begin{scope}[xshift=4cm]
    \draw [very thick,cloc3] (0,0.05) -- (0,3);
\end{scope}
\begin{scope}[xshift=3.9cm]
    \draw [very thick,cloc2] (0,1) -- (0,3.75);
\end{scope}

\node[draw=none, font=\scriptsize] (j2) at (3.5,3.2) {$j_2$};
\node[draw=none, font=\scriptsize,anchor=east] at ($(j2.east)+(0,-1.2)$) {$j_2,j_3$};
\node[draw=none, font=\scriptsize,anchor=east] at ($(j2.east)+(0,-2.7)$) {$j_3$};

\node[draw=none, font=\scriptsize] at (0.5,3.5) {$j_1$};
\node[draw=none, font=\scriptsize] at (2,3.5) {$j_1,j_2$};
\node[draw=none, font=\scriptsize] at (3.2,3.5) {$j_2$};

\end{tikzpicture}
    \caption{Separating $\Prophecies{}^1$, $\Prophecies{}^2$, $\Prophecies{}^3$ w.r.t. $r$ and $q$.}
    \label{subfig:prophecy-separating}
    \end{subfigure}%
    \hfil
    \begin{subfigure}{0.45\textwidth}
    \centering
        \begin{tikzpicture}[
scale=1,
lineshiftup/.style={yshift=0.6pt},
lineshiftdown/.style={yshift=-0.6pt},
lineshiftdowntwice/.style={yshift=-1.8pt},
decoratenodebrace/.style={decorate,decoration={brace,amplitude=4pt,mirror},yshift=0pt, font=\scriptsize},
font= \footnotesize]

	\useasboundingbox (-0.2,-1.25) rectangle (4.8,1);

\begin{scope}[yshift=-0.3cm]
    \draw [thick,black,-latex] (0,0) -- (4,0);
    \node[fill=white] (x) at (4.5,0) {$\prophecy(r)$};
    \foreach \x in {0,1,2,...,3}  {
        \node[tick] (\x) at (\x,-0.25) {$\x$}; 
        \draw (\x,-0.05)--(\x,+0.05);
    }
    \foreach \x in {1,3}  {
        \draw[dotted,thick] (\x,0.12)--(\x,+0.9);
    }
\end{scope}

\begin{scope}[yshift=0.3cm]
   \node[draw=none, font=\scriptsize,cloc1] at (4.5,0) {$\Prophecies{}^1$};
    \draw [Bracket-,very thick,cloc1] (0,0) -- (1,0);
    \draw [-Bracket,very thick,cloc1,] (1,0) -- (3,0);
\end{scope}
\begin{scope}%
    \node[draw=none, font=\scriptsize,cloc2] at (4.5,0) {$\Prophecies{}^2$};
    \draw [Bracket-,very thick,cloc2,yshift=0mm, ,] (1,0) -- (3,0);
    \draw [very thick,cloc2,yshift=0mm] (3,0) -- (4,0);
\end{scope}

\draw [decoratenodebrace] (0.05,-0.7) -- node [below, yshift=-0.1cm, pos=.5,font=\footnotesize] {$j_1$} (0.95,-0.7);

\draw [decoratenodebrace] (1.05,-0.7) -- node [below, yshift=-0.1cm, pos=.5,font=\footnotesize] {$j_1,j_2$} (2.95,-0.7);

\begin{scope}
\clip (3,-2) rectangle(4,2);
\draw [decoratenodebrace] (3.05,-0.7) to  node [below, yshift=-0.1cm, pos=.5,font=\footnotesize] {$j_2$} (4.2,-0.7);
\end{scope}

\end{tikzpicture}
\begin{tikzpicture}[
scale=1,
lineshiftup/.style={yshift=0.6pt},
lineshiftdown/.style={yshift=-0.6pt},
lineshiftdowntwice/.style={yshift=-1.8pt},
decoratenodebrace/.style={decorate,decoration={brace,amplitude=4pt,mirror},yshift=0pt, font=\scriptsize},
font= \footnotesize]

	\useasboundingbox (-0.2,-1.25) rectangle (4.8,1);

\begin{scope}[yshift=-0.3cm]
    \draw [thick,black,-latex] (0,0) -- (4,0);
    \node[fill=white] (x) at (4.5,0) {$\prophecy(q)$};
    \foreach \x in {0,1,2,...,3}  {
        \node[tick] (\x) at (\x,-0.25) {$\x$}; 
        \draw (\x,-0.05)--(\x,+0.05);
    }
    \foreach \x in {1,3}  {
        \draw[dotted,thick] (\x,0.12)--(\x,0.9);
    }
\end{scope}

\begin{scope}[yshift=0.3cm]
    \node[draw=none, font=\scriptsize,cloc2] at (4.5,0) {$\Prophecies{}^2$};
    \draw [Bracket-,very thick,cloc2,yshift=0mm, ] (1,0) -- (3,0);
    \draw [very thick,cloc2] (3,0) -- (4,0);
\end{scope}
\begin{scope}
   \node[draw=none, font=\scriptsize,cloc3] at (4.5,0) {$\Prophecies{}^3$};
    \draw [Bracket-,very thick,cloc3] (0,0) -- (1,0);
    \draw [-Bracket,very thick,cloc3,] (1,0) -- (3,0);
\end{scope}

\draw [decoratenodebrace] (0.05,-0.7) -- node [below, yshift=-0.1cm, pos=.5,font=\footnotesize] {$j_3$} (0.95,-0.7);

\draw [decoratenodebrace] (1.05,-0.7) -- node [below, yshift=-0.1cm, pos=.5,font=\footnotesize] {$j_2,j_3$} (2.95,-0.7);

\begin{scope}
\clip (3,-2) rectangle(4,2);
\draw [decoratenodebrace] (3.05,-0.7) to  node [below, yshift=-0.1cm, pos=.5,font=\footnotesize] {$j_2$} (4.2,-0.7);
\end{scope}

\end{tikzpicture}
    \caption{Splitting $\Prophecies{}^1$, $\Prophecies{}^2$ and $\Prophecies{}^2$, $\Prophecies{}^3$.}
    \label{subfig:prophecy-splitting}
    \end{subfigure}
      \begin{subfigure}{0.6\textwidth}
    \centering
        \begin{tikzpicture}[
	scale=1,
	node distance = 1cm and 0.8cm,
	baseline,
	remember picture,
	n/.style={draw, text width = 0.8cm, text = black, %
		align = center, font= \footnotesize, rounded corners, very thick, execute at begin node=\setlength{\baselineskip}{8pt}%
	},
	tn/.style={draw, text width = 0.2cm, text = black, %
		align = center, font= \footnotesize, circle, very thick, execute at begin node=\setlength{\baselineskip}{8pt}%
	},
	nd/.style={draw=none, text width = 2.5cm, anchor=left, text = black, %
		align = left, font= \footnotesize, rounded corners, very thick, execute at begin node=\setlength{\baselineskip}{8pt}%
	},
	en/.style={draw=none, minimum height=0cm, font = \scriptsize, align = center,outer sep=1mm},%
	ref/.style={draw, circle, minimum width=1.5mm, inner sep=0, fill,samplevaluescolor},%
	k/.style={draw=none, circle, minimum width=0, inner sep=0,outer sep=0, fill,samplevaluescolor},%
	c/.style={draw, fill, black, circle, inner sep=0, outer sep=0, minimum size=1mm},
	l/.style={anchor=west, inner sep=0, font=\footnotesize},
        secondlayerdistance/.style={node distance =0.8cm and 0.6cm},
        firstlayerdistance/.style={node distance =1.4cm and  1.35cm}
	]

    \begin{scope}[xshift=6cm]
        
	\node[n, draw=cloc0](l0) at (0,0) { $i$ };
	\node[n, draw=cloc0,firstlayerdistance,below left = of l0.south](l0r) at (0,0) { $i_r$ };
	\node[n, draw=cloc0,firstlayerdistance,below right = of l0.south](l0q) at (0,0){ $i_q$ };
	\node[n, draw=cloc1, secondlayerdistance,below left = of l0r.south]  (l1r) {${j_1}_r$}; 
	\node[n, draw=cloc2, secondlayerdistance, below = of l0r.south,yshift=-2mm] (l2r) {${j_2}_r$}; 
	\node[n, draw=cloc3, secondlayerdistance, below right= of l0r.south] (l3r) {${j_3}_r$}; 
	\node[n, draw=cloc1, secondlayerdistance,below left = of l0q.south]  (l1q) {${j_1}_q$}; 
	\node[n, draw=cloc2, secondlayerdistance, below = of l0q.south,yshift=-2mm] (l2q) {${j_2}_q$}; 
	\node[n, draw=cloc3, secondlayerdistance, below right= of l0q.south] (l3q) {${j_3}_q$};

    \draw[-latex,very thick] (l0) to (l0r);
    \draw[-latex,very thick] (l0) to (l0q);

    \draw[-latex, very thick, dashed] (l0) to
    node[near start,en, left,stochConflict,yshift=+0mm] {$\prophecy(r) \leq \prophecy(q)$}  
	(l0r);
    
    \draw[-latex, very thick, dashed] (l0) to
    node[near start,en, right,stochConflict,yshift=+0mm] {$\prophecy(q) \leq \prophecy(r)$}  
	(l0q);

    \draw[-latex, very thick] (l0r) to
    node[near start,en, left,nonstochConflict,yshift=-1mm] {$e_1$}  
	(l1r);

\draw[-latex, very thick] (l0r) to
    node[near start,en, right,nonstochConflict,yshift=-2mm,xshift=-1mm] {$e_2$}  
	(l2r);
    
\draw[-latex, very thick, color=black] (l0r) to
    node[near start,en, right,nonstochConflict,yshift=-1mm] {$e_3$}  
	(l3r);

    \draw[-latex, very thick, color=black] (l0q) to
    node[near start,en, left,nonstochConflict,yshift=-1mm] {$e_1$}  
	(l1q);

\draw[-latex, very thick] (l0q) to
    node[near start,en, right,nonstochConflict,yshift=-2mm,xshift=-1mm] {$e_2$}  
	(l2q);

\draw[-latex, very thick] (l0q) to
    node[near start,en, right,nonstochConflict,yshift=-1mm] {$e_3$}  
	(l3q);

    \end{scope}

\end{tikzpicture} 
        \caption{Modified reach tree after separating.}\label{subfig:separating-tree}
    \end{subfigure}%
    \hfil
    \caption{Separating and splitting of enabling prophecies.}
 \label{fig:enabling-prophecies}
\end{figure*}
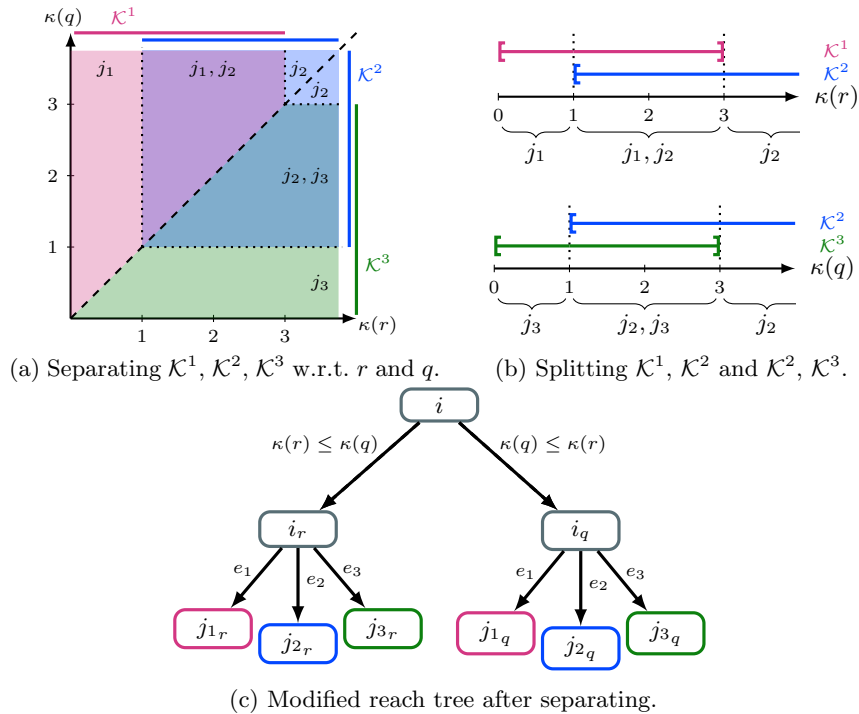%

\noindent Figure~\ref{fig:illustrate-separating} illustrates a mixed branching, where the two random events $r$ and $q$ ($\Event(e_1)=r$, $\Event(e_3)\allowbreak=q$) are enabled in the location corresponding to node $i$.
Figure~\ref{subfig:separating-tree} illustrates the modified reach tree after separation.
Specifiying the winning random event (either $r$ or $q$) results in nodes $i_r$ and $i_q$. 
A copy of the original child nodes and their subtrees is created and added to each node of the separating layer. %
In Figure~\ref{subfig:prophecy-separating}, the annotated enabling prophecies for all three children of $i$,  $\Prophecies{}^1$ (pink), $\Prophecies{}^2$ (blue) and $\Prophecies{}^3$ (green), are illustrated.
The prophecies and hence random events $r$ and $q$ are \emph{separated} along the diagonal (dashed) line.
In the top left area, $r$ is the winning event: those are all prophecies restricted to $\prophecy(r)\leq \prophecy(q)$, annotated to the edges in the subtree from $i_r$ (analogously for $q$ in the bottom right area).
\begin{definition}[Separating]
    Given a modified reach tree $\modtree=(N,E)$ 
    and a mixed branching at node $i=((id_{\reachtree},id_{\modtree}), \ell, \valSet)\in N$ with 
    competing random events $L=\{r\in \LabRandom \mid (i,e,\Prophecies{},j) \in E  \land \Event(e)=r \land j \in \validChildren(i)\}\subseteq\LabRandom$
    with $|L|> 1$.
    We \emph{separate} random events $r\in L$ by modifying sets $N$ an $E$ for each $r\in L$:
        \begin{enumerate}[(i)]
            \item add  node $i_{r}=((id_{\reachtree},id'_{\modtree}), \ell, \valSet)$ with $id_{\modtree}'=\maxid{id_{\reachtree}}$,
            \item add  edge $(i,\tau,\Prophecies{}',i_{r})$ with $\Prophecies{}'= \{\prophecy \in \Prophecies{}^i \mid \forall r' \in L\setminus \{r\}. \ \prophecy(r) \leq \prophecy(r')  \}$  with  $(i',e,\Prophecies{}^i,i)\in E$ for some  $i'\in N$, $e\in \Edge\cup\{\tau\}$ or $\Prophecies{}^i=\Prophecies{\RAC}$ if $i$ is root,
            \item for each (original) child node $j\in \validChildren(i)$ ($j\not = i_r$ for any $r\in L$) 
            add a valid copy of the subtree with root $j$,
            starting from $j_r$, and restrict to prophecies $\Prophecies{}=\{\prophecy \in \Prophecies{\RAC} \mid \forall r' \in L\setminus \{r\}. \ \prophecy(r) \leq \prophecy(r')  \} $,
            \item  add  edge $(i_{r},e^j,\Prophecies{}',j_r) $ for each node $j_r$  with $(i,e^j,\Prophecies{}^j,j)\in E$ and 
            $\Prophecies{}'= \{\prophecy\in \Prophecies{}^j\ |\ \forall r' \in L\setminus \{r\}. \ \prophecy(r) \leq \prophecy(r')  \}$, %
        \end{enumerate}
        and for all $j\in \validChildren(i)$  with $j\not = i_r$ for any $r\in L$, remove $(i,e,\Prophecies{},j)$ and $j$.
        
\end{definition}
Note that restricting the sets of prophecies annotated in a modified reach tree may result in null sets, e.g., for edges that  model the expiration of a random event that is not the designated winner in a specific subtree. 

Figure~\ref{subfig:prophecy-splitting} illustrates  the concept of splitting  enabling prophecies for nodes $i_r$ and $i_q$ from Figure~\ref{subfig:separating-tree}, where the upper and lower part of Figure~\ref{subfig:prophecy-splitting} correspond to the results of restricting the prophecies in the separation step.
In this example, both sets of prophecies $\prophecy{(r)}$  and $\prophecy{(q)}$ are  split into three sets of prophecies, such that the resulting sets enable the same set of jumps, i.e., $\{j_1\},\{j_1,j_2\},\{j_2\}$ and $\{j_3\},\{j_2,j_3\},\{j_2\}$, respectively. These sets of prophecies are then obtained by splitting prophecies along the dotted lines that are orthogonal to the respective axis in Figure~\ref{subfig:prophecy-separating}. For $\prophecy(r)$ this results, e.g., in
$\Prophecies{}^{\{j_1,j_2\}}=\{\prophecy\in \Prophecies{}^1 \cap \Prophecies{}^2 \mid \prophecy(r)\in [1,3] \land \prophecy(r)\leq \prophecy(q)\}$.

Due to the separation of $r$ and $q$, the splitting step for both $i_r$ and $i_q$ results in only three children and sets of prophecies, in contrast to a direct splitting that would result in five children, providing clarity in the modified reach tree. %

\subsection{Computational Complexity}\label{app:complexity}

Computational complexity for the approach to compute minimum reachability probabilities consists of multiple parts:
\begin{enumerate}[(i)]
    \item Reachability analysis is exponential in the jump depth $n$, resulting in the reach tree with $\mathcal{O}(k^n)$ nodes for at most $k$ jumps from any location and a tree depth of $n$.
    \item Each node in the reach tree induces a branching of at most $k$ children that potentially needs to be split. This modification introduces at most $2^k$ new children, all equipped with a copy of the subtree starting from our node (see Figure~\ref{fig:illustrate-splitting}). For each of the $n$ layers of the tree, this results in at most $k\cdot2^k$ nodes. Hence, performing this operation for every node in the original reach tree results in a fully modified reach tree of size $\mathcal{O}((k\cdot2^k) ^ n)$. Note that exponentiality w.r.t. $n$ stems from the need to copy subtrees with increasing size.
    \item To compute the enforcing prophecies, we now traverse the fully modified reach tree and perform union or intersection according to Definition~\ref{def:enforcing}. Due to the representation of non-convex polytopes as sets of convex polytopes, this is exponential in the number of children.
\end{enumerate}

\FloatBarrier
\subsection{Evaluation on Small Examples}\label{app:smallexamples}

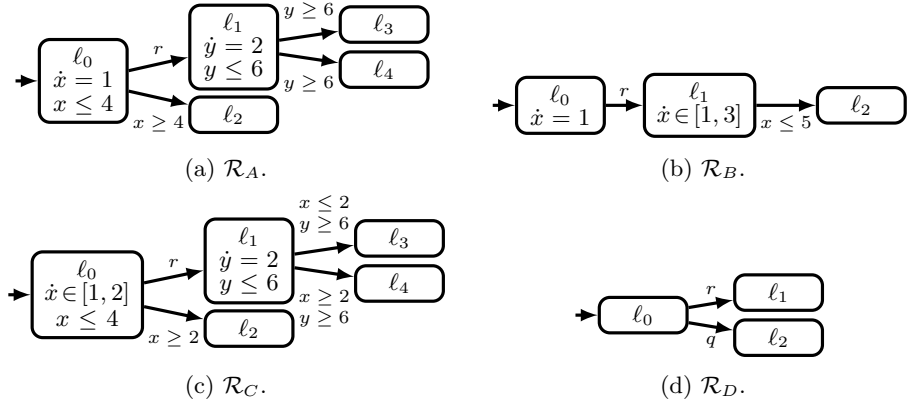
\begin{figure}[p]
    \centering
    \begin{subfigure}{0.48\linewidth}
    \centering
        \begin{tikzpicture}[
baseline,remember picture,
node distance = -0.5cm and 1.4cm,
n/.style={draw, text width = 0.9cm, %
align = center, font= \footnotesize, rounded corners, very thick, execute at begin node=\setlength{\baselineskip}{8pt}%
},
en/.style={draw=none, minimum height=0cm, font = \scriptsize, align = center},%
c/.style={draw, fill, black, circle, inner sep=0, outer sep=0, minimum size=1mm},
l/.style={anchor=west, inner sep=0, font=\footnotesize},
init/.style={en, 
text width=1.2cm,
anchor=south west,  
inner xsep=0pt},
smallnode/.style={draw, text width = 0.95cm, %
    align = center, font= \footnotesize, rounded corners, very thick, execute at begin node=\setlength{\baselineskip}{8pt}%
    }
]

\node[smallnode](l0) at (-1.8,0) {
    $\ell_0$\\\flowsepspace 
    $\dot{x}=1$\\\flowsepspace 
    $x\leq 4$
};

\node[smallnode](l1) at (0.2,0.45) {
    $\ell_1$\\\flowsepspace 
    $\dot{y}=2$\\\flowsepspace     
    $y\leq 6$
};

\node[smallnode](l2) at (0.2,-0.45) {$\ell_2$};

\node[smallnode](l3) at (2.2,.75) {
    $\ell_3$
};

\node[smallnode](l4) at (2.2,0.15) {
    $\ell_4$
};

\coordinate (init) at (-2.7,0);

\draw[-latex, very thick] (init) to node[en, above] {} (l0);
\draw[-latex, very thick] (l0) to node[en, above] {$r$} (l1);
\draw[-latex, very thick] (l0) to node[en, below, yshift=-1mm] {$x \geq 4$} (l2);
\draw[-latex, very thick] (l1) to node[en, above, yshift=1mm] {$y \geq 6$} (l3);
\draw[-latex, very thick] (l1) to node[en, below, yshift=-1mm] {$y \geq 6$} (l4);

\end{tikzpicture}
        \caption{$\RAC_A$.}
        \label{subfig:caseA}
    \end{subfigure}%
    \hfill    
    \begin{subfigure}{0.48\linewidth}
    \centering
        \begin{tikzpicture}[
baseline,remember picture,
node distance = -0.5cm and 1.4cm,
n/.style={draw, text width = 0.9cm, %
align = center, font= \footnotesize, rounded corners, very thick, execute at begin node=\setlength{\baselineskip}{8pt}%
},
en/.style={draw=none, minimum height=0cm, font = \scriptsize, align = center},%
c/.style={draw, fill, black, circle, inner sep=0, outer sep=0, minimum size=1mm},
l/.style={anchor=west, inner sep=0, font=\footnotesize},
init/.style={en, 
text width=1.2cm,
anchor=south west,  
inner xsep=0pt},
smallnode/.style={draw, text width = 0.95cm, %
    align = center, font= \footnotesize, rounded corners, very thick, execute at begin node=\setlength{\baselineskip}{8pt}%
    }
]

\node[smallnode](l0) at (-0.05,0) {
    $\ell_0$\\\flowsepspace 
    $\dot{x}=1$
};

\node[smallnode,text width=1.25cm](l1) at (1.8,0) {
    $\ell_1$\\\flowsepspace 
    $\dot{x}\!\in\![1,3]$
};

\node[smallnode](l2) at (3.95,0) {$\ell_2$};

\coordinate (init) at (-0.95,0);

\draw[-latex, very thick] (init) to node[en, above] {} (l0);
\draw[-latex, very thick] (l0) to node[en, above] {$r$} (l1);
\draw[-latex, very thick] (l1) to node[en, below] {$x\leq 5$} (l2);

\end{tikzpicture}
        \caption{$\RAC_B$.}
        \label{subfig:caseC}
    \end{subfigure}%
    
    \begin{subfigure}{0.48\linewidth}
    \centering
        \begin{tikzpicture}[
baseline,remember picture,
node distance = -0.5cm and 1.4cm,
n/.style={draw, text width = 0.9cm, %
align = center, font= \footnotesize, rounded corners, very thick, execute at begin node=\setlength{\baselineskip}{8pt}%
},
en/.style={draw=none, minimum height=0cm, font = \scriptsize, align = center},%
c/.style={draw, fill, black, circle, inner sep=0, outer sep=0, minimum size=1mm},
l/.style={anchor=west, inner sep=0, font=\footnotesize},
init/.style={en, 
text width=1.2cm,
anchor=south west,  
inner xsep=0pt},
smallnode/.style={draw, text width = 0.95cm, %
    align = center, font= \footnotesize, rounded corners, very thick, execute at begin node=\setlength{\baselineskip}{8pt}%
    }
]

\node[smallnode,text width = 1.25cm,](l0) at (-1.95,0) {
    $\ell_0$\\\flowsepspace 
    $\dot{x}\!\in\![1,2]$\\\flowsepspace 
    $x\leq 4$
};

\node[smallnode](l1) at (0.2,0.45) {
    $\ell_1$\\\flowsepspace 
    $\dot{y}=2$\\\flowsepspace     
    $y\leq 6$
};

\node[smallnode](l2) at (0.2,-0.45) {
    $\ell_2$
};

\node[smallnode](l3) at (2.2,0.75) {
    $\ell_3$
};

\node[smallnode](l4) at (2.2,0.15) {
    $\ell_4$
};

\coordinate (init) at (-3,0);

\draw[-latex, very thick] (init) to node[en, above] {} (l0);
\draw[-latex, very thick] (l0) to node[en, above] {$r$} (l1);
\draw[-latex, very thick] (l0) to node[en, below, yshift=-1mm] {$x \geq 2$} (l2);
\draw[-latex, very thick] (l1) to node[en, above, yshift=1mm] {$x\leq2$\\$y \geq 6$} (l3);
\draw[-latex, very thick] (l1) to node[en, below, yshift=-1mm] {$x\geq 2$\\$y \geq 6$} (l4);

\end{tikzpicture}
        \caption{$\RAC_C$.}
        \label{subfig:caseB}
    \end{subfigure}%
    \hfill
    \begin{subfigure}{0.48\linewidth}
    \centering
        \begin{tikzpicture}[
baseline,remember picture,
node distance = -0.5cm and 1.4cm,
n/.style={draw, text width = 1.2cm, %
align = center, font= \footnotesize, rounded corners, very thick, execute at begin node=\setlength{\baselineskip}{8pt}%
},
en/.style={draw=none, minimum height=0cm, font = \scriptsize, align = center},%
c/.style={draw, fill, black, circle, inner sep=0, outer sep=0, minimum size=1mm},
l/.style={anchor=west, inner sep=0, font=\footnotesize},
init/.style={en, 
text width=1.2cm,
anchor=south west,  
inner xsep=0pt},
smallnode/.style={draw, text width = 0.95cm, %
    align = center, font= \footnotesize, rounded corners, very thick, execute at begin node=\setlength{\baselineskip}{8pt}%
    }
]

\node[smallnode](l0) at (-1.8,0) {
    $\ell_0$
};

\node[smallnode](l1) at (0,0.3) {
    $\ell_1$
};

\node[smallnode](l2) at (0,-0.3) {
    $\ell_2$
};

\coordinate (init) at (-2.7,0);

\draw[-latex, very thick] (init) to node[en, above] {} (l0);
\draw[-latex, very thick] (l0) to node[en, above] {$r$} (l1);
\draw[-latex, very thick] (l0) to node[en, below] {$q$} (l2);

\end{tikzpicture}
        \caption{$\RAC_D$.}
        \label{subfig:caseD}
    \end{subfigure}%
    \caption{Reduced illustrations for RAC $\RAC_A$, $\RAC_B$, $\RAC_C$ and $\RAC_D$, which omit $\Init(\ell_0)=0^d$, and the flow of $x\in\VarCont$ if $\Flow(\ell)\proj{x}=0$. %
    }
 \label{fig:example-automata}
\end{figure}%

We provide minimum and maximum reachability probabilities for the  RAC in Figure~\ref{fig:example-automata} with different branchings:
(i) $\RAC_A$ has a stochastic and a nondeterministic branching, 
(ii) $\RAC_B$ has a stochastic and a mixed branching (one can stay forever in $\ell_1$),
(iii) $\RAC_C$ has a mixed and a hidden mixed branching (the choice in $\ell_1$ depends on $x$ and hence on $\prophecy(r)$), 
and (iv) $\RAC_D$ has a stochastic branching.

We present results for different goal locations, with random event $r$ (and $q$ in $\RAC_D$) either (i) uniformly or (ii) normally distributed. %
Table~\ref{tab:results} provides minimum reachability probabilities, their computation time $\comptime$ and the statistical error $\estat$ resulting from integration for both distributions. %
For comparison, we provide maximum reachability probabilities (c.f.~\cite{journal}) %
as well as the sets of  prophecies for both the minimum and the maximum scheduler ($\PropheciesSchedulerMin{}$ and $\PropheciesSchedulerMax{}$).
To yield maximum reachability probabilities, the set of prophecies $\PropheciesSchedulerMax{}$ allowing a maximal scheduler to reach the goal can also be computed via a fully modified reach tree by taking the union for stochastic \emph{and} nondeterministic branchings.
For a thorough explanation of the approach that computes maximum reachability probabilities (as implemented in \realyst) and performs scheduler synthesis via forward and backward analysis, we refer to~\cite{journal}.

Integration is invoked with \samples{1e5} samples for  $\RAC_A, \RAC_B$ and $\RAC_C$ and \samples{1e6} for $\RAC_D$, and results are computed using an AMD Ryzen $9$ $5900$X with 12$\times$\SI{3.70}{\giga\hertz} and \SI{32}{\giga\byte} RAM.

Due to the size of the examples, results for the uniform distribution can  be validated by hand. 
For comparison, we provide these manual computations below. %
Clearly, for more complex models this is not feasible. %

\begin{table*}[p]
    \centering
    \scriptsize
    \caption{
      Reachability probabilities $p_{\RAC}^\textsl{min}$ and $p_{\RAC}^\textsl{max}$ for RAC $\RAC\in \{\RAC_A,\RAC_B,\RAC_C,\RAC_D\}$,  goal locations $\Goal$, $\tmax=10$ and $\jumpmax=2$ for (i) $r,q\sim\mathcal{U}(0,10)$  (ii) $r,q\sim\mathcal{N}(3,2)$; with enforcing prophecies 
        $\PropheciesSchedulerMin{}$ and allowing prophecies $\PropheciesSchedulerMax{}$
        with 
      *: $\prophecy(r)\leq \prophecy(q)$,  **: $\prophecy(r)\geq \prophecy(q)$. %
      }
    \label{tab:results}
    \newcolumntype{Y}{>{\centering\arraybackslash}X}
    \renewcommand{\arraystretch}{1}
    \begin{tabularx}{\linewidth}{lcYYYYYYYY}
        \toprule
          \  &\  &    \multicolumn{2}{c}{$\RAC_A$} & $\RAC_B$ &  \multicolumn{3}{c}{$\RAC_C$}   & \multicolumn{2}{c}{$\RAC_D$}  \\
      \cmidrule(lr){3-4}\cmidrule(lr){5-5}\cmidrule(lr){6-8}\cmidrule(lr){9-10}
          \multicolumn{2}{c}{$\Goal$}  
          & $\ell_2$ & $\ell_4$  
          & $\ell_2$ 
          & $\ell_2$ & $\ell_3$& $\ell_4$  
          &  $\ell_1$ & $\ell_2$  \\
          \midrule
          \multirow{4}{*}{(i)} 
            & $p_{\RAC}^\textsl{min}$
            & \probabminuni{0.5999999999999467} & \probabminuni{0.0} 
            & \probabminuni{0.0}  
            & \probabminuni{0.5999999999999467} & \probabminuni{0.09999999999999393}& 
            \probabminuni{0.0} 
            & \probabminuni{0.5000309434004911}& \probabminuni{0.4999690565995044} 
            \\ 
            & $\comptime$
             & \rt{0.037835}& \rt{0.0331805}  
            & \rt{0.0186548}
            & \rt{0.0464783} & \rt{0.0465056}& \rt{0.0409716} 
            & \rt{0.0800123} & \rt{0.0721609}
            \\ 
            & $\estat$
            & \errortable{0}  &  -- 
            & --
            & \errortable{0} & \errortable{0} & --
            & \errortable{3.432005731600063e-05} & \errortable{3.432005731600063e-05}
            \\\cmidrule{2-10} 
            & $p_{\RAC}^\textsl{max}$
            & \probabminuni{0.5999999999999467}  
            & \probabminuni{0.3999999999999757} 
            & \probabminuni{0.4999999999999531} 
            & \probabminuni{0.9000000000001088} & \probabminuni{0.1999999999999879} & \probabminuni{0.2999999999999733} 
            & \probabminuni{0.5000309434004911}& \probabminuni{0.4999690565995044} 
            \\ 
            & $\comptime$
             & \rt{0.0315244}& \rt{0.0342866}  
            & \rt{0.0186339}
            & \rt{0.0353247} & \rt{0.0350993}& \rt{0.0322195} 
            & \rt{0.0784517} & \rt{0.0670379}
            \\ 
            & $\estat$
            & \errortable{0}  &  \errortable{0} 
            & \errortable{0}
            & \errortable{0} & \errortable{0} &\errortable{0}
            & \errortable{3.432005731600063e-05} & \errortable{3.432005731600063e-05}
            \\
            \midrule
          \multirow{4}{*}{(ii)} 
            & $p_{\RAC}^\textsl{min}$
            & \probabmin{0.3087702142150796} & \probabmin{0.0}  
            & \probabmin{0.0} 
            & \probabmin{0.3087702142150796} & \probabmin{0.135905120535349} & \probabmin{0.0} 
            & \probabmin{0.4999950203148315}& \probabmin{0.4999930890010845}  
            \\
            & $\comptime$
            & \rt{0.0439504} &  \rt{0.0337623} 
            & \rt{0.0179435}
            & \rt{0.0475123} &  \rt{0.0506088} & \rt{0.0432539}
            & \rt{0.094834} & \rt{0.0924368}
            \\ 
            & $\estat$
            & \errortable{1.266747991869975e-07} & -- 
            & --
            & \errortable{1.266747991869975e-07} & \errortable{1.527259660067558e-09} & --
            & \errortable{3.78524156188754e-05} & \errortable{3.532816695348094e-05}
            \\\cmidrule{2-10} 
            & $p_{\RAC}^\textsl{max}$
            & \probabmin{0.3087702142150796} & \probabmin{0.6912297835956611}  
            & \probabmin{0.8413129610391057} 
            & \probabmin{0.8640650268222011} & \probabmin{0.3023278646503927} & \probabmin{0.5553246808561304} 
            & \probabmin{0.4999950203148315}& \probabmin{0.4999930890010845}  
            \\
            & $\comptime$
            & \rt{0.0370975} &  \rt{0.0341294} 
            & \rt{0.0176783}
            & \rt{0.0418707} &  \rt{0.0328491} & \rt{0.0331428}
            & \rt{0.0942871} & \rt{0.0874447}
            \\ 
            & $\estat$
            & \errortable{1.715499066358638e-07} & \errortable{1.813785931987538e-08} 
            & \errortable{6.434862909296645e-08}
            & \errortable{5.004011661684443e-05} & \errortable{9.454021339067047e-09} & \errortable{1.389009644379703e-09}
            & \errortable{3.78524156188754e-05} & \errortable{3.532816695348094e-05}
            \\
        \midrule
            \multicolumn{2}{c}{$\PropheciesSchedulerMin{} $}
            & $\geq 4$ & $\emptyset$
            & $\emptyset$
            & $\geq 4$ & $\leq 1$& $\emptyset$ 
            & * & **
            \\ 
            \multicolumn{2}{c}{$\PropheciesSchedulerMax{} $}
            & $\geq 4$   & $\leq 4$ 
            & $\leq 5$
            & $\geq 1$ & $\leq 2$ & $[1,4]$
            & * & **
            \\ 
          \bottomrule
    \end{tabularx}
\end{table*}

Table~\ref{tab:results}  shows that our implementation is fast and highly accurate. 
As expected, the computed minimum reachability probability is always smaller or equal to the maximum, and $\PropheciesSchedulerMin{}\subseteq\PropheciesSchedulerMax{}$. %
For $\RAC_A$, we see that the minimum and maximum probabilities to reach $\ell_2$ are equal, as this choice  is purely stochastic. In contrast,
the minimum reachability probability for $\ell_4$ is zero, as a scheduler can avoid that location.
In $\RAC_B$, $\ell_2$ can be avoided by staying  in $\ell_1$ forever, 
resulting in $p_{\RAC_B}^\textsl{min}=0$ for both distributions. 
The maximum probability for (ii) is larger, as more probability mass is attributed to $\PropheciesSchedulerMax{}$.
In $\RAC_C$, 
considering a prophecy $\prophecy(r)\in[1,4]$, a prophetic scheduler 
in $\ell_0$ can either 
choose a small rate for $x$ to stay until $\valCont_r=\prophecy(r)$, and then leave to $\ell_1$, or 
choose a large rate to reach $\valCont_x\geq2$ quickly to leave to $\ell_2$. 
While minimum and maximum reachability probabilities for $\ell_2$ are both larger for (i), where $r$ is uniformly distributed, the range between them is larger for (ii), where $r$ is normally distributed.
This is due to the mean of the normal distribution not being equal to the expected value of the uniform distribution. 
Also, a prophetic scheduler can avoid $\ell_4$ by moving from $\ell_0$ to $\ell_2$, resulting in a minimum probablity of zero.
In $\RAC_D$, the choice is purely stochastic, hence minimum and maximum reachability probabilities coincide. 
As $r$ and $q$ follow the same distribution, $\ell_1$ and $\ell_2$ can be reached with equal  probability.
For only one uniformly distributed random event, the statistical error is always zero, while $\RAC_D$ required a larger number of integration samples. %

\subsubsection{Manual Validation of Results}\label{app:validation}
To validate the obtained results of Table~\ref{tab:results}, we can compute the set of prophecies $\PropheciesSchedulerMin{}$ enforcing the scheduler to reach $\Goal$ by hand for each of the considered cases.
Given $\PropheciesSchedulerMin{}$ and assuming  $r \sim \mathcal{U}(0,10)$ (and $q \sim \mathcal{U}(0,10)$ for $\RAC_D$), the minimum reachability probability $p_{\RAC_A}^\textsl{min}(\Goal, 10,2)$ can then be computed via integration:
\begin{align*}
    p_{\RAC_A}^\textsl{min}(\Goal, 10,2)=  
    \int_{\PropheciesSchedulerMin{}} G(\prophecy) \ d\prophecy.
\end{align*}

\noindent
For RAC $\RAC_A$, we consider $\ell_2$ and $\ell_4$: %
\begin{itemize}%
    \item $\Goal=\{\ell_2\}$:
    $\PropheciesSchedulerMin{} 
        = \{\prophecy\in \Prophecies{\RAC_A} \mid \prophecy(r)\in [4,10] \} = [4,10]$
    \item $\Goal=\{\ell_4\}$:
     $\PropheciesSchedulerMin{} 
        =\emptyset \cup (\emptyset \cap \{\prophecy\in \Prophecies{\RAC_A} \mid \prophecy(r)\in [0,4]   \}) = \emptyset$
\end{itemize}

\noindent
For RAC $\RAC_B$, we consider $\ell_2$: %
\begin{itemize}%
    \item $\Goal=\{\ell_2\}$: $\PropheciesSchedulerMin{} 
        = (\{\prophecy\in \Prophecies{\RAC_B} \mid \prophecy(r)\in [0,5] \}) \cap \emptyset) \cup \emptyset   = \emptyset$
\end{itemize}

\noindent
For RAC $\RAC_C$, we consider $\ell_2$, $\ell_3$ and $\ell_4$: %
\begin{itemize}%
    \item $\Goal=\{\ell_2\}$:
    $\PropheciesSchedulerMin{} 
          =
        \{\prophecy\in \Prophecies{\RAC_C} \mid \prophecy(r)\in [4,\infty)\}\\
        \phantom{\Goal=\{\ell_2\}: \PropheciesSchedulerMin{}=}
        \cup 
        (
        \{\prophecy\in \Prophecies{\RAC_C} \mid \prophecy(r)\in [1,4]\}
        \cap
        \emptyset
        )
        \cup \emptyset\\
        \phantom{\Goal=\{\ell_2\}: \PropheciesSchedulerMin{}}
        =[4,\infty)
        $
    \item $\Goal=\{\ell_3\}$:
     $\PropheciesSchedulerMin{} 
        =\emptyset \cup 
        (
        (\{\prophecy\in \Prophecies{\RAC_C} \mid \prophecy(r)\in [1,2]\}\cap \emptyset) \cup \emptyset)
         \\
        \phantom{\Goal=\{\ell_3\}: \PropheciesSchedulerMin{}=} 
        \cup 
        \{\prophecy\in \Prophecies{\RAC_C} \mid \prophecy(r)\in [0,1]\} 
        =[0,1]$
    \item $\Goal=\{\ell_4\}$:
     $\PropheciesSchedulerMin{} 
        =  \emptyset \cap
        (\emptyset \cup 
        ( \emptyset \cap 
        \{\prophecy\in \Prophecies{\RAC_C} \mid \prophecy(r)\in [1,2]\} )\\
        \phantom{\Goal=\{\ell_4\}: \PropheciesSchedulerMin{}=}
        \cup 
        \{\prophecy\in \Prophecies{\RAC_C} \mid \prophecy(r)\in [2,4]\}))
        =\emptyset$
\end{itemize}

\noindent
For RAC $\RAC_D$, we consider $\ell_1$ and $\ell_2$: %
\begin{itemize}%
    \item $\Goal=\{\ell_1\}$:
    $\PropheciesSchedulerMin{} 
        = \{\prophecy\in \Prophecies{\RAC_D} \mid \prophecy(r)\leq \prophecy(q) \}$
    \item $\Goal=\{\ell_2\}$:
    $\PropheciesSchedulerMin{} 
        = \{\prophecy\in \Prophecies{\RAC_D} \mid \prophecy(q)\leq \prophecy(r) \}$
\end{itemize}

\subsection{\emph{CAR} Case Study}\label{app:car}
\begin{figure*}[p]
    \centering
    \begin{tikzpicture}[
n/.style={draw, text width = 1.3cm, minimum height = 1.7cm, align = center, font= \footnotesize, rounded corners, very thick, execute at begin node=\setlength{\baselineskip}{8pt}%
},
nsmall/.style={draw, text width = 2cm, inner xsep=0.1em, minimum height = 1cm, align = center, font= \small, rounded corners, very thick, execute at begin node=\setlength{\baselineskip}{8pt}%
},
nsmallslim/.style={draw, text width = 1.4cm, minimum height = 1cm, align = center, font= \footnotesize, rounded corners, very thick, execute at begin node=\setlength{\baselineskip}{8pt}%
},
charging/.style={minimum height = 1.9cm
},
rest/.style={minimum height = 1.9cm
},
en/.style={draw=none, minimum height=0cm, font = \scriptsize, align = center},%
c/.style={draw, fill, black, circle, inner sep=0, outer sep=0, minimum size=1mm},
l/.style={anchor=west, inner sep=0, font=\footnotesize},
]

\node[nsmall,charging](chA) at (1.5,0) {
	$\ell_{\texttt{A}}$	\\
	\invariantsepspacesmall 	
	$\dot{x}\in[4,6]$\\\flowsepspace
    $\dot{q}=0$ \\
	\invariantsepspacesmall 	
	$x \leq 3$
	};

\node[nsmall,charging](chB) at (4.5,0) {
	$\ell_{\texttt{B}}$\\
	\invariantsepspacesmall 	
	$\dot{x}\in[3,4]$\\\flowsepspace
    $\dot{q}=0$ \\
	\invariantsepspacesmall 	
	$x \leq 9$	};

\node[nsmall,charging](chC) at (7.5,0) {
	$\ell_{\texttt{C}}$	\\
	\invariantsepspacesmall 	
	$\dot{x}\in[1,3]$\\\flowsepspace
    $\dot{q}=0$ \\
	\invariantsepspacesmall 	
	$x \leq 10$	};

\node[nsmallslim,charging](noch) at (10.5,0) {
	$\ell_{\texttt{full}}$	\\
	\invariantsepspacesmall 	
	$\dot{x}=0$\\\flowsepspace
    $\dot{q}=0$ \\
	\invariantsepspacesmall 	
	$x=10$	};

\node[nsmall,rest](dr) at ($(chA)-(0,2.75)$) {
	$\ell_{\texttt{driving}}$\\
	\invariantsepspacesmall 	
	$\dot{x}\in [-3,-2]$\\\flowsepspace
    $\dot{q}=0$ \\	};

\node[nsmall,rest](arr) at ($(dr)-(0,2.75)$) {
	$\phantom{g}\ell_{\texttt{arrival}}\phantom{g}$\\
	\invariantsepspacesmall
    $\dot{x}=0$ \\ 	\flowsepspace
    $\dot{q}=0$ \\ 		};
	
\node[nsmall,rest](em) at  ($(chB)-(0,2.75)$)  {
	$\phantom{g}\ell_{\texttt{empty}}\phantom{g}$\\
	\invariantsepspacesmall
    $\dot{x}=0$ \\\flowsepspace
    $\dot{q}=0$ \\ 		
	};

\node[nsmall,rest](re) at ($(chC)-(0,2.75)$) {
	$\phantom{g}\ell_{\texttt{detour}}\phantom{g}$\\
	\invariantsepspacesmall 	
	$\dot{x}\in [-3,-2]$\\\flowsepspace
    $\dot{q}=0$ \\	};

\node[nsmallslim,rest](new) at ($(noch)+(0,-2.75)$) { %
	$\ell_{\texttt{charge}}$\\
 	\invariantsepspacesmall
    $\dot{x}=0$ \\\flowsepspace
    $\dot{q}=1$ \\	
	\invariantsepspacesmall 	
    $q\leq0$		
	};

\node[en, text width=1.2cm,execute at begin node=\setlength{\baselineskip}{8pt}, anchor=north] (init) at ($(chA.110)+(0,1.5)$) {$x\in${$\,[1,2]$}\\
$q=0$\\
 };
\draw[-latex, very thick] ($(chA.110)+(0,0.7)$) to  ($(chA.110)+(0,0)$);

\draw[-latex,  very thick] (chA) to node[en, above] {$x \geq 2$} (chB);
\draw[-latex, very thick] (chB) to node[en, above] {$x\geq 8$} (chC);
\draw[-latex, very thick] (chC) to node[en, above] {$x=10$} (noch);

\coordinate  (corner) at ($(chA.south)-(0,0.4)$);
\draw[-latex,  very thick] (chB.south) 
-- (chB.south |- corner) 
-- (corner) 
to node[en, left] {}
(dr.north);
\draw[-latex,  very thick] (chC.south) 
-- (chC.south |- corner) 
-- (corner) 
to node[en, left] {}
(dr.north);
\draw[-latex,  very thick] (noch.south) 
-- (noch.south |- corner) 
-- (corner) 
to node[en, left] {}
(dr.north);

\draw[-latex,  very thick] (chA) to node[en, left] {$c$} (dr);

\draw[-latex,  very thick] (dr.-100) to node[en, left] {$d$} (arr.100);
\draw[-latex,  very thick] ($(dr.south)+(0,0)$)  --  ($(dr.south)+(0,-0.3)$) to  node[en, above, near start, ] {$r$}($(re.south)+(0,-0.3)$) -- (re);
\draw[-latex,  very thick,] (dr) to node[en, above] {$x=0$} (em);

\draw[-latex,  very thick,] (re) to node[en, above] {$x=0$} (em);

\draw[-latex, very thick] (re)  to node[en, above] {   $d $}  ($(new.west)+(0,0)$);%

\coordinate  (corner) at ($(noch.north east)+(0.2,0.7)$);

\draw[-latex,  very thick] (new.east) 
-- (new.east -| corner)  
-- (corner) 
-- (corner -| chA.north) 
to node[en, right] {$x\in[0,3]$}
(chA.north);

\draw[-latex,  very thick] (new.east) 
-- (new.east -| corner)  
-- (corner) 
-- (corner -| chB.north) 
to node[en, left] {$x\in[3,9]$} (chB.north);

\draw[-latex,  very thick] (new.east) 
-- (new.east -| corner)  
-- (corner) 
-- (corner -| chC.north) 
to node[en, left] {$x\in[9,10]$} (chC.north);

\draw[-latex,  very thick] (new.east) 
-- (new.east -| corner)  
-- (corner) 
-- (corner -| noch.north) 
to node[en, left] {$x=10$} (noch.north);

\end{tikzpicture}
    \caption{RAC for the $\car$ case study.}
    \label{fig:car}
\end{figure*}
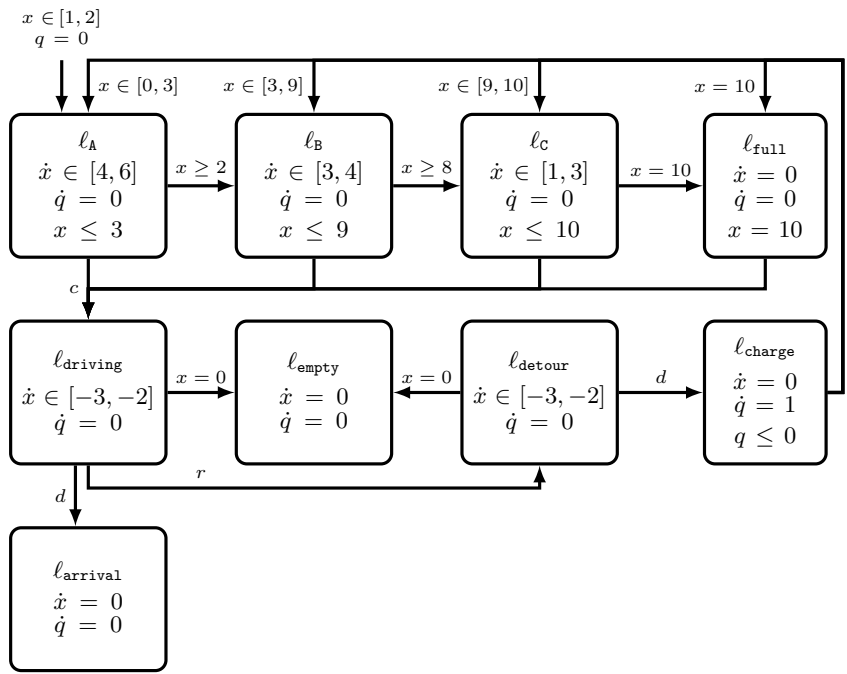
Figure~\ref{fig:car} illustrates the \emph{CAR} case study.
    Random event $c$ can occur in the charging locations $\ell_{\texttt{A}}$, $\ell_{\texttt{B}}$, $\ell_{\texttt{C}}$ and in $\ell_{\texttt{full}}$, $d$ can happen in locations $\ell_{\texttt{driving}}$ and $\ell_{\texttt{detour}}$ and $r$ is possible in location $\ell_{\texttt{driving}}$. 
The state of charge $x$ is restricted to $[0,10]$ in all locations unless stated otherwise. No time is spent in $\ell_{\texttt{charge}}$ due to the invariant $q\leq 0$. 
Note that to yield $\car^{0}$ and $\car^{1}$, the illustrated RAC is \emph{unrolled} w.r.t. the specific number of detours, with an absorbing goal location $\ell_{\texttt{empty}}$.
Hence, the number of locations for $\car^{0}$ and $\car^{1}$ are as follows:
\begin{itemize}
    \item For $\car^{0}$, $\lvert\Loc\rvert$ is, depending on the charging types, \texttt{A}: 5, \texttt{AB}: 6,  \texttt{ABC}: 7, and
    \item for $\car^{1}$, $\lvert\Loc\rvert$ is, depending on the charging types, \texttt{A}: 11, \texttt{AB}: 13,  \texttt{ABC}: 15.
\end{itemize}

Figure~\ref{fig:new_plots} compares maximum and minimum computation times for the model variants considered in Table~\ref{tab:car-esults}~(i). %
We illustrate the computation times separately for (1) forward reachability analysis, (2) classification of branchings (only minimum), (3) collection of prophecies and (4) integration.
The collection of prophecies for $p_{\RAC}^\textsl{min}$ follows Definition~\ref{def:enforcing}, whereas prophecies can be collected by always taking the union for $p_{\RAC}^\textsl{max}$, which is computationally more efficient.

In Figure~\ref{fig:new_plots}, we see that the minimum approach takes more computation time w.r.t. the maximum computation. 
With more complex model variants, this difference becomes even larger.
Specifically, the computation time for integration increases w.r.t. the computation of the maximum reachability probability. Further, the collection of prophecies takes much more time when computing the minimum reachability probability and also increases for larger model variants.  The same effects show 
for the larger model with one detour. However, note that  $\comptime$ increases with a factor $\sim200$ for both, maximum and minimum computations. 
Clearly, the size of the reach tree increases significantly with  additional charging locations and detours, which is reflected by $\vert \reachtree \vert$ in  Table~\ref{tab:car-esults}. %

Figure~\ref{fig:new_plots} also illustrates the impact of additional mixed branchings in the minimum setting, as stated in Section~\ref{sec:casestudy}.
E.g., the larger number of polytopes resulting from splitting and hence the representation of non-convex polytopes as a union of convex polytopes reflects in a larger integration time (purple in Figure~\ref{fig:new_plots}).

\end{appendix}

\end{document}